\documentclass[a4paper,11pt]{article}
\usepackage[left=1.8cm,top=1.8cm, bottom=1.8cm, right=1.8cm]{geometry}
\usepackage{graphicx,subfigure}
\usepackage{soul,color,tikz,framed}
\definecolor{shadecolor}{gray}{.85}
\setstcolor{black}
 \sethlcolor{yellow}
\usepackage{amsmath}
\usepackage{amssymb}
\usepackage{amsfonts}
\usepackage{amsopn,amsthm}
\usepackage{bbm}
\graphicspath{{img/}}
\newtheorem{theorem}{\bf Theorem}
\newtheorem{lemma}{\bf Lemma}

\newtheorem{corollary}{\bf Corollary}
\newtheorem{definition}{\bf Definition}

\newtheorem{remark}{\bf Remark}

\def\hY{\hat{Y}}
\def\hy{\hat{y}}

\def\hu{\hat{u}}

\def\hS{\hat{S}}
\def\hx{\hat{x}}
\def\hs{\hat{s}}

\def\tR{\tilde{R}}
\def\bX{\mathbf{X}}

\def\styp{\mathcal{A}_{\epsilon}^n}
\def\typ{\mathcal{T}_{\epsilon}^n}

\def\n{\nonumber\\}
\def\be{\begin{equation}}
\def\ee{\end{equation}}
\def\bes{\begin{equation*}}
\def\ees{\end{equation*}}
\def\beq{\begin{eqnarray}}
\def\eeq{\end{eqnarray}}
\def\beqs{\begin{eqnarray*}}
\def\eeqs{\end{eqnarray*}}

\def\ma{{\mathcal A}}
\def\mb{{\mathcal B}}
\def\cc{{\mathcal C}}

\def\mm{{\mathcal M}}
\def\mn{{\mathcal N}}

\def\ms{{\mathcal S}}

\def\mv{{\mathcal V}}

\def\mx{{\mathcal X}}
\def\my{{\mathcal Y}}
\def\mz{{\mathcal Z}}

\def\e{\mathbb{E}}

\def\tv#1{\left\|#1\right\|_1}
\def\apx#1{\stackrel{#1}{\approx}}

\def\iH{\underline{H}}

\def\ind{\mathbbmss{1}}

 \def\clap#1{\hbox to 0pt{\hss#1\hss}}

\allowdisplaybreaks
\providecommand{\keywords}[1]{\textbf{{Index terms---}} #1}
\begin{document}

\date{}
\author{Mohammad~Hossein~Yassaee, Mohammad~Reza~Aref~and~Amin~Gohari\thanks{\scriptsize \noindent 
The authors are with the Information Systems and Security Lab (ISSL), Department of Electrical Engineering, Sharif University of Technology, Tehran, Iran (e-mails: yassaee@ee.sharif.edu; \{aref,aminzadeh\}@sharif.edu). This work is supported by Iran-NSF under grant No. 92-32575. This paper was presented in part at ISIT 2012.}
}
\title{Achievability Proof via Output Statistics of Random Binning}

\maketitle
\begin{abstract}
 This paper {introduces} a new and ubiquitous framework for establishing achievability results in \emph{network information theory} (NIT) problems. The framework uses random binning arguments and is based on a duality between channel and source coding problems. {Further,} the framework uses pmf approximation arguments instead of counting and typicality. This allows for proving coordination and \emph{strong} secrecy problems where  certain statistical conditions  on the distribution of random variables need to be satisfied. {These statistical conditions include  independence between messages and eavesdropper's observations in secrecy problems and closeness to a certain distribution (usually, i.i.d. distribution) in coordination problems. One important feature of the framework is to enable one {to} add an eavesdropper and obtain a result on the secrecy rates ``for free."}

We make a case for generality of the framework by studying examples in the variety of settings  containing channel coding, lossy source coding, joint source-channel coding, coordination, strong secrecy, feedback and relaying. {In particular, by investigating the framework for the lossy source coding problem over broadcast channel, it is shown that the new framework provides a simple alternative scheme to \emph{hybrid} coding scheme. Also, new results on secrecy rate region (under strong secrecy criterion) of wiretap broadcast channel and wiretap relay channel are derived. In a set of accompanied papers, we have shown the usefulness of the framework to establish achievability results for coordination problems including interactive  channel simulation, coordination via relay and channel simulation via another channel.}
\end{abstract}
\keywords{Random binning, achievability, network information theory, strong secrecy, duality.}

\section{Introduction}
Random coding and random binning are widely utilized in
achievability proofs of the network information theory (NIT) problems.
Random coding is a coding technique that is commonly used to prove the existence of a good codebook (which is a subset of the product set $\mx_{[1:T]}^n:=\prod_{i=1}^T \mx_i^n$), while random binning is a coding technique that partitions the product set into bins with desired properties. Existing achievability proofs for {NIT problems are based on a repeated use of random coding and random binning.} 


\par In this paper, we provide an achievability framework which uses only random binning, by converting {NIT} problems into {\emph{certain}} source coding problems. Let us begin by the problem of sending a message $M$ over a channel $p(y|x)$. Traditional random coding considers an encoder $X^n(M,F)$ and a decoder $\hat{M}(Y^n,F)$ where $F$ is a common randomness, independent of $M$, available to both the transmitter and the receiver. R.v. $F$ represents the random nature of codebook generation. Since the probability of error is evaluated by averaging over all realizations of $F$, one can find $f$ such that $X^n(M,F=f)$ and $\hat{M}(Y^n, F=f)$ form appropriate encoder and decoder. In our framework, we depart from this by first generating $n$ i.i.d.\ copies of $X^n$ and $Y^n$; then we take both $F$ and $M$ to be functions of $X^n$ such that $F$ becomes nearly independent of $M$. Note that we still have the property that $p(y^n|x^n,F=f)=p(y^n|x^n)$ and $p(m|F=f)\apx{}p(m)$ meaning that $X^n(M,F=f)$ and $\hat{M}(Y^n, F=f)$ are legitimate choices as \emph{stochastic} encoder and decoder. We construct $F$ and $M$ as random partitions (binnings) of $X^n$. The question then arises that under what conditions two random bin indices are independent (as in the case of $F$ and $M$), and what is the sufficient condition for recovering $X^n$ from $Y^n$ and a bin index $F$.

In Section \ref{s:motivation}, we discuss the preceding argument and further bring up a key and novel duality between the channel coding problem and secret-key agreement (SK) problem (in the source-model sense) by interpreting $M$  as the key  and $F$ as the pubic message. In particular we discuss how an achievability proof for each of these problems can be converted to an achievability proof for the other one.  It turns out that the two questions in the end of the previous paragraph relate to the secrecy and reliability constraints in the SK agreement problem.

Since the SK agreement problem is a source coding problem with secrecy constraint, and has been previously studied using random binning ideas, the duality gives a proof for the point to point channel coding problem by means of random binning only.

To associate an appropriate source coding problem to a given problem, one needs to answer questions similar to the ones for the point-to-point channel coding problem, i.e. the independence and reliability constraints ({for point-to-point channel coding problem, we had the independence constraint {on $F$ and $M$}  and the reliability constraint {of  recovering $X^n$ from $Y^n$ and $F$}}).  To answer these questions in a more general framework, in Section \ref{s:1}, we prove two main theorems on approximating the joint pmf {(or statistics)} of the bin indices in a distributed random binning. We study properties of random binning in two extreme regimes, namely, when the binning rates are low and high. In the first case, we observe that if the rates of a distributed random binning are sufficiently small, the bin indices are nearly jointly independent, uniformly distributed and independent of a non-binned source $Z^n$. {We call this theorem ``Output Statistics of Random Binning (OSRB) theorem}". This result generalizes the one for the channel intrinsic randomness \cite{bloch}. The second case is the SW region, which shows that if the rates of distributed binning are sufficiently large, the outputs of random binning are enough to recover the sources. {Since the framework deals with the output statistics of random binning, we call the framework as OSRB framework.}

\subsection{Particular features of OSRB}
The proposed framework differs from traditional techniques in the following significant ways:
\begin{itemize}
\item It uses random binning only.
\item  It brings part of the randomness of random codebook
generation from the background into the foreground as
an explicit random variable.
\item It is not based on notions of ``counting" size of typical sets, or typicality decoding. Instead, it uses probability approximation in the sense of vanishing total variation {distance}. This has important implications in problems of secrecy and coordination, as discussed in subsection \ref{subsection:advantage}.
\item The technique allows us to add secrecy for free. Thus, for instance, going from traditional point-to-point communication problem to the wiretap channel problem is immediate.
\item The advantage of the conversion to an \emph{{appropriate}} source coding problem is that we only have \emph{one} copy of the random variables; all the messages and preshared randomness are next constructed as random bins of these i.i.d.\ rv's. However a direct approach to the channel coding problem requires dealing with a large codebook containing \emph{lots of} codeword sequences.
\item  While the traditional techniques view superposition coding and Marton coding as distinct coding constructions, in our framework the two constructions are nothing but two different ways of specifying the set of i.i.d.\ rv's we are binning. Thus, the new framework unifies the two coding strategies, for it \emph{only} uses random binning.
\end{itemize}

\subsection{Advantages of the proposed method}\label{subsection:advantage}
The proposed method has a simple structure (using only random binning), and can solve some problems much easier than the traditional techniques; see \cite{me2, me3} for two examples that are not included in this manuscript. These examples consider the problems of channel simulation and coordination. In \cite{me2}, we find an exact computable characterization of a multi-round channel simulation problem for which only inner and outer bounds were known previously. In coordination problems \cite{cuff},  we want to generate random variables whose joint distribution is close to a desired i.i.d.\ distribution in total variation distance. Traditional techniques (such as packing and covering lemmas) commonly address the probability of error events. This is not general enough to cover all of the total variation distance constraints that show up in the coordination problems. In such cases, one has to come up with new proof techniques. One particular case is the resolvability (or soft covering) lemma used by Cuff (see \cite[Lemma IV. 1]{cuff-trans}, {\cite{wcm,han-verdu}}).

In addition to what discussed above, the framework leads to more rigorous and simpler proofs for secrecy problems. In secrecy problems one has to deal with certain equivocation rates. Generally speaking there are two main techniques for proving lower bounds on equivocation rates: one is to prove existence of ``good" codebooks with given properties, which are then used to compute the equivocation rates. This approach was originally used by Csiszar and Korner in \cite{Csiszar}. The second approach is to compute the expected value of equivocation rates over codebooks, and prove existence of a ``good" codebook with large equivocation rate (in the same way that a codebook with small probability of error is identified).  Some existing works on secrecy follow the second approach in a non-rigorous way. Instead of defining a random variable for the random codebook and conditioning the equivocation rates by that, they use the unconditioned distribution to calculate the equivocation rates. The recent book by El Gamal and Kim \cite{elgamal} uses the second approach in a rigorous way. However, in some scenarios, calculation of equivocation rates conditioned on the codebook random variable can be involved. {We observe that the OSRB framework leads to simple proofs in such cases.} In fact, we show that whenever one solves a problem without secrecy constraint using OSRB framework, he can get a solution for this problem with addition of a secrecy constraint \emph{for free}! Moreover, we can \emph{directly}  prove strong secrecy results for multi terminal scenarios.

\subsection{Related previous works}

Some connections between certain source coding and channel coding problems have been observed in previous works. Slepian and Wolf, in their seminal paper on the lossless source coding \cite{sw}, interpreted the achievability of the rate $R=H(X|Y)$ for compressing the source $X^n$ at rate $R$ to a destination with access to the source $Y^n$, through a channel coding problem. In contrast, Csiszar and Korner, obtained an achievability proof for multiple access channels (MAC) through the distributed source coding problem of Slepian and Wolf \cite{csiszar}. In a recent work \cite{renner}, Renes and Renner showed the achievability of the channel capacity via a combination of Slepian-Wolf (SW) coding and privacy amplification. The main theme in these works is that the set of sequences mapped to the same index through SW coding constitutes a good channel code, and hence, we have a decomposition of the product set into the channel codebooks. However, these works do not provide a systematic and ubiquitous framework for proving achievability results.

Some of the ideas in this work were inspired by the work of Cuff \cite{cuff-trans}. These include use of pmfs as random variables, preserving joint statistics and reverse encoders. However the two frameworks have significant differences in terms of codebook construction and proofs. We consider our framework simpler and more general for the following reason: Cuff's framework is not easily applicable to complicated network structures (such as coordination with relay \cite{me3}), since if one were to extend Cuff's ``soft-covering" lemma to these scenarios, one has to define various mutual soft-covering lemmas and various codebook constructions (just like the traditional mutual covering lemmas). Further, binning provides a common framework and bypasses the need for proving mutual covering lemmas. 

Our approach for proving strong secrecy results resembles the resolvability techniques \cite{hayashi,bloch1}, but to best of our knowledge, resolvability techniques are not developed for multi terminal scenarios except for one work on MACs by Steinberg \cite{Steinberg}. The latter result has been used in \cite{me-itw10} and \cite{bloch2} 
 to prove strong secrecy results for multiple-access wiretap channels and two-way wiretap channels, respectively. However our approach  is able to deal with strong secrecy in general multi terminal scenarios.  There is also another approach for proving  strong secrecy results using the ideas of  privacy amplification of Maurer and Wolf \cite{Maurer}. In this technique, one first proves weak secrecy for a problem and then employs privacy amplification to extract secret message or key in the strong sense. We are not aware if this technique has been extended to the multi-terminal setting. Regardless, in this technique one needs to prove weak secrecy which may be difficult in general multi-terminal setting using traditional techniques. In contrast, the OSRB framework leads to a simple and direct proof for strong secrecy in multi-terminal setting.    

There are connections between the OSRB framework and recent \emph{hybrid coding} approach of Minero, Lim and Kim \cite{hybrid}. In fact, the OSRB framework implicitly employs hybrid coding by its construction. This is discussed in details in Remark \ref{rmkhybridcoding}.

It was brought to our attention by Muramatsu that a structure similar to OSRB based on random binning (more generally, hash functions) has been used in his works \cite{jun:channel,jun:wt12}. While the use of random binning in the works of Muramatsu et al. is similar to ours (in particular its use of binning to get Marton coding), obtaining superposition coding part of Marton's inner bound via binning has been left as an open problem (this can be done by binning nested sets of variables in our framework) \cite{jun:bc}. More importantly, our construction of stochastic encoders-decoders (based on pmf decompositions and using the terms in the decomposition to define encoders) differs from the ones in these works. Further,  Muramatsu et al.'s works use typicality lemmas and counting approximation tools, whereas we use pmf approximation together with the idea of preserving joint pmf among rv's in the source coding and channel coding forms of a problem. Lastly, we apply our framework to a much wider range of problems including those with feedback and relay, and also prove new achievability results. On the other hand, Muramatsu et al.'s works are interested in designing practical codes, whereas we are not.

The OSRB framework is inspired by certain duality between channel coding and the source model SK problem. Broadly speaking, there are  two kinds of duality in the literature, namely \emph{functional} duality and \emph{operational} duality. Functional duality is the duality between formal expressions of the primal and dual problems, e.g. the duality between the mutual information terms in the channel capacity and rate-distortion functions. This type of duality was first pointed out by Shannon between source and channel coding problems \cite{shannon59}. Other examples include duality between source coding with side information and channel coding with state information \cite{cover,pradhan}, duality between packing and covering lemmas and binning and multicoding \cite{elgamal}. The functional duality does not provide an explicit relation between solutions of the primal and the dual problems. On the other hand, operational duality provides a way to construct a solution (a code) for the primal problem using a solution for the dual problem. Operational duality was explored in \cite{verdu-gupta} for lossy compression and channel coding problems, showing that a certain channel decoder can be used as a lossy compressor. The duality used in the OSRB is an operational one.


\subsection{Organization}
This paper is organized as follows: in Section \ref{s:motivation}, we illustrate the main idea of converting a channel coding problem to a source coding problem by showing an interesting duality between the channel coding and SK agreement problems. We also discuss in Subsection \ref{sub:free} how one can obtain secrecy for free from the proof for a problem without secrecy constraint. In Section \ref{s:1}, we state the main theorems to approximate pmfs.    In Section \ref{s:2}, we begin by demonstrating our approach for some primitive problems of NIT, i.e. channel coding and lossy source coding problems, before getting into our new results. Moreover, we show that the achievability proof for channel coding problem can be extended for free to an achievability proof for wiretap channel. We also illustrate how our framework can be used to prove channel (network) synthesis problems by applying our framework to the original channel synthesis problem \cite{cuff-trans}, studied by Cuff, and apply our approach to complicated networks with more than two users. In Subsection \ref{sub:3r} we apply our framework to obtain a new achievable rate region for the problem of three receiver wiretap broadcast channel under a strong secrecy criterion. In Subsection \ref{sub:BT}, we re-prove the achievable rate region for the problem of distributed lossy compression, due to Berger and Tung. In Subsection \ref{sub:jscc}, the OSRB framework is applied to the problem of lossy coding over broadcast channels. In Subsection \ref{sub:nnc}, we show the applicability of OSRB framework to multi-hop networks. To do this, we consider relay channel and re-prove the noisy network coding (NNC) \cite{NNC} inner bound for this problem. We also easily extend the proof to get an extension of NNC inner bound for the problem of wiretap relay channel with strong secrecy criterion, which was not known before. In Section \ref{s:3}, we discuss connections between our framework and the covering  lemma in a multivariate setup by observing that the set of typical sequences can be decomposed into covering 
 codebooks.
\subsection{Notations} In this paper,
we use $X_{\ms}$ to denote $(X_j:j\in\ms)$,
$p^U_{\ma}$ to denote the uniform distribution over the set $\ma$ and $p(x^n)$ to denote the the i.i.d.\ pmf $\prod_{i=1}^np(x_i)$, unless otherwise stated. The total variation between two pmf's $p$ and $q$ on the same alphabet $\mx$ , is defined by $\tv{p(x)-q(x)}:=\frac{1}{2}\sum_x|p(x)-q(x)|$.

\begin{remark} Similar to \cite{cuff-trans}, in this work we frequently use the concept of \emph{random} pmfs which we denote by capital letters (e.g. $P_X$). For any countable set $\mx$, let $\Delta^{\mx}$ be the probability simplex for distributions on $\mx$. A random pmf $P_X$ is a probability distribution over $\Delta^{\mx}$. In other words, if we use $\Omega$ to denote the sample space, the mapping $\omega\in \Omega \mapsto P_X(x;\omega)$ is a random variable for all $x\in\mx$ such that $P_X(x;\omega)\geq 0$ and $\sum_{x}P_X(x;\omega)=1$ for all $\omega$. Thus, $\omega\mapsto P_X(\cdot;\omega)$ is a vector of random variables, which we denote by $P_X$. We can definite $P_{X,Y}$ on product set $\mx\times\my$ in a similar way. We note that we can continue to use the law of total probability with random pmfs (e.g. to write $P_X(x)=\sum_{y}P_{XY}(x,y)$ meaning that $P_X(x;\omega)=\sum_yP_{XY}(x,y;\omega)$ for all $\omega$) and the conditional probability pmfs (e.g. to write $P_{Y|X}(y|x)=\frac{P_{XY}(x,y)}{P_X(x)}$ meaning that $P_{Y|X}(y|x;\omega)=\frac{P_{XY}(x,y;\omega)}{P_X(x;\omega)}$ for all $\omega$).
\end{remark}
\section{Motivation}\label{s:motivation}
The key technique used in the OSRB is to covert an primary problem to a dual problem such that the statistics (i.e. the joint distribution of the r.v.'s) of the primary problem and the dual problem are almost identical. The dual problem is more tractable than the primary one. Solving the dual problem implies a solution for the primary problem. We illustrate this technique by showing a duality between channel coding for a point to point (PTP) channel and the secret key agreement (source model) problem. Indeed, we show how one can use the (Shannon's) achievability proof of channel coding to obtain an achievability proof for the secret key problem and vice versa. Indeed, this duality yields previously unknown results about the source model problem as discussed in Remark \ref{rmk1ll}.
\subsection{Duality between channel coding  and secret key agreement}
\subsubsection{Shannon's achievability proof results in a SK achievability proof} \par
Consider Shannon's achievability proof for the problem of sending a \emph{uniform} message $M$ of rate $R$ over a DMC channel $p_{Y|X}$. In his proof, Shannon used a random codebook $\cc=\{X^n(m)\}_{m=1}^{2^{nR}}$ in which the codewords are generated independently according to an i.i.d.\ pmf $\prod_{i=1}^np_X(x_i)$. The codebook is  shared between the encoder and the decoder. Thus, the random codebook can be viewed as a shared randomness. Given the message $M$ and the codebook $\cc$, the encoder sends $X^n(M,\cc)$ over the channel. The decoder uses his observation $Y^n$ and the codebook $\cc$ to estimate the transmitted message. Shannon showed that the error probability, averaged over the random codebooks, is small; therefore there exists a \emph{good} codebook with a negligible error probability. Let $p_{M\cc X^nY^n}$ be the induced pmf on the message, codebook, channel input and the channel output. The following observations are useful in the rest of this subsection:
\begin{itemize}
\item The codebook $\cc$ and the message $M$ are independent; thus, $p_{M\cc}=p^U_{M}p_{\cc}$. Hence, conditioned on an instance of the codebook, the uniformity of the message is not disturbed.
\item While for a fixed codebook the channel input distribution is uniform over the codewords, the i.i.d.\ generation of codebook makes the input distribution i.i.d., that is, $p_{X^n}(x^n)=\prod_i p_X(x_i)$.
\item The Markov chain $M,\cc-X^n-Y^n$ holds and the channel is DMC. Thus, the joint distribution of channel input and channel output is i.i.d., that is, $p_{X^nY^n}(x^n,y^n)=\prod_i p_{XY}(x_i,y_i)$. Moreover we have $p_{M\cc X^nY^n}=p^U_{M}p_{\cc}p_{X^n|M,\cc}p_{Y^n|X^n}=p_{X^n}p_{M\cc |X^n}p_{Y^n|X^n}$.
\end{itemize}

\begin{figure}
\centering\includegraphics[width=.55\linewidth]{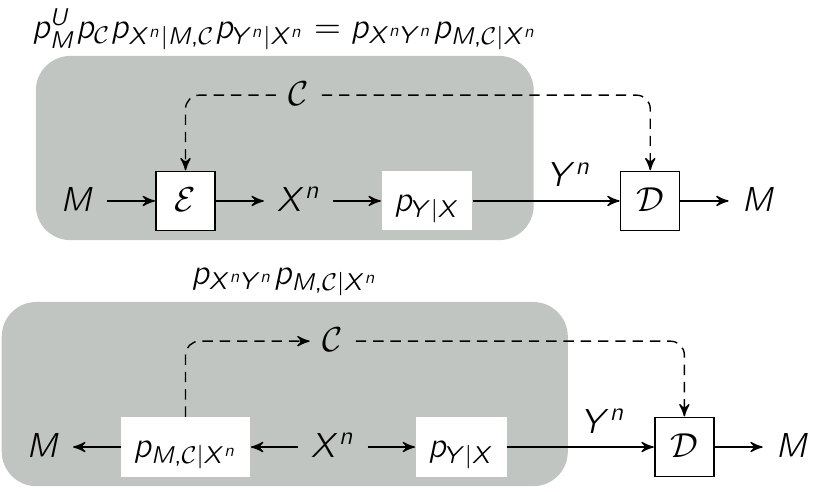}
        \caption{\small (Top) Shannon's achievability proof with random codebook as a pre-shared randomness. Here, the codebook and the message are independent. Randomness in the codebook makes the input and output jointly i.i.d.\ (Bottom) Source model SK problem with the i.i.d.\ correlated sources $X^n$ and $Y^n$. Reversing the encoder gives a feasible encoder for the SK problem, in which $M$ and $\cc$ take the roles of key and public message, respectively. Since the joint distribution of all r.v.'s is preserved, the uniformity of the message (key) and the independence between  key (message) and public message (random codebook) are also preserved.}\label{fig:sh}
\end{figure}

These observations are illustrated in the top diagram of Fig. \ref{fig:sh}. To convert Shannon's achievability proof to a SK achievability proof, we proceed as follows. Since $p_{M\cc X^n}=p_{X^n}p_{M\cc |X^n}$,  one can think of this as passing an i.i.d.\ source $X^n$ through a reverse encoder $p_{X^n}p_{M\cc |X^n}$ to obtain $M$ and $\cc$. This is depicted in the bottom diagram of Fig. \ref{fig:sh} where we have changed the direction of the arrows to reflect this change of order. Moreover, as $(X^n,Y^n)$ are jointly i.i.d., one can consider $p_{M\cc |X^n}$ as an encoder for the SK problem in which $\cc$ and $M$ are the public message and key, respectively. For decoding, we take the decoder of channel coding problem and use it for the SK problem. Observe that the joint distribution of r.v.'s in the channel coding problem and the SK problem are equal; thus these models are equivalent. In particular,
\begin{itemize}
\item  The key $M$ and the public message $\cc$ are independent.
\item The error probability of decoding of the key $M$ is equal to that of channel coding. Thus, if the error probability of the channel coding is negligible, then the error probability of SK problem is also negligible. This shows that the rate $I(X;Y)$ is achievable.
\end{itemize}
To sum this up, Shannon's achievability proof results in a SK achievability proof. Further, we have \emph{complete} independence between the key and the public message.
\begin{remark}\label{rmk1ll}
Although the preceding argument is used to prove the SK achievability result in the asymptotic regime for the i.i.d.\ sources, it can be applied to one shot (single-use) regime. To see this, one can replace the i.i.d.\ sources $X^n$ and $Y^n$ with sources $X$ and $Y$, generate codebook according to $p_X$ instead of the i.i.d.\ $p_{X^n}$ and use $p_{M\cc |X}$ instead of $p_{M\cc |X^n}$. Then, the error probability of Shannon's achievability proof and SK achievability proof are the same and we have \emph{complete} independence between the key and the public message. In addition, applying this result to general sources in the asymptotic regime implies that the key-rate $\underline{I}(X;Y)$ is achievable using its achievability for channel coding with general input-output (see \cite{book:han} for a definition of a general input-output channel). This potentially improves on the previous random binning bound  $\underline{H}(X)-\overline{H}(X|Y)$ in \cite{bloch}. More importantly, this proof technique is not restricted to discrete sources and can be applied to any correlated sources with abstract alphabets.
\end{remark}
\subsubsection{SK achievability proof results in an achievability proof for channel coding problem}\label{subs:bin}
The traditional SK achievability proof is based on a random binning argument. Similarly, we  show that a random binning argument can be used to prove the achievability part of the PTP channel coding problem. In SK agreement problem, we have  i.i.d.\ copies of correlated sources $(X^n,Y^n)$. The traditional  SK achievability proof uses two random bin indices of the source $X^n$ to obtain the public message $F$ and the key $M$. The relation among r.v.'s is depicted in the top diagram of Fig. \ref{fig:bin}. The random bin $F$ serves as a Slepian-Wolf (SW) index with rate $R_F>H(X|Y)$. It enables the receiver to recover $X^n$ with high probability. Through this, it can recover $M$ as a bin index of $X^n$. Next, we consider the channel coding counterpart. Again, one can interpret the key as the message and the public message as the shared randomness. We use reverse encoder $P_{X^n|MF}$ obtained from random binning as a stochastic encoder for the channel coding problem. Also, we use the decoder of SK problem as a channel decoder. The relation among r.v.'s in the channel coding counterpart is depicted in the bottom diagram of Fig. \ref{fig:bin}. If the joint distribution of $M$ and $F$ is equal to the $P_{MF}$ (induced by random binning), then the joint distribution of all r.v.'s in the SK problem and its channel coding counterpart are equal which implies that the error probability  of channel coding problem is negligible. To get away with shared randomness, one can find a good instance $F=f$ of the shared randomness such that $\Pr(\texttt{error}|F=f)$ is also negligible. However, conditioned on $F=f$, the distribution $P_{M|F=f}$ may be disturbed and it is not necessarily uniform. Therefore, we are interested to finding constraint on the rates of $M$ and $F$ such that the following properties hold:
\begin{figure}
\centering\includegraphics[width=.55\linewidth]{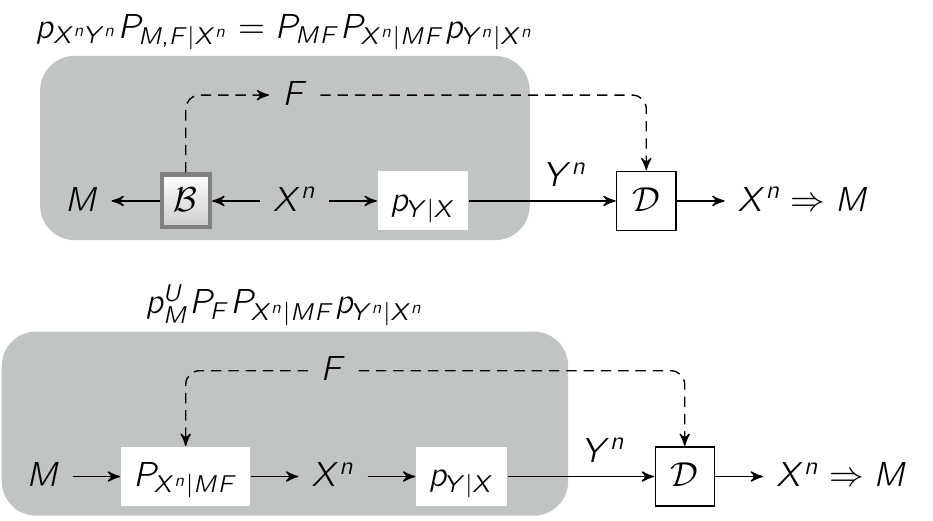}
        \caption{ (Top) SK achievability using random binning. $F$ is a SW bin index with rate  $R_F>H(X|Y)$. This results in the reliable decoding of the source $X^n$ and consequently reliable decoding of the key $M$. (Bottom) Channel coding counterpart of SK problem. Reversing the encoder gives a feasible encoder for the channel coding problem, in which $M$ and $F$ take the role of message and shared randomness, respectively. The SK problem and its channel coding counterpart are equivalent, if the secrecy requirements of SK problem are satisfied; that is, $P_{MF}\apx{} p^U_MP_F$. The constraint $R_F+R_M<H(X)$ is sufficient to guarantee this approximation.}\label{fig:bin}
\end{figure}

\begin{itemize}
\item $M$ is almost a uniform random variable,
\item $M$ and $F$ are almost independent. This ensures that conditioned on an instance $F=f$, the uniformity of the message is not disturbed.
\end{itemize}
These two properties are the secrecy requirements of the SK problem. Using a result of \cite{csiszar:96,bloch}, one can see that these two properties hold as long as $R_F+R_M<H(X)$.

The above argument is a common one used in the OSRB framework. We always associate a source coding problem to a given problem, calling it ``the source coding side of problem". In this simple example, the top diagram of Fig. \ref{fig:bin} is the source coding side of the channel coding problem. We then convert the associated source coding side to the main problem using appropriate reverse encoders, with one exception; here we have added a shared randomness to the main problem. We then find constraints that the joint distribution of r.v.'s in the main problem and the source coding side are approximately equal. Next we find constraints that satisfy the desired properties such as reliability and secrecy in the source coding side. Finally, we remove the shared randomness  without disturbing the desired properties.

The advantage of conversion to a source coding problem is that in the source coding side of the problem we only have one copy of {i.i.d.} random variables. In the source coding side of the problem discussed above, we started from a single i.i.d.\ copy of $X^n,Y^n$. All the other rv's (i.e. $M$ and $F$) are random bins of these i.i.d.\ rv's. However if we were to directly attack the channel coding problem, we had to create a codebook of size $2^{nR}$ containing lots of $x^n$ sequences. This may not seem significant in this simple channel coding example. However, in problems involving multi-round interactive communication with several auxiliary random variables (e.g. [4]), it is desirable to have just a single i.i.d.\ repetition of all the original and auxiliary random variables in our framework (rather than having many i.i.d.\ copies of these random variables related to each other through superposition or Marton coding type structures). Once we take a single i.i.d.\ copy, all the messages and pre-shared randomness (such as $F$) can be constructed as random bins of these i.i.d.\ rv's. Traditional coding techniques start with the messages and then create the many codewords. Here, we are reversing the order by starting from a single i.i.d.\ copy of the original and auxiliary rv's and constructing the messages as bin indices afterwards.

 \subsection{Secrecy is free!}\label{sub:free}
 One advantage of the proposed framework is to solve secrecy problems for free!, in the sense that once a  problem without secrecy constraint is solved, the corresponding problem with secrecy constraint can be solved with minor modifications. To illustrate this, we show how the achievability proof of the channel coding problem using random binning gives an achievability proof for wiretap channel problem for free.

 \begin{figure}
\centering\includegraphics[width=.55\linewidth]{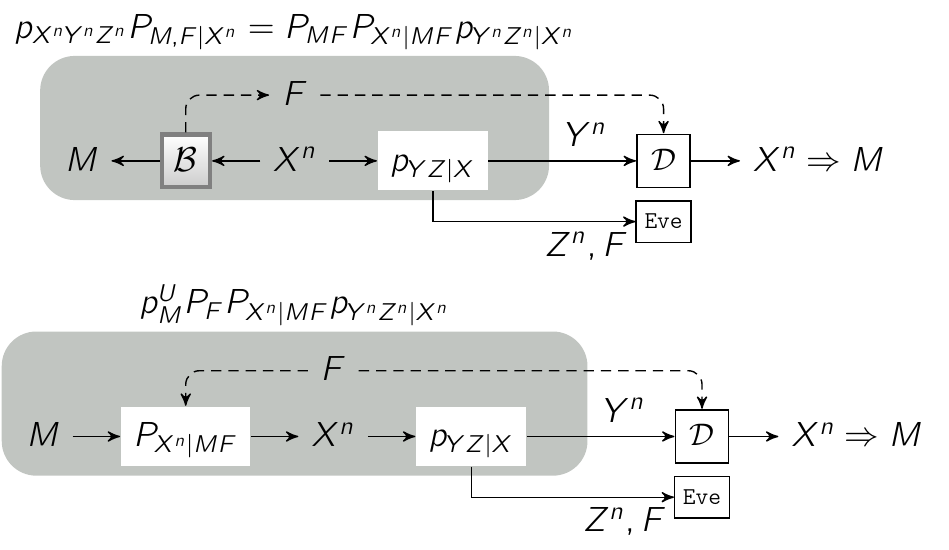}
        \caption{{\small (Top) Source coding side of the wiretap channel. Here, the eavesdropper has access to the shared randomness $F$ in addition to channel output $Z^n$. (Bottom) Wiretap channel. We need to have equivalence between wiretap channel and its source coding side and the secrecy constraint $M\bot(F,Z^n)$. For these to happen, it suffices to have mutual independence among $M$, $F$ and $Z^n$. This holds as long as $R_F+R_M<H(X|Z)$.}}\label{fig:binsec}
\end{figure}
Consider the diagrams in the Fig. \ref{fig:binsec} which are the same as the ones in the Fig. \ref{fig:bin} for channel coding problem with one exception; we have added eavesdropper to this figure. Following the argument used in the subsection \ref{subs:bin}, we have the equivalence between the source coding side of the problem (the top diagram) and the wiretap channel (the bottom diagram), i.e. the joint distribution of all r.v.'s are approximately equal, as long as $R_F+R_M<H(X)$. Thus, it suffices to ensure the secrecy constraint for the source coding side  and it will automatically hold for the channel coding side of the problem. Since $F$ is a shared randomness, eavesdropper has access to it. Thus, the secrecy requirement is the independence between  $M$ and $(F,Z^n)$ available at the eavesdropper (it is worth to note that conditioned on an instance of $F=f$, the independence between $M$ and $Z^n$ conditioned on $F=f$ is satisfied.). In fact, we obtain a constraint on the rates of $M$ and $F$ such that $M$, $F$ and $Z^n$ are almost mutually independent. This immediately implies the desired independence. It turns out that this condition holds as long as $R_F+R_M<H(X|Z)$. Comparing this with the constraint $R_F+R_M<H(X)$ for the channel coding problem without secrecy constraint (coming from the independence of $M$ and $F$ without $Z^n$), we observe that $Z$ is only added to the conditioning part of the entropy in the constraint. This is a common phenomenon in secrecy problems. Having solved a problem without secrecy using OSRB, the corresponding problem with secrecy can be solved by adding eavesdropper's information to the conditioning part of appropriate constraints appearing in the solution of the problem without secrecy; thus, our remark that secrecy is free in the OSRB framework.

Finally, the reliability constraint $R_F>H(X|Y)$ and the secrecy constraint $R_F+R_M<H(X|Z)$ give the achievability of the rate $R_M<I(X;Y)-I(X;Z)$. The achievability of more general formula $R_M<I(U;Y)-I(U;Z)$ can be proved using the combination of channel prefixing technique and the above argument.
\section{Output statistics of random binning}\label{s:1}
Let $(X_{[1:T]},Z)$ be  a {discrete memoryless correlated sources} distributed according to a joint pmf $p_{X_{[1:T]},Z}$ on a countably infinite set $\prod_{i=1}^T \mx_i\times \mz$. A distributed  random  binning consists of a set of random mappings $\mb_i: \mx_i^n\rightarrow [1:2^{nR_i}]$, $i\in[1:T]$, in which $\mb_i$ maps each sequence of $\mx_i^n$ uniformly and independently to the set $[1:2^{nR_i}]$. We denote the random variable $\mb_t(X_t^n)$ by $B_t$. {Also we denote the realization of $B_t$ by $b_t$.} A random distributed  binning induces the following \emph{random pmf} on the set $\mx_{[1:T]}^n\times\mz^n\times\prod_{t=1}^T [1:2^{nR_t}]$,
\[
P(x^n_{[1:T]},z^n,b_{[1:T]})=p(x_{[1:T]}^n,z^n)\prod_{t=1}^T\ind\{\mb_t(x_t^n)=b_t\},
\]
where we have used capital $P$ to indicate the probabilistic random binning, implying that the pmf induced on $x^n_{[1:T]},z^n,b_{[1:T]}$ is random. One can easily verify that $(B_1,\cdots,B_T)$ are uniformly distributed and mutually independent of $Z^n$ in the mean, that is
\bes
\label{eq:apx}
\e P(z^n,b_{[1:T]})=2^{-n\sum_{t=1}^TR_t}p(z^n)=p(z^n)\prod_{t=1}^T p^U_{[1:2^{nR_t}]}(b_t).
\ees
\par The following theorem finds constraints on the rate-tuple $(R_1,\cdots,R_T)$, such that the preceding observation about the mean holds for almost any realization of the distributed binning. We will be using this theorem frequently in the proofs. A more general form of this theorem is provided and proved in Appendix \ref{apx:osrb}.
\begin{theorem}\label{thm:re}
If for each $\ms\subseteq [1:T]$, the following constraint holds
\begin{align}
\sum_{t\in\ms}R_t<H(X_{\ms}|Z),
\end{align}
then as $n$ goes to infinity, we have
\be
\e_{{\mb}}\tv{P(z^n,b_{[1:T]})-p(z^n)\prod_{t=1}^T p^U_{[1:2^{nR_t}]}(b_t)}\rightarrow 0,
\ee
{where $\mb$ is the set of all random mappings, i.e. $\mb=\{\mb_i:i\in[1:T]\}$}.
\end{theorem}

\begin{remark}
In \cite{bloch}, the \emph{channel intrinsic randomness} was defined ``as the maximum random bit rate that can be extracted from a channel output independently of an input with known statistics". One can generalize this definition to the broadcast channel $p_{X_{[1:T]}|Z}$, in the sense of finding $T$ strings of random bits with rates $(R_1,\cdots,R_T)$ such that the $i-th$ string is extracted individually from the $i-th$ channel output $X_i^n$, while making sure that these random strings are mutually independent of each other and of the channel input $Z^n$. Theorem \ref{thm:re} gives an achievable rate region for this scenario and implies that random binning is sufficient to prove the achievability.\footnote{{
In fact,  \cite{bloch} considered the case for general channel with general input and the results is based on the information spectrum methods. The achievability proof in \cite{bloch} follows from \cite[Theorem 1]{csiszar:96} whose proof is based on graph-coloring. The proof of Theorem \ref{thm:re}  can be easily extended  to this general setting, 
in which one should substitute average entropy with the \emph{spectral inf-entropy} (which is defined in \cite{book:han}), to get the result for this general case. Our proof is based on a simple application of Jensen's inequality.}}
\end{remark}

Sometimes  we \emph{only} need the independence of one random bin from other random bins and $Z^n$. The following corollary provides sufficient conditions for the independence of $B_1$ from $(B_2,\cdots,B_T,Z^n)$. The proof is provided in Appendix \ref{apx:osrbcor}.
\begin{corollary}\label{cor:OSRB}
Let $\mv$ be an arbitrary subset of $[2:T]$. If for each $\ms\subseteq [2:T]-\mv$, the following constraint holds
\be
R_1+\sum_{t\in\ms}R_t<H(X_1X_{\ms}|ZX_{\mv}),\label{eq:1000}
\ee
then as $n$ goes to infinity, we have
\be
\e_{{\mb}}\tv{P(z^n,b_{[1:T]})-p^U(b_1)P(z^n,b_{[2:T]})}\rightarrow 0.
\ee
\end{corollary}

 Theorem \ref{thm:re} enables us to approximate the pmf $P(z^n,b_{1:T})$. We now consider another region for which we can approximate a specified pmf. This region is the Slepian-Wolf region for reconstructing $X^n_{[1:T]}$ in the presence of $(B_{1:T},Z^n)$ at the decoder.
 As in the achievability proof of the \cite[Section 10.3.2]{elgamal}, we can define a decoder with respect to any fixed distributed binning.  We denote the decoder by the random conditional pmf $P^{SW}(\hat{x}^n_{[1:T]}|z^n,b_{[1:T]})$ (note that since the decoder is a function, this pmf takes only two values, 0 and 1).\footnote{{For a Slepian-Wolf decoder that uses a jointly typical decoder, $P^{SW}(\hat{x}^n_{[1:T]}|z^n,b_{[1:T]})=1$ if $\hat{x}^n_{[1:T]}$ is the only jointly typical sequence with $z^n$ in the bin $b_{[1:T]}$. If the \emph{unique} jointly typical sequence in the bin does not exist, then $\hat{x}^n_{[1:T]}$ is taken to be a fixed arbitrary sequence.}} Now we write the Slepian-Wolf theorem in the following equivalent form.
\begin{lemma}\label{le:sw}
If for each $\ms\subseteq [1:T]$, the following constraint holds
\be
\sum_{t\in\ms}R_t>H(X_{\ms}|X_{\ms^c}, Z),
\ee
then as $n$ goes to infinity, we have
\bes
\e_{{\mb}}\tv{P(x^n_{[1:T]},z^n,\hat{x}^n_{[1:T]})-p(x^n_{[1:T]},z^n)\ind\{\hat{x}^n_{[1:T]}=x^n_{[1:T]}\}}\rightarrow 0.
\ees
\end{lemma}
\begin{proof} From the definition of the total variation, we know that $\tv{p(x)-q(x)}=\sum_{x:p(x)> q(x)}[p(x)-q(x)]$. Using this property we can write
\begin{align}
\e\tv{P(x^n_{[1:T]},z^n,\hat{x}^n_{[1:T]})-p(x^n_{[1:T]},z^n)\ind\{\hat{x}^n_{[1:T]}=x^n_{[1:T]}\}}
&\stackrel{(a)}{=}\e\sum_{x^n_{[1:T]},z^n,\hat{x}^n_{[1:T]}:\atop \hat{x}^n_{[1:T]}\neq x^n_{[1:T]}}P(x^n_{[1:T]},z^n,\hat{x}^n_{[1:T]})\n
&=\e P(\hat{X}^n_{[1:T]}\neq X^n_{[1:T]})\rightarrow 0,
\end{align}
where (a) follows from the fact that whenever $P(x^n_{[1:T]},z^n,\hat{x}^n_{[1:T]}) > p(x^n_{[1:T]},z^n)\ind\{\hat{x}^n_{[1:T]}=x^n_{[1:T]}\}$ we must have $\ind\{\hat{x}^n_{[1:T]}=x^n_{[1:T]}\}=0$, since $P(x^n_{[1:T]},z^n,\hat{x}^n_{[1:T]})=p(x^n_{[1:T]},z^n)P(\hat{x}^n_{[1:T]}|x^n_{[1:T]},z^n)$.
\end{proof}

Sometimes we need a special case of SW theorem for recovering only one source $X_1^n$ from random bins $B_1,\cdots,B_T$ and $Z^n$. The following lemma gives sufficient conditions on this problem:
\begin{lemma}
\label{le:2000}
If for each $\ms\subseteq [2:T]$, the following constraint holds
\be
R_1+\sum_{t\in\ms}R_t>H(X_1X_{\ms}|X_{\ms^c}Z),\label{eq:2000}
\ee
then there exists an appropriate decoder such that the error probability of recovering $X_1^n$ from $(Z^n,B_{[1:T]})$ 
 tends to zero as $n\rightarrow\infty$. Equivalently, we have
\bes
\e_{{\mb}}\tv{P(x^n_{[1:T]},z^n,\hat{x}^n_{1})-p(x^n_{[1:T]},z^n)\ind\{\hat{x}^n_{1}=x^n_{1}\}}\rightarrow 0.
\ees
\end{lemma}
\section{Achievability proof through probability approximation}\label{s:2}
In this section, we illustrate the OSRB framework in details
 through some examples. Before going through these examples, we state some useful lemmas on total variation of arbitrary (random) pmfs.
\begin{definition}\label{def:1}
For any random pmfs $P_X$ and $Q_X$ on $\mx$, we write $P_X\stackrel{\epsilon}{\approx}Q_X$ if $\e\tv{P_X-Q_X}<\epsilon$. Similarly we use $p_X\apx{\epsilon}q_X$ for two (non-random) pmfs to denote the total variation constraint $\tv{p_X-q_X}<\epsilon$.\end{definition}
\begin{definition}\label{def:0-1}
For any two sequences of random pmfs $P_{X^{(n)}}$ and $Q_{X^{(n)}}$ on $\mx^{(n)}$ (where $\mx^{(n)}$ is arbitrary and it differs from $\mx^n$ which is a cartesian product), we write $P_{X^{(n)}}\stackrel{}{\approx}Q_{X^{(n)}}$ if {$\lim_{n\rightarrow\infty}\e\tv{P_{X^{(n)}}-Q_{X^{(n)}}}=0$}. Similarly we use $p_{X^{(n)}}\apx{}q_{X^{(n)}}$ for two {sequences of} (non-random) pmfs.\end{definition}
\begin{lemma}\label{le:0-total} We have
\begin{enumerate}
\item \cite[Lemma 17]{cuff}:
$\tv{p_Xp_{Y|X}-q_{X}p_{Y|X}}=\tv{p_X-q_X}$\\
\cite[Lemma 16]{cuff}:$~~~~~~~~~~~\quad\tv{p_X-q_X}\le\tv{p_Xp_{Y|X}-q_{X}q_{Y|X}}$.
\item If $p_Xp_{Y|X}\stackrel{\epsilon}{\approx}q_Xq_{Y|X}$, then there exists $x\in\mx$ such that $p_{Y|X=x}\stackrel{2\epsilon}{\approx}q_{Y|X=x}$.
\item[] ~~$2')$ More generally the probability of the set $\{x\in\mx: p_{Y|X=x}\stackrel{\sqrt{\epsilon}}{\approx}q_{Y|X=x}\}$ under both $p_X$ and $q_X$ is at least $1-2\sqrt\epsilon$.
\item If $P_X\stackrel{\epsilon}{\approx}Q_X$ and  $P_XP_{Y|X}\stackrel{\delta}{\approx}P_XQ_{Y|X}$, then $P_{X}P_{Y|X}\stackrel{\epsilon+\delta}{\approx}Q_{X}Q_{Y|X}$.
\end{enumerate}
\end{lemma}
\begin{proof}
See Appendix \ref{apx:le-total}.
\end{proof}

Lemma \ref{le:0-total} and Definition \ref{def:0-1} immediately imply the following variant of Lemma \ref{le:0-total} which is used throughout the paper.
\begin{lemma}\label{le:total} We have
\begin{enumerate}
\item   $P_{X^{(n)}}\apx{}Q_{X^{(n)}}\Rightarrow P_{X^{(n)}}P_{Y^{(n)}|X^{(n)}}\apx{}Q_{X^{(n)}}P_{Y^{(n)}|X^{(n)}}$,\\
           $P_{X^{(n)}}P_{Y^{(n)}|X^{(n)}}\apx{}Q_{X^{(n)}}Q_{Y^{(n)}|X^{(n)}}\Rightarrow P_{X^{(n)}}\apx{}Q_{X^{(n)}}$.

\item If $p_{X^{(n)}}p_{Y^{(n)}|X^{(n)}}\stackrel{}{\approx}q_{X^{(n)}}q_{Y^{(n)}|X^{(n)}}$, then there exists a sequence $x^{(n)}\in\mx^{(n)}$ such that $p_{Y^{(n)}|X^{(n)}=x^{(n)}}\stackrel{}{\approx}q_{Y^{(n)}|X^{(n)}=x^{(n)}}$.
\item If $P_{X^{(n)}}\apx{}Q_{X^{(n)}}$ and  $P_{X^{(n)}}P_{{Y^{(n)}}|{X^{(n)}}}\apx{}P_{X^{(n)}}Q_{{Y^{(n)}}|{X^{(n)}}}$, then $P_{{X^{(n)}}}P_{{Y^{(n)}}|{X^{(n)}}}\apx{}Q_{{X^{(n)}}}Q_{{Y^{(n)}}|{X^{(n)}}}$.
\end{enumerate}
\end{lemma}
\color{black}
\begin{lemma}\label{le:distortion}
If $d:\mx\times\my\rightarrow[0,d_{max}]$ is a bounded distortion measure, $p_{XY}$ is a pmf with $\e_{p_{XY}}d(X,Y)=D$ and $q_{XY}$ is a pmf such that $q_{XY}\apx{\epsilon}p_{XY}$, then we have
\be\label{eq:dis}
\e_{q_{XY}}d(X,Y)\le D+\epsilon d_{max}.
\ee
\end{lemma}
\begin{proof}
See Appendix \ref{apx:le-distortion}.
\end{proof}
\subsection{OSRB framework}
{In the previous section, we described the OSRB framework for  the achievability proof of the channel coding problem at an intuitive level. Here}
we set up a general proof structure that we will use consistently throughout this paper.
 The OSRB farmework is divided into three parts.
 \begin{framed}
 \begin{itemize}
  \item 
  \emph{Part (1) of the proof:} we introduce two protocols each of which induces a pmf on a certain set of r.v.'s. { To define these protocols, we assume that there exists a shared randomness among all parties of the problem.} The first protocol {is related to the dual problem (or source coding side of the problem) and does not lead to a concrete coding algorithm}. However the second protocol is suitable for construction of a code, with one exception: the second protocol is assisted with a common randomness that does not really exist in the model.
  \item
   \emph{Part (2) of the proof:} we {first} find constraints implying that the two induced distributions are almost identical. {In other words,  the two protocols are equivalent. Thus  it suffices to resort to the source coding side of problem and investigate the desired properties such as reliability (or vanishing error probability), secrecy, distortions, etc in the source coding side of the problem.}
      \item
    \emph{Part (3) of the proof:} we eliminate the shared  randomness  given to the second protocol without disturbing {the desired properties}. {To do this, we find an instance of the shared randomness such that conditioned on it, the desired properties  still hold.} This makes the second protocol useful for code construction.
\end{itemize}
\end{framed}
\subsection{Channel coding}
A formal proof of the point-to-point channel coding problem is as follows:

\emph{Part (1) of the proof:}
Take some arbitrary $p(x)$. We define two protocols each of which induces a joint distribution on random variables that are defined during the protocol.

\emph{Protocol A (source coding side of the problem). }
Let $(X^n,Y^n)$ be i.i.d.\ and distributed according to $p(x,y)=p(x)p(y|x)$.

\underline{Random Binning}: Consider the following random binning: to each sequence $x^n$, assign uniformly and independently two bin indices $m\in[1:2^{nR}]$ and $f\in[1:2^{n\tR}]$. Further, we use a $P^{SW}(\hat{x}^n|y^n,f)$ Slepian-Wolf decoder to recover $x^n$ from $(y^n,f)$. We denote the output of the decoder by  $\hat{x}^n$. The rate constraint for the success of the decoder will be discussed later, although this decoder can be conceived even when there is no guarantee of success.

The random \footnote{The pmf is random due to the random binning assignment in the protocol.}
pmf induced by the random binning, denoted by $P$, can be expressed as follows:
\begin{align}
P(x^n,y^n,m,f,\hat{x}^n)&=p(x^n,y^n)P(m,f|x^n)P^{SW}(\hat{x}^n|y^n,f)\n
&=P(m,f,x^n)p(y^n|x^n)P^{SW}(\hat{x}^n|y^n,f)\n
&= P(m,f)P(x^n|m,f)p(y^n|x^n)P^{SW}(\hat{x}^n|y^n,f).\label{eq:pmf0}\end{align}

\emph{Protocol B (main problem assisted with shared randomness).} In this protocol we assume that the transmitter and the receiver have access to the
shared randomness $F$ where $F$ is uniformly distributed over $[1:2^{n\tR}]$.
Then, the protocol proceeds as follows:
\begin{itemize}
\item The transmitter chooses a message $m$ uniformly distributed over $[1:2^{nR}]$ and independently of $F$.
\item In the second stage, knowing $(m,f)$, the transmitter generates a sequence $x^n$ according to  the conditional pmf $P(x^n|m,f)$ of the protocol A. Then it sends $x^n$ over the channel.
\item At the final stage, the receiver, knowing $(y^n,f)$ uses the Slepian-Wolf decoder $P^{SW}(\hat{x}^n|y^n,f)$ of protocol A to obtain $\hx^n$ as an estimate of $x^n$. Then, it declares the bin index $\hat{m}=\mathsf{M}(\hx^n)$ assigned to $\hx^n$ as the estimate of the transmitted message $m$.
\end{itemize}
The random pmf induced by the protocol, denoted by $\hat{P}$, factors as
\begin{align}
\hat{P}(x^n,y^n,m,f,\hat{x}^n)&=p^U(f)p^U(m)P(x^n|m,f)p(y^n|x^n)P^{SW}(\hat{x}^n|y^n,f).\label{eq:pmf2}
\end{align}

\emph{Part (2a) of the proof: Sufficient conditions that make the induced pmfs approximately the same}: To find the constraints that imply that the pmf $\hat{P}$ is close to the pmf $P$ in total variation distance,
we start with $P$ and make it close to $\hat{P}$ in a few steps. The first step is to observe that $m$ and $f$ are the bin indices of $x^n$ in Protocol A. Theorem \ref{thm:re} implies that if $R+\tR<H(X)$
then we have $P(m,f)\apx{}p^U(m)p^U(f)=\hat{P}(m,f)$.
Equations \eqref{eq:pmf0} and \eqref{eq:pmf2} imply
\begin{align}
\hat{P}(m,f,x^n,y^n,\hat{x}^n)&\apx{}
P(m,f,x^n,y^n,\hat{x}^n).\label{eqn:ee1}
\end{align}

\emph{Part (2b) of the proof: Sufficient conditions that make the Slepian-Wolf decoder succeed}: The next step is 
{to see that when the Slepian-Wolf decoder of protocol A can reliably decode the transmitted sequence $X^n$.} Lemma \ref{le:sw} requires imposing the constraint $\tR>H(X|Y)$.
It yields 
\begin{align}
P(m,f,x^n,y^n,\hat{x}^n)\apx{}P(m,f,x^n,y^n)\ind\{\hat{x}^n=x^n\}.\label{eq:pmf1.5}
\end{align}

Using equations \eqref{eqn:ee1}, \eqref{eq:pmf1.5} and the triangle inequality, we have
\begin{align}
\hat{P}(m,f,x^n,y^n,\hat{x}^n)\apx{}P(m,f,x^n,y^n)\ind\{\hat{x}^n=x^n\}.\label{eq:pmf1}
\end{align}

\emph{Part (3) of the proof: Eliminating the shared randomness.}
 In the protocol we assumed that the transmitter and the receiver have access to shared randomness\footnote{{It is worthy to note that the random binning map is also shared between the transmitter and the receiver. So we must  simultaneously find  a good fixed binning and a good instance $f$. However the random binning is a usual shared randomness and can be regarded as  \emph{background} randomness. In the other hand, the shared randomness $F$ plays an essential role in our framework and does not exist in the other works, so one can regard this kind of randomness as \emph{foreground} randomness. In the rest of the paper, we emphasize the foreground randomness while bearing in mind that random binning is the background randomness. Also, whenever we eliminate the shared randomness, we first find a good fixed binning with some desired properties and then remove the foreground randomness with respect to this fixed binning.}} $F$ which is not present in the model. Nevertheless, we show that the transmitter and the receiver can agree on an instance $f$ of $F$.  Using Definition \ref{def:1}, equation \eqref{eq:pmf1} guarantees the existence of  a fixed binning with the corresponding pmf $p$ such that if we replace $P$ with $p$ in \eqref{eq:pmf2} and denote the resulting pmf with $\hat{p}$, then $\hat{p}(m,f,x^n,y^n,\hat{x}^n)\apx{}p(m,f,x^n,y^n)\ind\{\hat{x}^n=x^n\}$. In particular, this gives
$\hat{p}(\hat{X}^n\neq X^n)\le{\epsilon_n}$ for some vanishing sequence $\epsilon_n$. This guarantees the existence of a good instance $F=f$ such that
$\hat{p}(X^n\neq\hat{X}^n|f)\le \epsilon_n$. Reliable recovery of the transmitted $X^n$ implies reliable recovery of the message $M$. 

Finally, identifying $p(x^n|m,f)$ as the encoder and ($p^{SW}(\hat{x}^n|y^n,f),\mathsf{M}(\hx^n)$) as the decoder results in a pair of encoder-decoder with the probability of error at most $\epsilon_n$.

\subsection{Wiretap channel (secrecy for free)}
We now turn our attention to wiretap channel, to show in details that how one can prove secrecy for free. We use the strong secrecy in terms of vanishing total variation distance as our secrecy criterion. First we have the following formal definition.

\emph{Problem definition:}
Consider the problem of secure transmission over a wiretap channel, $p(y,z|x)$. Here, we wish to securely transmit a  message $M\in[1:2^{nR_0}]$ to the receiver $Y$,  while concealing it from the wiretapper. We use the total variation distance as a measure for analyzing the secrecy. Formally speaking there are,
\begin{itemize}
\item A message $M$ which are mutually independent and uniformly distributed,
\item A stochastic encoder $p^{enc}(x^n|m)$,
\item A decoder which assigns an estimate $\hat{M}$ of $M$ to each $y^n$.
\end{itemize}
\par A rate $R$ is said to be achievable if $\Pr\{\hat{M}\neq M\}\rightarrow 0$ and $M$ is nearly independent of the wiretapper output, $Z^n$, that is,
\[
\tv{p(m,z^n)-p^U_{\mm}(m)p(z^n)}\rightarrow 0,
\]
where, here $p(z^n)$ is the induced pmf on $Z^n$ and is not an i.i.d.\ pmf.

Here we want to prove the achievability of the rate $I(U;Y)-I(U;Z)$ using our framework. Without loss of generality, we can assume $U=X$, so we will prove  the achievability of the rate $I(X;Y)-I(X;Z)$ for the wiretap channel. The proof follows exactly from the proof of channel coding with one exception, here we must satisfy the secrecy criterion in addition to reliability.

We set up Protocol A and Protocol B in the same way of the ones introduced in the proof of channel coding (using the same random binning). We only replace $p(y|x)$ by $p(y,z|x)$. That is, the pmf's $P$ in \eqref{eq:pmf0} and $\hat{P}$ in \eqref{eq:pmf2} are replace by
\begin{align}
P(x^n,y^n,z^n,m,f,\hat{x}^n)&=p(x^n,y^n,z^n)P(m,f|x^n)P^{SW}(\hat{x}^n|y^n,f)\n
&= P(m,f)P(x^n|m,f)p(y^n,z^n|x^n)P^{SW}(\hat{x}^n|y^n,f).\label{eq:s-pmf0}\\
\hat{P}(x^n,y^n,z^n,m,f,\hat{x}^n)&=p^U(f)p^U(m)P(x^n|m,f)p(y^n,z^n|x^n)P^{SW}(\hat{x}^n|y^n,f).\label{eq:s-pmf2}
\end{align}
The same argument used in the part (2a) of the proof of channel coding shows that if $\tR+R<H(X)$ the source coding side of the problem (Protocol A) and the main problem (Protocol B) are equivalent, that is $\hat{P}(m,f,x^n,y^n,z^n,\hat{x}^n)\apx{}P(m,f,x^n,y^n,z^n,\hat{x}^n)$. The same argument used in the part (2b) of the proof of channel coding guaranties the reliability, whenever $R>H(X|Y)$. That is,
\begin{align}
\hat{P}(m,f,x^n,y^n,z^n,\hat{x}^n)\apx{}P(m,f,x^n,y^n,z^n)\ind\{\hat{x}^n=x^n\}.
\end{align}
Using part one of lemma \ref{le:total}, we can introduce $\hat{m}$ in the above equation, because random variable $\hat{M}$ is a function of $\hat{X}^n$.
\begin{align}
\hat{P}(m,f,x^n,y^n,z^n,\hat{x}^n,\hat{m})\apx{}P(m,f,x^n,y^n,z^n)\ind\{\hat{x}^n=x^n\}\ind\{\mathsf{M}(\hx^n)=\hat{m}\},\label{eq:s-pmf01}
\end{align}
where $\mathsf{M}(x^n)$ is the bin assigned to $x^n$. Using the fact that $m=\mathsf{M}(x^n)$ and $\hat{m}=\mathsf{M}(\hx^n)$ are the outputs of the same function, one can easily show that the marginal pmf of the RHS of \eqref{eq:s-pmf01} for the random variables $(M,F,Z^n,\hat{M})$ factorizes as $P(m,f,z^n)\ind\{\hat{m}=m\}$. Using \eqref{eq:s-pmf01} and the first part of Lemma \ref{le:total}, we get

\be
 \hat{P}(m,f,z^n,\hat{m})\apx{}P(m,f,z^n)\ind\{\hat{m}=m\}.\label{eq:s-pmf1}
\ee

We now add a third sub-part to the second part of the proof of channel coding which guarantees secrecy.

\emph{Part (2c) of the proof: Sufficient conditions that make the protocols secure}: We must take care of independence of $M$, and $(Z^n,F)$ consisting of the the wiretapper's output and the shared randomness.  Consider the random variables of protocol A. Substituting $X_1=X, Z=Z$ in Theorem \ref{thm:re} implies that $M$ is nearly independent of $(Z^n,F)$ if
\begin{align}
R+\tR&<H(X|Z).\label{eq:s-S2sec}
\end{align}
In other words, the above constraints imply that
 \begin{align}P(z^n,f,m)\apx{}p(z^n)p^U(f)p^U(m).\label{eq:s-Spmf1.875}\end{align}
  Also observe that the pmf $P(z^n)$ is equal to i.i.d.\ pmf $p(z^n)$ in Protocol A.

Using equations \eqref{eq:s-pmf1}, \eqref{eq:s-Spmf1.875} and the third part of Lemma \ref{le:total} we have
\begin{align}
\hat{P}(m,f,z^n,\hat{m})\apx{}p(z^n)p^U(f)p^U(m)
\ind\{\hat{m}=m\}.\label{eq:s-Spmf1.9}
\end{align}

\emph{Part (3) of the proof: Eliminating the shared randomness: }
 In the protocol we assumed that the transmitter, the receivers and the wiretapper have access to shared randomness $F$ which is not present in the model. Nevertheless, we show that the transmitter and the receivers can agree on an instance $f$ of $F$.  Using Definition \ref{def:1}, equation \eqref{eq:s-Spmf1.9} guarantees existence of  a fixed binning with the corresponding pmf $p$ such that if we replace $P$ with $p$ in \eqref{eq:s-pmf2} and denote the resulting pmf with $\hat{p}$, then
 \begin{align*}
 \hat{p}(m,f,z^n,\hat{m})\apx{}p(z^n)p^U(f)p^U(m)
\ind\{\hat{m}=m\}.
 \end{align*}
 Now, the second part of Lemma \ref{le:total} shows that there exists an instance $f$ such that
 \begin{align*}
 \hat{p}(m,z^n,\hat{m}|f)\apx{}p(z^n)p^U(m)
\ind\{\hat{m}=m\}.
\end{align*}
This approximation yields both the secrecy and the reliability requirements as follows:
\begin{itemize}
\item \emph{Reliability:} Using the second item in part 1 of Lemma \ref{le:total} we conclude that
\[
\hat{p}(m, \hat{m}|f)\apx{}p^U(m)\ind\{\hat{m}=m\},
\]
 which is equivalent to $\hat{p}\left(\hat{M}\neq M|f\right)\rightarrow 0$. 
\item \emph{Secrecy:} Using the second item in part 1 of Lemma \ref{le:total} we conclude that
 $\hat{p}(z^n, m| f)\apx{}p^U(m)p(z^n)$.
\end{itemize}
Finally, identifying $p(x^n|m,f)$  as the encoder and the Slepian-Wolf decoder results in reliable and secure encoder-decoder.

\subsection{Lossy source coding}
\emph{Problem definition:} Consider the problem of lossy compression of a source within a desired distortion. In this setting, there is an i.i.d.\ source $X^n$ distributed according to $p(x)$, an (stochastic) encoder mapping $\mx^n$ to $M\in[1:2^{nR}]$, a decoder that reconstructs a lossy version of $X^n$ (namely $Y^n$) and a distortion measure $d:\mx\times\my\rightarrow [0,d_{max}]$. A rate $R$ is said to be achievable at the distortion $D$, if $\mathbb{E}(d(X^n,Y^n))\le D+\epsilon_n$, where $\epsilon_n\rightarrow 0$ and $d(X^n,Y^n)$ is the average per letter distortion.

\emph{Statement:} Here we wish to reprove the known result on the achievability of the rate $R>I(X;Y)$ for any $p(x,y)$ where $\e d(X,Y)<D$.

\emph{Proof:}
An overview of the proof is given in Fig. \ref{fig:RD}. Take some arbitrary $p(x,y)$ where $\e d(X,Y)<D$. 

\emph{Part (1) of the proof:}
We define two protocols each of which induces a joint distribution on random variables that are defined during the protocol. Fig. \ref{fig:RD} illustrates how the source coding side of problem can be used to prove the main problem.
\begin{figure}
\centering\includegraphics[width=.45\linewidth]{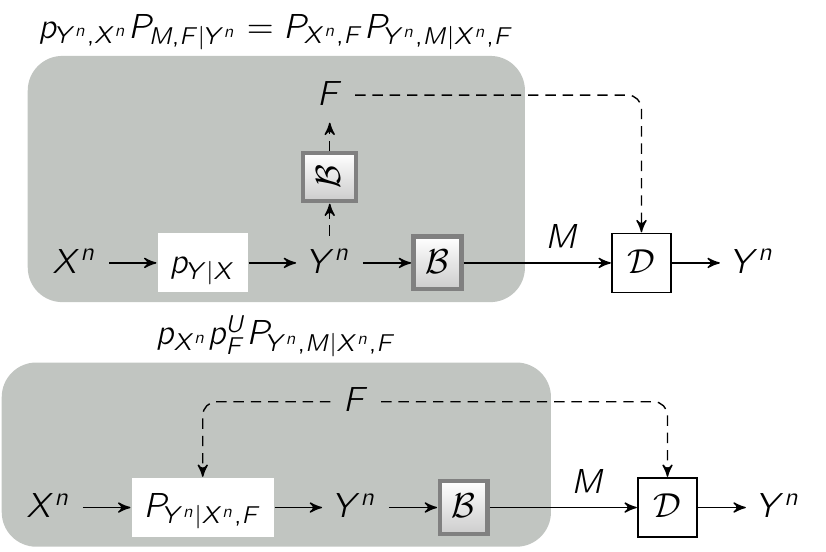}
        \caption{{\small (Top) Source coding side of the problem (Protocol A). We pass i.i.d.\ source $X^n$  through a virtual DMC $p_{Y|X}$ to get an i.i.d.\ sequence $Y^n$ with the desired distortion less than $D$. This is because $\e_{p(x^n,y^n)}d(X^n,Y^n)=\e d(X,Y)<D$. So we aim to describe the \emph{good} sequence $Y^n$ to decoder. We describe $Y^n$ through two random bins $M$ and $F$ at rates $R$ and $\tilde{R}$, where $M$ will serve as the message for the receiver in the main problem, while $F$ will serve as the shared randomness. We use SW decoder for decoding. As long as $R+\tR>H(Y)$, decoder can reliably decode the good sequence $Y^n$ with the desired distortion.  (Bottom) Coding for the lossy source coding problem assisted with the shared randomness (Protocol B). We pass the source $X^n$ and the shared randomness $F$ through the reverse encoder to get a sequence $Y^n$. If the joint distribution of $X^n$ and $F$ is equal to that of protocol A, then the two protocols are equivalent, meaning that $Y^n$ is a sequence with the desired distortion. Since 
$\e[d(X^n,Y^n)]=\e_F\e[d(X^n,Y^n)|F]<D$,  the parties can find a good instance $F=f$ of shared randomness without disturbing the distortion criterion, i.e. $\e[d(X^n,Y^n)|F=f]<D$. However conditioned on $f$, the distribution of the source can be disturbed (it is not equal to $p_{X^n}$). To get rid of this bad effect, we assume that the shared randomness and the source are nearly independent. So to get the equivalence between the two protocols, we need to impose constraint implying $P_{X^n,F}\apx{}p_{X^n}p^U_F$. This is holds as long as $\tR<H(Y|X)$. }    }\label{fig:RD}
\end{figure}
\emph{Protocol A (Source coding side of the problem).}
Let $(X^n,Y^n)$ be i.i.d.\ and distributed according to $p(x,y)$.

\underline{Random Binning}: Consider the following random binning: to each sequence $y^n$, assign uniformly and independently two bin indices $m\in[1:2^{nR}]$ and $f\in[1:2^{n\tR}]$. Further, we use a Slepian-Wolf decoder to recover $y^n$ from $(m,f)$. We denote the output of the decoder by  $\hat{y}^n$. The rate constraint for the success of the decoder will be discussed later, although this decoder can be conceived even when there is no guarantee of success.

The random pmf induced by the random binning, denoted by $P$, can be expressed as follows:
\begin{align}
P(x^n,y^n,m,f,\hat{y}^n)&=p(x^n,y^n)P(f|y^n)P(m|y^n)P^{SW}(\hat{y}^n|m,f)\n
&=P(f,x^n,y^n)P(m|y^n)P^{SW}(\hat{y}^n|m,f)\n
&= P(f,x^n)P(y^n|x^n,f)P(m|y^n)P^{SW}(\hat{y}^n|m,f).\label{eq:pmfL0}
\end{align}
The relation among random variables and random bin assignments is depicted in the top diagram of Fig. \ref{fig:RD}.

\emph{Protocol B (coding for the main problem assisted with the shared randomness). } In this protocol we assume that the transmitter and the receiver have access to the
shared randomness $F$ where $F$ is uniformly distributed over $[1:2^{n\tR}]$.
Then, the protocol proceeds as follows (see also the bottom diagram of Fig. \ref{fig:RD} demonstrating the protocol B):
\begin{itemize}
\item The transmitter generates $Y^n$ according to the conditional pmf $P(y^n|x^n,f)$ of protocol A.
\item Next, knowing $y^n$, the transmitter sends $m$ which is the bin index of $y^n$. Random variable $M$ is generated according to the conditional pmf $P(m|y^n)$ of protocol A.
\item At the final stage, the receiver, knowing $(m,f)$ uses the Slepian-Wolf decoder $P^{SW}(\hat{y}^n|m,f)$ of protocol A to obtain an estimate of $y^n$.
\end{itemize}
The random pmf induced by the protocol, denoted by $\hat{P}$, factors as
\begin{align}
\hat{P}(x^n,y^n,m,f,\hat{y}^n)=p^U(f)p(x^n)P(y^n|x^n,f)P(m|y^n)P^{SW}(\hat{y}^n|m,f)\label{eq:pmfL2}
\end{align}

\emph{Part (2a) of the proof: Sufficient conditions that make the induced pmfs approximately the same}: To find the constraints that imply that the pmf $\hat{P}$ is close to the pmf $P$ in total variation distance,
we start with $P$ and make it close to $\hat{P}$ in a few steps. The first step is to observe that $f$ is a bin index of $y^n$ in protocol A. Theorem \ref{thm:re} implies that if $\tR<H(Y|X)$
then  $P(f,x^n)\apx{}p^U(f)p(x^n)=\hat{P}(f,x^n)$.
Equations \eqref{eq:pmfL0} and \eqref{eq:pmfL2} imply
\begin{align}
\hat{P}(m,f,x^n,y^n,\hat{y}^n)&\apx{}
P(m,f,x^n,y^n,\hat{y}^n)\label{eqn:Lee1}
\end{align}

\emph{Part (2b) of the proof: Sufficient conditions that make the Slepian-Wolf decoder succeed}: The next step is 
{to see that when the Slepian-Wolf decoder of protocol A can reliably decode the  sequence $Y^n$.}
 Lemma \ref{le:sw} requires imposing the constraint $R+\tR>H(Y)$.
It yields that 
\begin{align}
P(m,f,x^n,y^n,\hat{y}^n)\apx{}P(m,f,x^n,y^n)\ind\{\hat{y}^n=y^n\}.\label{eq:Lpmf1.5}
\end{align}

Using equations \eqref{eqn:Lee1}, \eqref{eq:Lpmf1.5} and the triangle inequality \ref{le:total} we have
\begin{align}
\hat{P}(m,f,x^n,y^n,\hat{y}^n)\apx{}P(m,f,x^n,y^n)\ind\{\hat{y}^n=y^n\}.\label{eq:Lpmf1}
\end{align}

\emph{Part (3) of the proof: Eliminating the shared randomness {$F$}: }
 Using Definition \ref{def:1}, equation \eqref{eq:Lpmf1} guarantees existence of  a fixed binning with the corresponding pmf $p$ such that if we replace $P$ with $p$ in \eqref{eq:pmfL2} and denote the resulting pmf with $\hat{p}$, then $\hat{p}(m,f,x^n,y^n,\hat{y}^n)\apx{}p(m,f,x^n,y^n)\ind\{\hat{y}^n=y^n\}:=\tilde{p}(m,f,x^n,y^n,\hat{y}^n)$. Using the second item of part one of Lemma \ref{le:total}, we have $\hat{p}(x^n,\hat{y}^n)\apx{}\tilde{p}(x^n,\hat{y}^n)=p_{X^nY^n}(x^n,\hat{y}^n)$.



Applying lemma \ref{le:distortion} to $\tilde{p}_{X^n\hY^n}$ and $\hat{p}_{X^nY^n}$ and noting that $\e_{p_{X^nY^n}(x^n,\hy^n)}d(X^n,\hY^n)<D$, we obtain $$\e_{\hat{p}(x^n,\hy^n)}d(X^n,\hY^n)<D,$$ for sufficiently large $n$. Using the law of iterated expectation, we conclude that there exists an $F=f$ such that $\e_{\hat{p}(x^n,\hy^n|f)}d(X^n,\hY^n)<D$.

 Finally, specifying $p(m|x^n,f)$ as the encoder (which is equivalent to generating a random sequence $y^n$ according to $p(y^n|x^n,f)$ and then transmitting the bin index $m$ assigned to $y^n$) and $p^{SW}(\hy^n|m,f)$ as the decoder results in a pair of encoder-decoder obeying the desired distortion.

\subsection{{Distributed channel synthesis}}
One important application of our framework is to prove achievability part of the channel simulation problems, see \cite{me2,me3, farzin13}. In this subsection, we illustrate how our achievability framework can be adopted to prove the achievability part of channel simulation problems over networks. To do this, we apply our framework to re-prove the achievability part of the channel synthesis problem \cite{cuff-trans} as a building block of channel simulation problems. First we give a formal definition of the problem.

\emph{Channel synthesis problem:} In this setting, there are a stochastic encoder, a stochastic decoder, a communication link of limited rate $R_1$  between encoder and decoder, an i.i.d.\ source $X^n$ distributed according to $p_X$ and a common randomness $\omega$ uniformly distributed over a finite set $[1:2^{nR_0}]$ that is independent of the source. Observing the source and the common randomness, the encoder chooses an index $M\in[1:2^{nR_1}]$  and transmits it over the communication link to the decoder. Observing $M$ and the common randomness $\omega$, the decoder produces an output $Y^n$. The goal is to find an encoder-decoder such that the induced distribution on $(X^n,Y^n)$ could not be distinguished from a given  joint i.i.d.\ distribution according to $p_{XY}=p_Xp_{Y|X}$, which can be thought of as  the joint distribution of $(X^n,Y^n)$ when $X^n$ is transmitted over a DMC channel $p_{Y|X}$. A rate pair $(R_0,R_1)$ is said to be achievable, if there exists a sequence of encoder-decoders such that the total variation distance between the induced distribution $p^{(\mathtt{ind})}(x^n,y^n)$ and the i.i.d.\ distribution $p(x^n,y^n)$ vanishes as $n$ goes to infinity, that is
\be
\label{eq:t}
\lim_{n\rightarrow\infty}\tv{p^{(\mathtt{ind})}(x^n,y^n)-\prod_{i=1}^n p(x_{1,i},y_{1,i})}=0.
\ee
\begin{theorem}[{\cite[Theorem II.1]{cuff-trans}}]\label{thm:cuff}
A rate pair $(R_0,R_1)$ is achievable iff there exists a random variable   $U$ such that $X-U-Y$ is a Markov chain, the marginal distribution of $(X,Y)$ is equal to the desired distribution $p_{XY}$ and the following inequalities hold:
\be
\begin{split}
R_1&>I(X;U),\\
R_0+R_1&>I(XY;U).
\end{split}\label{eq:cuff}
\ee

\begin{proof}
\emph {Part (1) of proof}: We define two protocols each of which induces a joint distribution on random variables that are defined during the protocol. Fig. \ref{fig:synth} illustrates how the source coding side of problem can be used to prove the main problem. \\

\emph{Protocol A.} Let $(X^{n},U^n,Y^{n})$ i.i.d.\ and distributed according to $p(x,u,y)$ given in the Theorem \ref{thm:cuff}. Consider the following random binning:
\begin{itemize}
\item To each $u^{n}$, assign three random bin indices $F\in[1:2^{n\tilde{R}}]$, $m\in[1:2^{nR}]$ and $\omega\in[1:2^{nR_0}]$.
\item We use Slepian-Wolf decoder to estimate $\hat{u}^{n}$  from $(\omega,f,m)$.
\end{itemize}
The rate constraints for the success of these decoders will be imposed
later, although these decoders can be conceived even when there is no guarantee of success. The random pmf induced by the random binning, denoted by $P$, can be expressed as follows:
\begin{align}
P(x^{n},u^{n},y^{n},f,m,\omega,\hu^n)&=p(x^{n}u^{n})p(y^{n}|u^{n})P(f,m,\omega|u^{n})P^{SW}(\hat{u}^{n}|f,m,\omega)
\n
&=p(x^{n})P(u^{n},f,m,\omega|x^{n})p(y^n|u^n)P^{SW}(\hat{u}^{n}|f,m,\omega)\n
&=P(x^n,f,\omega)P(u^{n}|f,\omega,x^{n})P(m|u^n)P^{SW}(\hat{u}^{n}|f,m,\omega)p(y^n|u^n)
\end{align}
The relation among random variables and random bin assignments is depicted in the left diagram of Fig. \ref{fig:synth}.

\begin{figure}
\centering\includegraphics[width=\linewidth]{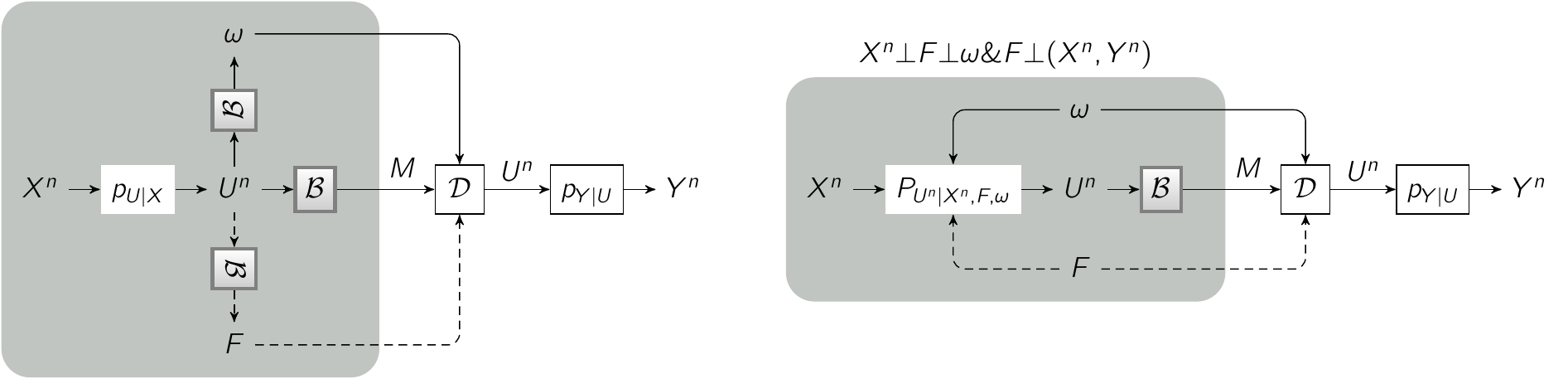}
        \caption{\small (Left) Source coding side of the problem (Protocol A). We pass the i.i.d.\ source $X^n$  through a virtual DMC $p_{U|X}$ to get an i.i.d.\ sequence $U^n$ and then we pass $U^n$ a virtual DMC $p_{Y|U}$ to get an i.i.d.\ sequence $Y^n$. The sequences $X^n,U^n$ and $Y^n$  are jointly i.i.d.\ and distributed according to $p(x,u,y)=p(x,u)p(y|u)$. In particular, $(X^n,Y^n)$ have the desired i.i.d.\ distribution $p(x^n,y^n)$. To produce $Y^n$, decoder only needs to access to $U^n$, so we aim to describe the sequence $U^n$ to the decoder. We describe $U^n$ through three random bins $M$, $\omega$ and $F$, where $M$ will serve as the message from the encoder to the decoder in the main problem, while $\omega$ and $F$ will serve as the common randomness and the extra shared randomness. We use SW decoder for decoding. As long as the SW constraint \eqref{eq:cst-c2} is satisfied, decoder can reliably decode the  sequence $U^n$ and then produces $Y^n$ using the virtual DMC $p_{Y|U}$.  (Right) Coding for the channel synthesis problem assisted with the extra shared randomness (Protocol B). We pass the source $X^n$, the common randomness $\omega$ and the extra shared randomness $F$ through the reverse encoder to get a sequence $U^n$. Similar to the lossy source coding problem, one needs to have mutual independence among the source, the common randomness and the shared randomnesses. This is because the source and the common randomness are independent by the problem definition and at the last step of proof, we must eliminate the shared randomness by conditioning on an instance of it, without disturbing the joint distribution of the source and common randomness. To get the equivalence between the two protocols, we need to impose a constraint implying the desired mutual independence.  This holds as long as \eqref{eq:cst-c1} is satisfied. Finally we must eliminate the shared randomness $F$ by conditioning on an instance of the shared randomness, without disturbing the desired joint i.i.d.\ distribution of $(X^n,Y^n)$. To do this, it suffices to have $F\bot(X^n,Y^n)$. Resorting to the source coding side, we observe $F\bot(X^n,Y^n)$ holds whenever \eqref{eq:cst-c3} is satisfied.   }\label{fig:synth}
\end{figure}

\emph{Protocol B}. In this protocol we assume that the nodes have access to the extra common randomness $F$ where $F$ is distributed uniformly over the sets $[1:2^{n\tilde{R}}]$. Now we use the following protocol (see also the right diagram of Fig. \ref{fig:synth} demonstrating the protocol B):
\begin{itemize}
 \item At the first stage, encoder knowing $(f,\omega,x^{n})$ generate a sequence $u^{n}$ according to the pmf $P(u^{n}|f,\omega,x^{n})$, and sends the bin index of $m(u^{n})$ of protocol A to the decoder.
 \item In the second stage, knowing $(f,\omega,m)$, decoder uses the Slepian-Wolf decoder $P^{SW}(\hat{u}^{n}|f,m,\omega)$ to obtain an estimate of $u^{n}$. Then it generates $y^{n}$ according to $ p(y^{n}|\hat{u}^{n})$ {(more precisely, $p_{Y^n|U^n}(y^{n}|\hat{u}^{n})$)}.
\end{itemize}
The random pmf induced by the protocol, denoted by $\hat{P}$, can be written as follows:

\begin{align}
\hat{P}(x^{n},u^{n},y^{n},f,m,\omega,\hu^n)=p(x^{n})p^U(f)p^U(\omega)P(u^{n}|f,\omega,x^{n})P(m|u^n)P^{SW}(\hat{u}^{n}|f,m,\omega)p(y^n|\hat{u}^n).\label{eq:cst0}
\end{align}

\emph{Part 2 of the proof: Sufficient conditions that make the induced pmfs approximately the same:}
To find the constraints that imply that the pmf $\hat{P}$ is close to the pmf $P$ in total variation distance, we start with $P$ and make it close to $\hat{P}$ in a few steps. The first step is to observe that $(f,\omega)$  is the bin index of $u^n$.  Theorem \ref{thm:re} implies that if
\begin{align}
R_0+\tilde{R}&<H(U|X),\label{eq:cst-c1}
\end{align}
then  $P(x^n,f,\omega)\apx{}p(x^{n})p^U(f)p^U(\omega)=\hat{P}(x^n,f,\omega)$. This implies
\begin{align}
P(x^{n},u^{n},f,m,\omega,\hu^n)\apx{}\hat{P}(x^{n},u^{n},f,m,\omega,\hu^n).\label{eq:cst1}
\end{align}

The next step is
{to see that when the Slepian-Wolf decoder of protocol A can reliably decode the  sequence $U^n$.} Lemma \ref{le:sw} requires
imposing the following constraint:
\begin{align}
\tilde{R}+R_{0}+R_{1}>H(U).\label{eq:cst-c2}
\end{align}
It yields that
\begin{align}
P(x^{n},u^{n},f,m,\omega,\hu^n)\apx{}P(x^{n},u^{n},f,m,\omega)\ind\{\hu^n=u^n\}.
\end{align}
This besides  \eqref{eq:cst1} and the first part and the third part of Lemma \ref{le:total} give
\begin{align}
\hat{P}(x^{n},u^{n},y^n,f,m,\omega,\hu^n)&=\hat{P}(x^{n},u^{n},f,m,\omega,\hu^n)p(y^n|\hu^n)\n
                                                                     &\apx{}P(x^{n},u^{n},f,m,\omega)\ind\{\hu^n=u^n\}p(y^n|\hu^n)\n
                                                                     &=P(x^{n},u^{n},f,m,\omega)\ind\{\hu^n=u^n\}p(y^n|u^n)\n
                                                                     &=P(x^{n},u^{n},y^n,f,m,\omega)\ind\{\hu^n=u^n\}.
\end{align}
Using the first of Lemma \ref{le:total} we conclude that
\begin{align}
\hat{P}(f,x^n,y^n)\apx{}P(f,x^n,y^n) .
\end{align}
In particular, the marginal pmf of $(X^n,Y^n)$ of the RHS of this expression is equal to $p(x^n,y^n)$ which is the desired pmf.

\emph{Part(3) of proof:}
In the protocol we assumed that the nodes have access to an external randomness $F$
which is not present in the model. Nevertheless, we can assume that the nodes agree on an instance $f$ of $F$. In this case, the induced pmf $\hat{P}(x^n,y^n)$ changes to the conditional pmf $\hat{P}(x^n,y^n|f)$. But if $F$ is independent of $(X^n,Y^n)$, then the conditional pmf $\hat{P}(x^n,y^n|f)$ is also close to the desired distribution. To
obtain the independence, we again use Theorem \ref{thm:re}. Substituting $T = 1$, $X_1=U$ and
$Z = XY$ in Theorem \ref{thm:re}, asserts that if
\begin{align}
\tilde{R}<H(U|XY),\label{eq:cst-c3}
\end{align}
then $P(x^n,y^n,f)\apx{}p^U(f)p(x^n,y^n)$ implying $\hat{P}(x^n,y^n,f)\apx{}p^U(f)p(x^n,y^n)$. Thus, there exists a fixed binning with the corresponding pmf $\tilde{p}$ such that if we replace $P$ with $\tilde{p}$ in \eqref{eq:cst0} and denote the resulting pmf with $\hat{p}$, then $\hat{p}(x^n,y^n,f)\apx{}p^U(f)p(x^n,y^n)$. Now the second part of Lemma \ref{le:total} shows that there exists an instance $F$ such that $\hat{p}(x^n,y^n|f)\apx{}p(x^n,y^n)$.

Specifying $p(m|x^n,f,\omega)$ as the encoder (which is equivalent to generating a random sequence $u^n$ according to $p(u^n|x^n,f,\omega)$ and then transmitting the bin index $m$ assigned to $u^n$) and $(p^{SW}(\hu^n|m,f,\omega)$ as the decoder results in a pair of encoder-decoder obeying the desired vanishing total variation distance.
 Finally, eliminating $\tilde{R}_1$ from \eqref{eq:cst-c1}, \eqref{eq:cst-c2} and \eqref{eq:cst-c3}  using Fourier-Motzkin elimination (FME) results in the rate region \eqref{eq:cuff}.
\end{proof}
\begin{remark}\label{re:fme}
{
We have applied FME to the constraint \eqref{eq:cst-c1}, \eqref{eq:cst-c2} and \eqref{eq:cst-c3}. However we have the implicit constraint  $\tR\ge0$. Nevertheless, we show that this constraint is redundant. To do this, we show that if $(R_0,R_1,\tR)$ satisfies \eqref{eq:cst-c1}, \eqref{eq:cst-c2} and \eqref{eq:cst-c3} for some r.v. $U$ and $\tR$ (which is not necessarily positive), then there exists a r.v. $\bar{U}$ and $\bar{R}\ge 0$ such that $(R_0,R_1,\bar{R})$ satisfies \eqref{eq:cst-c1}, \eqref{eq:cst-c2} and \eqref{eq:cst-c3} for $\bar{U}$ instead of $U$. If $\tR\ge0$, it is nothing to prove. So we assume $\tR<0$. Let $W$ be a r.v. with entropy $H(W)>|\tR|$. Further assume that $W$ is  independent of all other r.v.'s, i.e. $(U,X,Y)$. Let $\bar{R}=\tR+H(W)$ and $\bar{U}=(U,W)$. It is clear that $\bar{R}>0$. Now it can easily shown that $(R_0,R_1,\bar{R})$ satisfies \eqref{eq:cst-c1}, \eqref{eq:cst-c2} and \eqref{eq:cst-c3} for $\bar{U}$, using the independence of $W$ from all other r.v.'s and the  fact that $(R_0,R_1,\tR)$ satisfies \eqref{eq:cst-c1}, \eqref{eq:cst-c2} and \eqref{eq:cst-c3}.
}
\end{remark}

\end{theorem}

\subsection{Wiretap broadcast channels with strong secrecy criterion}\label{sub:3r}
\emph{Problem definition:}
Consider the problem of secure transmission over a broadcast channel with a wiretapper, $p(y_1,y_2,z|x)$. Here, we wish to securely transmit a common message $m_0\in[1:2^{nR_0}]$ to the receivers $Y_1,Y_2$ and two private messages $m_j\in[1:2^{nR_j}], j=1,2$ to the receivers $Y_j,j=1,2$, respectively, while concealing them from the wiretapper. We use the total variation distance as a measure for analyzing the secrecy. Formally speaking there are,
\begin{itemize}
\item Three messages $M_0,M_1,M_2$ which are mutually independent and uniformly distributed,
\item A stochastic encoder $p^{enc}(x^n|m_{[0:2]})$,
\item Two decoders, where decoder $j$ assigns an estimate $(\hat{m}_{0,j},\hat{m}_j)$ of $(m_0,m_j)$ to each $y_j^n$.
\end{itemize}
\par A rate-tuple $(R_0,R_1,R_2)$ is said to be achievable if $\Pr\{\cup_{j=1,2}(\hat{M}_{0,j},\hat{M}_j)\neq(M_0,M_j)\}\rightarrow 0$ and $M_{[0:2]}$ is nearly independent of the wiretapper output, $Z^n$, that is,
\[
\tv{p(m_{[0:2]},z^n)-p^U_{\mm_{[0:2]}}(m_{[0:2]})p(z^n)}\rightarrow 0,
\]
where, here $p(z^n)$ is the induced pmf on $Z^n$ and is not an i.i.d.\ pmf.

Below, we state an extension of Marton's inner bound for the capacity region of wiretap broadcast channel.
\begin{theorem}\label{thm:marton}
A rate-tuple $(R_0,R_1,R_2)$ is achievable for the secure transmission over the wiretap broadcast channel, if it belongs to the convex hull of
\begin{align}
R_0+R_j&<I(U_0U_j;Y_j|Q)-I(U_0U_j;Z|Q),\ j=1,2.\n
R_0+R_1+R_2&<\min\left\{I(U_0;Y_1|Q),I(U_0;Y_2|Q)\right\}+I(U_1;Y_1|U_0,Q)\n&\quad+I(U_2;Y_2|U_0,Q)-I(U_1;U_2|U_0,Q)-I(U_{[0:2]};Z|Q)\n
2R_0+R_1+R_2&<I(U_0U_1;Y_1|Q)-I(U_0U_1;Z|Q)+I(U_0U_2;Y_2|Q)\n&\quad-I(U_0U_2;Z|Q)-I(U_1;U_2|U_0,Z,Q)\label{eq:bcsec}
\end{align}
where $Q,U_{[0:2]}-X-(Y_1,Y_2,Z)$ forms a Markov chain.
\end{theorem}
\begin{proof} For simplicity we consider the case where the time-sharing r.v. $Q$ is a constant random variable. One can incorporate this into our proof by generating its i.i.d.\ copies, and sharing it among all parties and conditioning everything on it.

Take some arbitrary $p(u_{[0:2]},x)p(y_1,y_2,z|x)$. 

\emph{Part (1) of the proof:}
We define two protocols each of which induces a joint distribution on random variables that are defined during the protocol.

\emph{Protocol A. }
Let $(U_{[0:2]}^n,X^n,Y_1^n,Y_2^n,Z^n)$ be i.i.d.\ and distributed according to $p(u_{[0:2]},x,y_1,y_2,z)$. Fig. \ref{fig:marton} illustrates how the source coding side of problem can be used to prove the main problem, for the case of original broadcast channel without common message.

\underline{Random Binning}: Consider the following random binning:
\begin{itemize}
\item To each $u_0^n$ assign uniformly and independently two random bin indices $m_0\in[1:2^{nR_0}]$ and $f_0\in[1:2^{n\tR_0}]$,
\item For $j=1,2$, to each pair $(u^n_0,u_j^n)$ assign uniformly and independently two random bin indices $m_j\in[1:2^{nR_j}]$ and $f_j\in[1:2^{n\tR_j}]$.
\item We use a Slepian-Wolf decoder to recover $(\hat{u}_{0,1}^n,\hat{u}_1^n)$ from $(y_1^n,f_0,f_1)$, and another Slepian-Wolf decoder to recover $\hat{u}_{0,2}^n,\hat{u}_2^n$ from $(y_2^n,f_0,f_2)$.  Note that we denote the two estimates of $u_0^n$ by the two receivers with $\hat{u}_{0,1}^n$ and $\hat{u}_{0,2}^n$. The rate constraints for the success of these decoders will be imposed later, although these decoders can be conceived even when there is no guarantee of success.
\item {Upon obtaining the  pair $(\hat{u}_{0,j}^n,\hat{u}_j^n)$, decoder $j=1,2$ declares the bin indices $\hat{m_{0,j}}=\mathsf{M}_0(\hat{u}_{0,j}^n)$ and $\hat{m}_j=\mathsf{M}_j(\hat{u}_{0,j}^n,\hat{u}_j^n)$ (assigned to $\hat{u}_{0,j}^n$ and $(\hat{u}_{0,j}^n,\hat{u}_j^n)$, respectively) as the estimate of the  pair $(m_0,m_j)$.}

\end{itemize}
The random pmf induced by the random binning, denoted by $P$, can be expressed as follows:
\begin{align}
P(u_{[0:2]}^n, y_1^n,y_2^n,z^n, m_{[0:2]}, f_{[0:2]}, \hat{u}_{0,1}^n,\hat{u}_1^n,\hat{u}_{0,2}^n,\hat{u}_2^n)&=p(u_{[0:2]}^n, y_1^n,y_2^n,z^n)P(f_{[0:2]}|u_{[0:2]}^n)P(m_{[0:2]}|u_{[0:2]}^n)\n&\qquad\times P^{SW}(\hat{u}_{0,1}^n,\hat{u}_1^n|y_1^n,f_0,f_1)
P^{SW}(\hat{u}_{0,2}^n,\hat{u}_2^n|y_2^n,f_0,f_2)
\n
&=P(f_{[0:2]},m_{[0:2]},u_{[0:2]}^n)p(x^n|u_{[0:2]}^n)p(y_{1}^n,y_{2}^n,z^n|x^n)\n&\qquad\times P^{SW}(\hat{u}_{0,1}^n,\hat{u}_1^n|y_1^n,f_0,f_1) P^{SW}(\hat{u}_{0,2}^n,\hat{u}_2^n|y_2^n,f_0,f_2)\n
&=P(f_{[0:2]},m_{[0:2]})P(u_{[0:2]}^n|f_{[0:2]},m_{[0:2]})p(x^n|u_{[0:2]}^n)p(y_{1}^n,y_{2}^n,z^n|x^n)\n&\qquad\times P^{SW}(\hat{u}_{0,1}^n,\hat{u}_1^n|y_1^n,f_0,f_1) P^{SW}(\hat{u}_{0,2}^n,\hat{u}_2^n|y_2^n,f_0,f_2).\label{eq:Spmf0}
\end{align}
We have ignored $\hat{M}$-r.v.s from the pmf at this stage since they are  functions of other random variables. The relation among random variables and random bin assignments for the broadcast channel are depicted in the left diagram of Fig. \ref{fig:marton}, where for simplicity we assume $U_0$ is a constant random variable.

\begin{figure}
\centering\includegraphics[width=\linewidth]{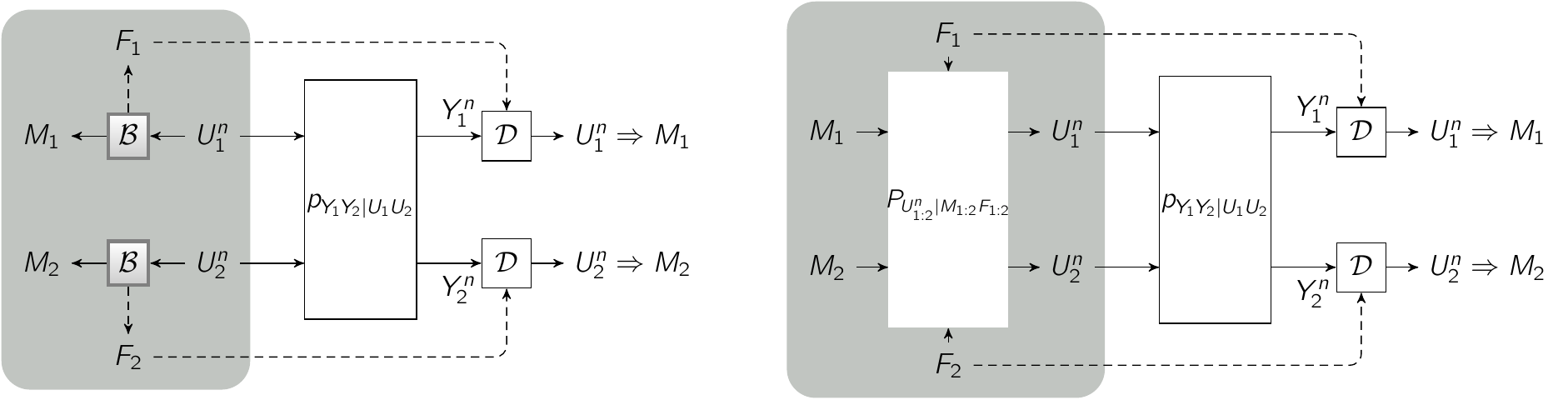}
        \caption{{\small (Left) Source coding side of the broadcast channel (Protocol A). Here $(U_1^n,U_2^n)$ are jointly i.i.d.\ and distributed according to $p_{U_1U_2}$. Then $(U_1^n,U_2^n)$ pass through the virtual DMC $p_{Y_1Y_2|U_1U_2}$ (which is the channel resulted from concatenating two DMC's $p_{X|U_{[1:2]}}$ and $p_{Y_1Y_2|X}$). Thus $(U_{[1:2]}^n,Y_{[1:2]}^n)$ are jointly i.i.d.\  We take the message $M_j$ as a bin index of $U_j^n$. Hence to transmit the message $M_j$ to the decoder $j$, it suffices to describe $U_{j}^n$ for the decoder $j$. We describe $U_j^n$ through a random bin $F_j$ of rate $\tR_j$.  $F_1,F_2$ will serve as the shared randomness. Using SW decoder, we observe that as long as the SW constraints $\tR_j>H(U_j|Y_j), j=1,2$ hold, decoder $j$ can reliably decode $U_{j}^n$ and thus the message $M_j$.  (Right) Coding for the broadcast channel assisted with the shared randomness (Protocol B). Encoder  passes the source $M_{[1:2]}$ and the shared randomness $F_{[1:2]}$ through the reverse encoder to get  sequences $U_1^n$ and $U_2^n$. Similar to the channel coding  problem, one needs to have independence among the shared randomnesses and the messages. This is because at the last step of proof, we must eliminate the shared randomness by conditioning on an instance of it, without disturbing the uniformity of the messages and the independence between them. To get the equivalence between the two protocols, we need to impose constraints implying $M_1\bot M_2\bot F_{[1:2]}$. This is holds if $R_j+\tR_j<H(U_j)$, $j=1,2$ and $R_1+R_2+\tR_1+\tR_2<H(U_{[1:2]})$. The SW constraints and these constraints give the Marton's inner bound without common message.  Adding secrecy is free, one needs only to replace the entropies in these constraints with the conditional entropies $H(U_1|Z),H(U_2|Z)$ and $H(U_{[1:2]}|Z)$ which implies the mutual independence among $M_1,M_2,F_1,F_2$ and $Z^n$, where $Z^n$ is the channel output at the eavesdropper.
   }
        }\label{fig:marton}
\end{figure}

\emph{Protocol B.} In this protocol we assume that the transmitter, the two receivers and the wiretapper have access to the
shared randomness $F_{[0:2]}$ where $F_{[0:2]}$ is uniformly distributed over $[1:2^{n\tR_0}]\times [1:2^{n\tR_1}]\times [1:2^{n\tR_2}]$. Observe that this implies that $F_0$, $F_1$ and $F_2$ are mutually independent.
Then, the protocol proceeds as follows (see also the right diagram of Fig. \ref{fig:marton} demonstrating the protocol B):
\begin{itemize}
\item The messages $M_0$, $M_1$ and $M_2$ are mutually independent of each other and of $F_{[0:2]}$, uniformly distributed over $[1:2^{nR_0}]\times [1:2^{nR_1}]\times [1:2^{nR_2}]$.
\item The transmitter generates $U_{[0:2]}^n$ according to the conditional pmf $P(u_{[0:2]}^n|m_{[0:2]},f_{[0:2]})$ of protocol A.
\item Next, $X^n$ is generated according to the $n$ i.i.d.\ copies of the conditional pmf $p(x|u_{[0:2]})$ (computed from the arbitrary $p(x, u_{[0:2]})$ we chose at the beginning). R.v. $X^n$ is transmitted over the broadcast channel.
\item At the final stage, the receiver $j=1,2$, knowing $(y_j^n,f_{0},f_j)$ uses the Slepian-Wolf decoder \\$P^{SW}(\hat{u}_{0,j}^n,\hat{u}_j^n|y_j^n,f_0,f_j)$ of protocol A to obtain estimates of $u_{0}^n$ and $u_j^n$. We note that while the receiver $j=1,2$ knows $f_0, f_1$ and $f_2$, it uses only $f_{0},f_j$ in its Slepian-Wolf decoder.
\item We use the output of the SW decoder $j=1,2$ for decoding of the messages $(M_0,M_j)$ in the same way of the last step of Protocol A.
\end{itemize}
The random pmf induced by the protocol, denoted by $\hat{P}$, factors as
\begin{align}
\hat{P}(u_{[0:2]}^n, y_1^n,y_2^n,z^n, m_{[0:2]}, f_{[0:2]},& \hat{u}_{0,1}^n,\hat{u}_1^n,\hat{u}_{0,2}^n,\hat{u}_2^n)=p^U(f_{[0:2]})p^U(m_{[0:2]})P(u_{[0:2]}^n|m_{[0:2]},f_{[0:2]})p(x^n|u_{[0:2]}^n)\n&~~~~~ p(y_{[1:2]}^n,z^n|x^n) P^{SW}(\hat{u}_{0,1}^n,\hat{u}_1^n|y_1^n,f_0,f_1)P^{SW}(\hat{u}_{0,2}^n,\hat{u}_2^n|y_2^n,f_0,f_2).\label{eq:Spmf2}
\end{align}
We have ignored $\hat{M}$-r.v.s from the pmf at this stage since they are (random) functions of other random variables.

\emph{Part (2a) of the proof: Sufficient conditions that make the induced pmfs approximately the same}: To find the constraints that imply that the pmf $\hat{P}$ is close to the pmf $P$ in total variation distance,
we start with $P$ and make it close to $\hat{P}$ in a few steps. The first step is to observe that in protocol A, $(m_0,f_0)$ is a bin index of $u_0^n$, $(m_1,f_1)$ is a bin index of $(u_0^n,u_1^n)$ and $(m_2,f_2)$ is a bin index of $(u_0^n, u_2^n)$. Substituting $X_1=U_0, X_2=U_0U_1, X_3=U_0U_2, Z=constant$ in Theorem \ref{thm:re} implies that $M_{[0:2]}$ is nearly independent of $F_{[0:2]}$ if\footnote{{Theorem \ref{thm:re} gives \emph{seven} inequalities. However it can be easily seen that the inequalities associated to the subsets of $\{2,3\}$ are redundant and implied by others. For example, the inequality associated to the subset $\{2\}$ is $\tR_1+R_1<H(U_0U_1)$, which is weaker than the second inequality in \eqref{eq:Ssec}.}}
\begin{align}
R_0+\tR_0&<H(U_0),\n
R_0+R_j+\tR_0+\tR_j&<H(U_0U_j)\ ,j=1,2,\n
R_0+R_1+R_2+\tR_0+\tR_1+\tR_2&<H(U_{[0:2]}).\label{eq:Ssec}
\end{align}
In other words, the above constraints yields $P(f_{[0:2]},m_{[0:2]})\apx{}p^U(f_{[0:2]})p^U(m_{[0:2]})=\hat{P}(f_{[0:2]},m_{[0:2]})$.
Equations \eqref{eq:Spmf0} and \eqref{eq:Spmf2} imply
\begin{align}
\hat{P}(u_{[0:2]}^n, y_1^n,y_2^n,z^n, m_{[0:2]}, f_{[0:2]}, \hat{u}_{0,1}^n,\hat{u}_1^n,\hat{u}_{0,2}^n,\hat{u}_2^n)
&\apx{}
P(u_{[0:2]}^n, y_1^n,y_2^n,z^n, m_{[0:2]}, f_{[0:2]}, \hat{u}_{0,1}^n,\hat{u}_1^n,\hat{u}_{0,2}^n,\hat{u}_2^n).\label{eqn:See1}
\end{align}

\emph{Part (2b) of the proof: Sufficient conditions that make the Slepian-Wolf decoder succeed}: The next step is 
{to see that when the Slepian-Wolf decoder $j, j=1,2$ of protocol A can reliably decode the  pair $(U_0^n,U_j^n)$.}
 Lemma \ref{le:sw} for $X_1=U_0, X_2=U_0U_j, Z=Y_j$ yields that the decoding of $U_0^nU_j^n$ is reliable if,
\begin{align}
\tR_0+\tR_j&>H(U_0U_j|Y_j),\n
\tR_j&>H(U_j|U_0Y_j)~~\mbox{for }j=1,2.\label{eq:S13}
\end{align}
It yields
\begin{align}
P(u_{[0:2]}^n, y_1^n,y_2^n,z^n, m_{[0:2]}, f_{[0:2]}, \hat{u}_{0,1}^n,\hat{u}_1^n,\hat{u}_{0,2}^n,\hat{u}_2^n)&\apx{}P(u_{[0:2]}^n, y_1^n,y_2^n,z^n, m_{[0:2]}, f_{[0:2]})\n&\qquad\times\ind\{\hat{u}_{0,1}^n=\hat{u}_{0,2}^n=u_0^n, \hat{u}_1^n=u_1^n, \hat{u}_2^n=u_2^n\}.\label{eq:Spmf1.5}
\end{align}

Using equations \eqref{eqn:See1}, \eqref{eq:Spmf1.5} and the triangle inequality we have
\begin{align}
\hat{P}(u_{[0:2]}^n, y_1^n,y_2^n,z^n, m_{[0:2]}, f_{[0:2]}, \hat{u}_{0,1}^n,\hat{u}_1^n,\hat{u}_{0,2}^n,\hat{u}_2^n)&\apx{}P(u_{[0:2]}^n, y_1^n,y_2^n,z^n, m_{[0:2]}, f_{[0:2]})\n&\qquad\times\ind\{\hat{u}_{0,1}^n=\hat{u}_{0,2}^n=u_0^n, \hat{u}_1^n=u_1^n, \hat{u}_2^n=u_2^n\}.\label{eq:Spmf1}
\end{align}
Using the first part of Lemma \ref{le:total} we have
\begin{align}
\hat{P}(u_{[0:2]}^n, z^n, m_{[0:2]}, f_{[0:2]}, \hat{u}_{0,1}^n,\hat{u}_1^n,\hat{u}_{0,2}^n,\hat{u}_2^n)&\apx{}P(u_{[0:2]}^n, z^n, m_{[0:2]}, f_{[0:2]})\n&\qquad\times\ind\{\hat{u}_{0,1}^n=\hat{u}_{0,2}^n=u_0^n, \hat{u}_1^n=u_1^n, \hat{u}_2^n=u_2^n\}.\label{eq:Spmf1.75}
\end{align}
Using part one of lemma \ref{le:total}, we can introduce $(\hat{m}_{0,1},\hat{m}_{0,2},\hat{m}_1,\hat{m}_2)$ in the above equation, because these random variables are functions of other random variables.
\begin{align}
\hat{P}(m_{[0:2]},f_{[0:2]},u_{[0:2]}^n,z^n,\hat{u}_{0,1}^n,\hat{u}_1^n,&\hat{u}_{0,2}^n,\hat{u}_2^n,\hat{m}_{0,1},\hat{m}_1,\hat{m}_{0,2},\hat{m})\apx{}P(m_{[0:2]},f_{[0:2]},u_{[0:2]}^n,z^n)\n
&~~~\times\ind\{\hat{u}_{0,1}^n=\hat{u}_{0,2}^n=u_0^n, \hat{u}_1^n=u_1^n, \hat{u}_2^n=u_2^n\}
\n&\times\ind\{\mathsf{M}_0(\hat{u}_{0,1}^n)=\hat{m}_{0,1},\mathsf{M}_0(\hat{u}_{0,2}^n)=\hat{m}_{0,2}, \mathsf{M}_1(\hat{u}^n_{0,1},\hat{u}_1^n)=\hat{m}_1,\mathsf{M}_2(\hat{u}^n_{0,2},\hat{u}_2^n)=\hat{m}_2\},\n
&=P(m_{[0:2]},f_{[0:2]},u_{[0:2]}^n,z^n)\ind\{\hat{u}_{0,1}^n=\hat{u}_{0,2}^n=u_0^n, \hat{u}_1^n=u_1^n, \hat{u}_2^n=u_2^n\}
\n&~\times\ind\{\hat{m}_{0,1}=\hat{m}_{0,2}=m_0, \hat{m}_1=m_1,\hat{m}_2=m_2\},
\label{eq:sbc-pmf01}
\end{align}
where we use $\mathsf{M}_0(u_0^n)$, $\mathsf{M}_1(u_0^n,u_1^n)$ and $\mathsf{M}_2(u_0^n,u_2^n)$ to denote the bins assigned to $u_0^n$, $(u_0^n,u_1^n)$ and $(u_0^n,u_2^n)$,  respectively. It can can be easily seen  that the marginal pmf of the RHS of \eqref{eq:sbc-pmf01} for the random variables \\$(M_{[0:2]},F_{[0:2]},Z^n,\hat{M}_{0,1},\hat{M}_{0,2},\hat{M}_1,\hat{M}_2)$ factorizes as $P(m_{[0:2]},f_{[0:2]},z^n)\ind\{\hat{m}_{0,1}=\hat{m}_{0,2}=m_0, \hat{m}_1=m_1,\hat{m}_2=m_2\}$. Using \eqref{eq:sbc-pmf01} and the first part of Lemma \ref{le:total}, we get

\be
 \hat{P}(m_{[0:2]},f_{[0:2]},z^n,\hat{m}_{0,1},\hat{m}_1,\hat{m}_{0,2},\hat{m}_2)\apx{}P(m_{[0:2]},f_{[0:2]},z^n)\ind\{\hat{m}_{0,1}=\hat{m}_{0,2}=m_0, \hat{m}_1=m_1,\hat{m}_2=m_2\}.\label{eq:sbc-pmf1}
\ee

{Before we consider the secrecy part of problem, we assume that there is no eavesdropper, i.e. $Z=constant$. So we deal with the broadcast channel. It can be easily seen that the constraints \eqref{eq:Ssec} and \eqref{eq:S13} imply Marton's inner bound for the broadcast channel. In the sequel, we show how one can find the extension of Marton's inner bound for wiretap broadcast channel, for free!}

\emph{Part (2c) of the proof: Sufficient conditions that make the protocols secure}: We must take care of independence of $M_{[0:2]}$, and $(Z^n,F_{[0:2]})$ consisting of the the wiretapper's output and the shared randomness.  Consider the random variables of protocol A. Substituting $X_1=U_0, X_2=U_0U_1,X_3=U_0U_2, Z=Z$ in Theorem \ref{thm:re} implies that $M_{[0:2]}$ is nearly independent of $(Z^n,F_{[0:2]})$ if
\begin{align}
R_0+\tR_0&<H(U_0|Z),\n
R_0+R_j+\tR_0+\tR_j&<H(U_0U_j|Z)\ ,j=1,2,\n
R_0+R_1+R_2+\tR_0+\tR_1+\tR_2&<H(U_{[0:2]}|Z).\label{eq:S2sec}
\end{align}
In other words, the above constraints imply that
 \begin{align}P(z^n,f_{[0:2]},m_{[0:2]})\apx{}p(z^n)p^U(f_{[0:2]})p^U(m_{[0:2]}).
 \label{eq:Spmf1.875}\end{align}

 Observe that the pmf $P(z^n)$ is equal to i.i.d.\ pmf $p(z^n)$ in Protocol A. Also  observe that  \eqref{eq:S2sec} is the same as \eqref{eq:Ssec} with one exception, we have conditioned all entropies on the $Z$ due to secrecy requirement. Thus we get secrecy for free.

Using equations \eqref{eq:sbc-pmf1} and \eqref{eq:Spmf1.875} and the third part of Lemma \ref{le:total} we have
\begin{align}
 \hat{P}(m_{[0:2]},f_{[0:2]},z^n,\hat{m}_{0,1},\hat{m}_1,\hat{m}_{0,2},\hat{m}_2)\apx{}p(z^n)p^U(f_{[0:2]})p^U(m_{[0:2]})\ind\{\hat{m}_{0,1}=\hat{m}_{0,2}=m_0, \hat{m}_1=m_1,\hat{m}_2=m_2\}.\label{eq:Spmf1.9}
\end{align}

\emph{Part (3) of the proof: Eliminating the shared randomness {$F_{[0:2]}$ without disturbing the secrecy and reliability requirements}: }
Using Definition \ref{def:1}, equation \eqref{eq:Spmf1.9} guarantees existence of  a fixed binning with the corresponding pmf $p$ such that if we replace $P$ with $p$ in \eqref{eq:Spmf2} and denote the resulting pmf with $\hat{p}$, then
 \begin{align*}
 \hat{p}(m_{[0:2]},f_{[0:2]},z^n,\hat{m}_{0,1},\hat{m}_1,\hat{m}_{0,2},\hat{m}_2)&\apx{}p(z^n)p^U(f_{[0:2]})p^U(m_{[0:2]})
\ind\{\hat{m}_{0,1}=\hat{m}_{0,2}=m_0, \hat{m}_1=m_1,\hat{m}_2=m_2\}.
 \end{align*}
 Now, the second part of Lemma \ref{le:total} shows that there exists an instance $f_{[0:2]}$ such that
 \begin{align*}
 \hat{p}(m_{[0:2]},z^n,\hat{m}_{0,1},\hat{m}_1,\hat{m}_{0,2},\hat{m}_2|f_{[0:2]})&
 \apx{}p(z^n)p^U(m_{[0:2]})
\ind\{\hat{m}_{0,1}=\hat{m}_{0,2}=m_0, \hat{m}_1=m_1,\hat{m}_2=m_2\}.
\end{align*}
{ 
This approximation yields both the secrecy and the reliability requirements as follows:}
\begin{itemize}
\item \emph{Reliability:} Using the second item in part 1 of Lemma \ref{le:total} we conclude that
$$
\hat{p}(m_0,m_j, \hat{m}_{0,j},\hat{m}_j|f_{[0:2]})\apx{}\ind\{\hat{m}_{0,j}=m_0, \hat{m}_j=m_j\},
$$ which is equivalent to $\hat{p}\left((\hat{M}_{0,j}\hat{M}_j)\neq(M_{0},M_j)|f_{[0:2]}\right)\rightarrow 0$. 
\item \emph{Secrecy:} Using the second item in part 1 of Lemma \ref{le:total} we conclude that
 $\hat{p}(z^n, m_{[0:2]}| f_{[0:2]})\apx{}p^U(m_{[0:2]})p(z^n)$.
\end{itemize}

Finally, identifying $p(x^n|m_{[0:2]},f_{[0:2]})$ (which is done via generating $u_{[0:2]}$ first) as the encoder and the Slepian-Wolf decoders results in reliable and secure encoder-decoders.

Applying FME on \eqref{eq:Ssec}, \eqref{eq:S13} and \eqref{eq:S2sec} gives \eqref{eq:bcsec}. Note that the equations of \eqref{eq:Ssec} are completely redundant.
\end{proof}

\begin{remark}
Although in the OSRB-based proof of Theorem \ref{thm:marton}, we did not deal with codebook explicitly but one can find similarity between it and codebook generation for superposition coding and Marton coding as follows.
\begin{itemize}
\item (Superposition) To focus on superposition coding, let $U_2$ be a constant random variable. Inspecting the proof, it is seen that we used a binning for $U_0^n$ and a joint binning for $(U_0^n,U_1^n)$. Observe that conditioned on an instance $(F_0,F_1)=(f_0,f_1)$ of the shared randomness, the inputs of encoder are restricted to those sequence assigned to $(f_0,f_1)$. Thus one can interpret these sequences as codewords of a codebook (although the encoder is not deterministic and there is not a one-to-one map between the messages and the codewords.). Conditioned on $f_0$, we get a codebook $\cc_0$ of sequences $u_0^n$ assigned to $f_0$. Also conditioned on $(f_0,f_1)$, for each $u_0^n\in\cc_0$ we get a codebook $\cc_1(u_0^n)$ of sequences $u_1^n$ which together with $u_0^n$ are \emph{jointly assigned}   to $f_1$. This resembles the superposition codebook generation which uses an inner codebook $\cc_0$ and a set of outer (superimposed) codebook for each sequence of the inner codebook.
\item (Marton)  To focus on Marton coding without common message, let $U_0$ be a constant random variable. Inspecting the proof, it is seen that we used separate random bin assignments for $U_1^n$ and $U_2^n$. Again observe that conditioned on an instance $(F_1,F_2)=(f_1,f_2)$ of the shared randomness, the inputs of encoder are restricted to those sequence assigned to $(f_1,f_2)$. Thus one can interpret these sequences as codewords of a codebook. Conditioned on $f_j, j=1,2$, we get a codebook $\cc_j$ of sequences $u_j^n$ assigned to $f_j$. We note that contrary to the case of superposition, the codebooks $\cc_2$ is not related to $\cc_1$. This resembles the independent  codebook generation of Marton coding.
\end{itemize}
In general, whenever we require superposition coding, we use binning for nested random variables while whenever we require independent codebook generation, we use separate binning.
\end{remark}

\begin{remark}
If one only wants to securely transmit a common message $M_0$ as in \cite{chia} (i.e. $R_1=R_2=0$), 
the second inequality of the \eqref{eq:S13} can be neglected, because it is sufficient to recover $M_0$ through only $U_0^n$; it is not necessary to make sure that we do not make any error in decoding $U_1^n$ and $U_2^n$ {(the efficacy of using $U_1$ and $U_2$ without decoding them has been clarified in  \cite{chia}).} This resembles the idea of indirect decoding of Nair and El Gamal \cite{nair}. Applying FME gives
the following lower bound on $R_0$ which subsumes the lower bound given in \cite{chia} under weak secrecy criterion:
\begin{align}
R_0=\max_{p_{QU_{[0:2]}X}}\min\Big\{&I(U_0U_1;Y_1|Q)-I(U_0U_1;Z|Q),I(U_0U_2;Y_2|Q)-I(U_0U_2;Z|Q),\n&\frac{1}{2}\big(I(U_0U_1;Y_1|Q)-I(U_0U_1;Z|Q)+I(U_0U_2;Y_2|Q)-I(U_0U_2;Z|Q)-I(U_1;U_2|Q,U_0,Z)\big)\Big\}\nonumber
\end{align}
We now compare our lower bound with the one given in \cite{chia}. The lower bound given in \cite{chia} is the maximum of
\begin{align}
\min\Big\{&I(U_0U_1;Y_1|Q)-I(U_0U_1;Z|Q),I(U_0U_2;Y_2|Q)-I(U_0U_2;Z|Q)\big)\Big\}\nonumber
\end{align}
over all $p(q,u_0)p(u_1,u_2,x|u_{0})p(y_1,y_2,z|x)$ where $I(U_1, U_2; Z|U_0) \leq I(U_1; Z|U_0)+I(U_2; Z|U_0)- I(U_1; U_2|U_0)$.
We first note that because of the Markov chain $Q-U_0-U_1U_2XZY_1Y_2$ the above constraint is equivalent with
$I(U_1, U_2; Z|Q, U_0) \leq I(U_1; Z|Q, U_0)+I(U_2; Z|Q, U_0)-I(U_1; U_2|Q, U_0)$. Algebraic manipulation shows that this constraint holds only when $I(U_1;U_2|Q,U_0,Z)=0$. We note that the Markov constraint $Q-U_0-U_1U_2X$ can be dropped. This is because given any $(Q,U_0,U_1,U_2,X)$ where the Markov chain does not hold we can replace $U_0$ with $(U_0,Q)$ and find a new set of random variables where the Markov chain holds and the inner bound expression remains unchanged. To sum this up, the lower bound of \cite{chia} can be rewritten as the maximum of
\begin{align}
\min\Big\{&I(U_0U_1;Y_1|Q)-I(U_0U_1;Z|Q),I(U_0U_2;Y_2|Q)-I(U_0U_2;Z|Q)\big)\Big\}\nonumber
\end{align}
over all $p(q,u_0,u_1,u_2,x)$ where $I(U_1; U_2| Q, U_0, Z)=0$. It is clear that when $I(U_1; U_2| Q, U_0, Z)=0$ our lower bound reduces to this lower bound. In general it may be higher because we are taking the maximum over \emph{all} $p(q,u_{[0:2]},x)$ without any constraints.

\end{remark}

\subsection{Distributed lossy compression}\label{sub:BT}
\emph{Problem definition:} Consider the problem of distributed lossy compression of two correlated sources $X_1$ and $X_2$ source within  desired distortions $D_1$ and $D_2$. In this setting, there are two correlated i.i.d.\ sources $X_1^n$  and $X_2^n$, distributed according to $p(x_1,x_2)$, two (stochastic) encoders mapping $\mx_j^n$ to $M_j\in[1:2^{nR_j}]$ ($j=1,2$), a decoder that reconstructs  lossy versions of $X_j^n,j=1,2$ (namely $\hat{X}_j^n,j=1,2$) and two bounded distortion measures $d_j:\mx_j\times\hat{\mx}_j\rightarrow [0,d_{j,max}]$. A rate pair $(R_1,R_2)$ is said to be achievable at the distortions $(D_1,D_2)$, if $\mathbb{E}(d_j(X_j^n,\hat{X}_j^n))\le D_j+\epsilon_n$, where $\epsilon_n\rightarrow 0$.
\\

\emph{Statement:} Here we wish to reprove the known Berger-Tung inner bound for this problem.

\begin{theorem}[Berger-Tung inner bound]
A rate pair $(R_1,R_2)$ is achievable with distortions $D_1$ and $D_2$ if there exist conditional pmf's $p(u_1|x_1)$ and $p(u_2|x_2)$, and two decoding functions $\hat{x}_1(u_1,u_2)$ and $\hat{x}_2(u_1,u_2)$  such that $\e d_j(X_j,\hat{X}_j)\le D_j$, $j=1,2$ and the following inequalities hold:
\be\label{eq:BT00}
\begin{split}
R_1&>I(X_1;U_1|U_2),\\
R_2&>I(X_2;U_2|U_1),\\
R_1+R_2&>I(X_1X_2;U_1U_2).
\end{split}
\ee
\end{theorem}
\begin{proof}
Take some arbitrary $p(x_1,x_2,u_1,u_2)=p(x_1,x_2)p(u_1|x_1)p(u_2|x_2)$ and functions $\hx_1(u_1,u_2)$ and $\hx_2(u_1,u_2)$ such that $\e d_j(X,\hat{X}_j)< D_j$, $j=1,2$.

\emph{Part (1) of the proof:}
We define two protocols each of which induces a joint distribution on random variables that are defined during the protocol.

\emph{Protocol A. }
Let $(X_1^n,X_2^n,U_1^n,U_2^n)$ be i.i.d.\ and distributed according to $p(x_1,x_2,u_1,u_2)$. Fig. \ref{fig:BT} illustrates how the source coding side of problem can be used to prove the main problem.

\underline{Random Binning}: Consider the following random binning:
\begin{itemize}
\item For $j=1,2$, to each sequence $u_j^n$ assign uniformly and independently a random two bin indices $m_j\in[1:2^{nR_j}]$ and $f_j\in[1:2^{n\tR_j}]$,
\item We use a Slepian-Wolf decoder to recover $\hat{u}_1^n,\hat{u}_2^n$ from $(m_1,m_2,f_1,f_2)$. The rate constraints for the success of this decoder will be imposed later, although this decoder can be conceived even when there is no guarantee of success.
\item Random variables $\hx_j^n,j=1,2$ are created as  functions of $(\hat{u}_1^n,\hat{u}_2^n)$ using the two decoding functions given at the beginning.
\end{itemize}

The random pmf induced by the random binning, denoted by $P$, can be expressed as follows:
\begin{align}
P(x_{[1:2]}^n,u_{[1:2]}^n,m_{[1:2]},f_{[1:2]},\hat{u}_{[1:2]}^n)&=p(x_{[1:2]}^n)p(u_1^n|x_1^n)p(u_2^n|x_2^n)P(m_1,f_1|u_1^n)P(m_2,f_2|u_2^n)P^{SW}(\hat{u}^n_{[1:2]}|m_{[1:2]},f_{[1:2]})\n
&=p(x_{[1:2]}^n)P(u_1^n,f_1|x_1^n)P(u_2^n,f_2|x_2^n)P(m_1|u_1^n)P(m_2|u_2^n)P^{SW}(\hat{u}^n_{[1:2]}|m_{[1:2]},f_{[1:2]})\n
&=p(x_{[1:2]}^n)P(u_1^n,f_1|x_1^n)P(u_2^n,f_2|x_2^n)P(m_1|u_1^n)P(m_2|u_2^n)P^{SW}(\hat{u}^n_{[1:2]}|m_{[1:2]},f_{[1:2]})\n
&= p(x_{[1:2]}^n)P(f_1|x_1^n)P(f_2|x_2^n)P(u_1^n|f_1,x_1^n)P(u_2^n|f_2,x_2^n)\n&~~~~~~~~~~~~~~~~~~~~~~~~~\times P(m_1|u_1^n)P(m_2|u_2^n)P^{SW}(\hat{u}^n_{[1:2]}|m_{[1:2]},f_{[1:2]})\label{eq:pmfBT00}\\
&= P(f_{[1:2]},x_{[1:2]}^n)P(u_1^n|f_1,x_1^n)P(u_2^n|f_2,x_2^n)\n&~~~~~~~~~~~~~~~~~~~~~~~~~\times P(m_1|u_1^n)P(m_2|u_2^n)P^{SW}(\hat{u}^n_{[1:2]}|m_{[1:2]},f_{[1:2]}).\label{eq:pmfBT0}
\end{align}
We have ignored $\hx_1^n$ and $\hx_2^n$ from the pmf at this stage since they are functions of other random variables. They will be introduced later.
The relation among random variables and random bin assignments are depicted in the left diagram of Fig. \ref{fig:BT}.

\begin{figure}
\centering\includegraphics[width=\linewidth]{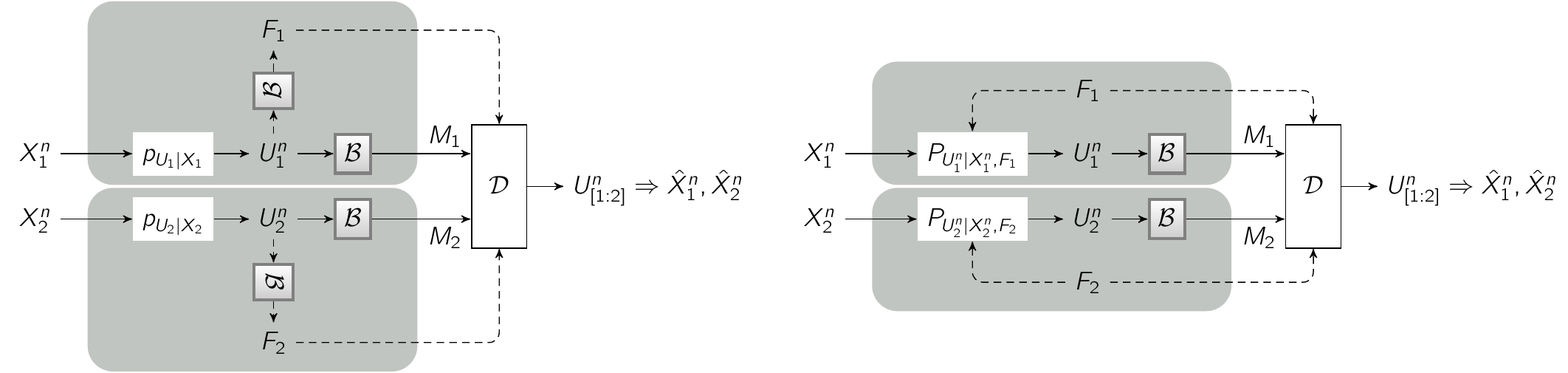}
        \caption{\small (Left) Source coding side of the Berger-Tung problem (Protocol A). Encoder $j$ passes the source $X_j^n$  through a virtual DMC $p_{U_j|X_j}$ to get a sequence $U_j^n$, such that the sequences $\hx_1^n(U_1^n,U_2^n)$ and $\hx_2^n(U_1^n,U_2^n)$ have the desired distortions less than $D_1$ and $D_2$ with $X_1^n$ and $X_2^n$, respectively. So we aim to describe $U_{[1:2]}^n$ for the decoder to enable it to find the \emph{good} sequences $\hat{X}_1^n$ and $\hat{X}_2^n$. We describe $U_j^n$ through two random bins $M_j$ and $F_j$, where $M_j$ will serve as the message of encoder $j$ for the receiver in the main problem, while $F_1,F_2$ will serve as the shared randomness. We use SW decoder for decoding. As long as the SW constraints \eqref{eq:BT15} holds, decoder can reliably decode $U_{[1:2]}^n$ and thus the good sequences $\hat{X}_1^n$ and $\hat{X}_2^n$ with the desired distortions.  (Right) Coding for the Berger-Tung problem assisted with the shared randomness (Protocol B). Encoder $j$ passes the source $X_j^n$ and the shared randomness $F_j$ through the reverse encoder to get a sequence $U_j^n$. Similar to the lossy source coding problem, one needs to have independence among the shared randomnesses and the sources. This is because at the last step of proof, we must eliminate the shared randomness by conditioning on an instance of it, without disturbing the joint distribution of the sources. To get the equivalence between the two protocols, we need to impose constraints implying $X_1^n\bot F_1$ and $X_2^n\bot F_2$ (we note that this implies $X_{[1:2]}^n\bot F_{[1:2]}$, due to Markov chain $F_1-X_1-X_2-F_2$). This is holds as long as $\tR_j<H(U_j|X_j)$, $j=1,2$.
        }\label{fig:BT}
\end{figure}

\emph{Protocol B (coding for the main problem assisted with the shared randomness). } In this protocol we assume that the transmitter and the receiver have access to the
shared randomness $F_{[1:2]}$ where $F_{[1:2]}$ is uniformly distributed over $[1:2^{n\tR_1}]\times[1:2^{n\tR_1}]$.
Then, the protocol proceeds as follows (see also the right diagram of Fig. \ref{fig:BT} demonstrating the protocol B):
\begin{itemize}
\item For $j=1,2$, the transmitter $j$ generates $U_j^n$ according to the conditional pmf $P(u_j^n|x_j^n,f_j)$ of protocol A.
\item Next, knowing $u_j^n$, the transmitter $j$ sends $m_j$ which is the bin index of $u_j^n$. Random variable $M_j$ is generated according to the conditional pmf $P(m_j|u_j^n)$ of protocol A.
\item At the final stage, the receiver, knowing $(m_{[1:2]},f_{[1:2]})$ uses the Slepian-Wolf decoder $P^{SW}(\hat{u}_{1:2}^n|m_{[1:2]},f_{[1:2]})$ of protocol A to obtain  estimates of $u_1^n$ and $u_2^n$.
\end{itemize}
The random pmf induced by the protocol, denoted by $\hat{P}$, factors as
\begin{align}
\hat{P}(x_{[1:2]}^n,u_{[1:2]}^n,m_{[1:2]},f_{[1:2]},\hat{u}_{[1:2]}^n)
&= p^U(f_{[1:2]})p(x_{[1:2]}^n)P(u_1^n|f_1,x_1^n)P(u_2^n|f_2,x_2^n)\n&~~~~~~~~~~~~~~~~~~\times P(m_1|u_1^n)P(m_2|u_2^n)P^{SW}(\hat{u}^n_{[1:2]}|m_{[1:2]},f_{[1:2]}).\label{eq:pmfBT2}
\end{align}

\emph{Part (2a) of the proof: Sufficient conditions that make the induced pmfs approximately the same}: To find the constraints that imply that the pmf $\hat{P}$ is close to the pmf $P$ in total variation distance,
we start with $P$ and make it close to $\hat{P}$ in a few steps. The first step is to observe that in protocol A, $f_1$ is a bin index of $u_1^n$ and $f_2$ is a bin index of $u_2^n$. Substituting $X_1=U_1, X_2=U_2, Z=X_{[1:2]}$ in Theorem \ref{thm:re} implies that $X_{[1:2]}^n$ is nearly independent of $F_{[1:2]}$ and that its pmf is equal to $p(x_{[1:2]}^n)$, if
\begin{align}
\tR_1&<H(U_1|X_1X_2)=H(U_1|X_1),\n
\tR_2&<H(U_2|X_1X_2)=H(U_2|X_2),\n
\tR_1+\tR_2&<H(U_1U_2|X_1X_2)=H(U_1|X_1)+H(U_2|X_2),\label{eq:BT14}
\end{align}
where we have used the Markov chain $U_1-X_1-X_2-U_2$ to simplify the inequalities (Observe that the last inequality is redundant. In fact to have the equivalence between the two protocols, one needs the independence between $F_j$ and $X_j$ for $j=1,2$ due to \eqref{eq:pmfBT00}. These independence are guaranteed as long as the first two inequalities are satisfied.). In other words, the above constraints imply that $P(f_{[1:2]},x_{[1:2]}^n)\apx{}p^U(f_{[1:2]})p(x^n_{[1:2]})=\hat{P}(f_{[1:2]},x_{[1:2]}^n)$.
Equations \eqref{eq:pmfBT0} and \eqref{eq:pmfBT2} imply
\begin{align}
\hat{P}(x_{[1:2]}^n,u_{[1:2]}^n,m_{[1:2]},f_{[1:2]},\hat{u}_{[1:2]}^n)&\apx{}
P(x_{[1:2]}^n,u_{[1:2]}^n,m_{[1:2]},f_{[1:2]},\hat{u}_{[1:2]}^n).\label{eq:eeBT1}
\end{align}

\emph{Part (2b) of the proof: Sufficient conditions that make the Slepian-Wolf decoder succeed}: The next step is {to see that when the Slepian-Wolf decoder  of protocol A can reliably decode the  pair $(U_1^n,U_2^n)$.}
Substituting $X_1=U_1, X_2=U_2$ in Lemma \ref{le:sw} yields that the decoding of $U_1^nU_2^n$ is reliable if,
\be
\begin{split}
R_1+\tR_1&>H(U_1|U_2),\\
R_2+\tR_2&>H(U_2|U_1),\\
R_1+R_2+\tR_1+\tR_2&>H(U_1U_2).
\end{split}\label{eq:BT15}
\ee
It yields
\begin{align}
P(x_{[1:2]}^n,u_{[1:2]}^n,m_{[1:2]},f_{[1:2]},\hat{u}_{[1:2]}^n)\apx{}P(x_{[1:2]}^n,u_{[1:2]}^n,m_{[1:2]},f_{[1:2]})\ind\{\hat{u}_{[1:2]}^n=u_{[1:2]}^n\}.\label{eq:BTpmf1.5}
\end{align}

Using equations \eqref{eq:eeBT1}, \eqref{eq:BTpmf1.5} and the triangle inequality we have
\begin{align}
\hat{P}(x_{[1:2]}^n,u_{[1:2]}^n,m_{[1:2]},f_{[1:2]},\hat{u}_{[1:2]}^n)\apx{}P(x_{[1:2]}^n,u_{[1:2]}^n,m_{[1:2]},f_{[1:2]})\ind\{\hat{u}_{[1:2]}^n=u_{[1:2]}^n\}.\label{eq:BTpmf1}
\end{align}

\emph{Part (3) of the proof: Eliminating the shared randomness {$F_{[1:2]}$} without disturbing the desired distortions: }
 Using Definition \ref{def:1}, equation \eqref{eq:BTpmf1} guarantees existence of  a fixed binning with the corresponding pmf $p$ such that if we replace $P$ with $p$ in \eqref{eq:pmfBT2} and denote the resulting pmf with $\hat{p}$, then
 \begin{align*}
\hat{p}(x_{[1:2]}^n,u_{[1:2]}^n,m_{[1:2]},f_{[1:2]},\hat{u}_{[1:2]}^n)&\apx{}p(x_{[1:2]}^n,u_{[1:2]}^n,m_{[1:2]},f_{[1:2]})\ind\{\hat{u}_{[1:2]}^n=u_{[1:2]}^n\}\n&:=\tilde{p}(x_{[1:2]}^n,u_{[1:2]}^n,m_{[1:2]},f_{[1:2]},\hat{u}_{[1:2]}^n).
 \end{align*}
Using part one of lemma \ref{le:total} we can introduce $\hx_1$ and $\hx_2$ in the above equation. Random variable $\hx_1^n$ was a function of $\hat{u}_{[1:2]}^n$ and $\hx_2^n$ was a function of $\hat{u}_{[1:2]}^n$.
 \begin{align*}
 \hat{p}(x_{[1:2]}^n,u_{[1:2]}^n,m_{[1:2]},f_{[1:2]},\hat{u}_{[1:2]}^n,\hx_{[1:2]}^n)&\apx{}p(x_{[1:2]}^n,u_{[1:2]}^n,m_{[1:2]},f_{[1:2]})\ind\{\hat{u}_{[1:2]}^n=u_{[1:2]}^n\} p(\hx_{[1:2]}^n|\hu_{1:2}^n)
\n
 &~~:=\tilde{p}(x_{[1:2]}^n,u_{[1:2]}^n,m_{[1:2]},f_{[1:2]},\hat{u}_{[1:2]}^n,\hx_{[1:2]}^n).
 \end{align*}
  Note that because of the indicator function terms in $\tilde{p}$, $\tilde{p}(X_1^n, X_2^n,\hat{X}_1^n, \hat{X}_2^n)$ is an i.i.d.\ marginal distribution according to the pmf that we started with at the beginning. Thus, under the probability measure $\tilde{p}$ the distortion constraints $\e d_j(X_j^n,\hat{X}^n_j)< D_j$, $j=1,2$ are satisfied. Using the first part of lemma \ref{le:total} we can drop all the random variables except $x^n_1,x_2^n, \hx_1^n, \hx_2^n$ to get:
  \begin{align}
 \hat{p}(x^n_1,x_2^n,\hx_1^n,\hx_2^n)&\apx{\epsilon_n}\tilde{p}(x^n_1,x_2^n,\hx_1^n,\hx_2^n),\label{eq:BT}
 \end{align}
 for some vanishing sequence $\epsilon_n$.
Remember that the pmf $\hat{p}$ associated to Protocol B which was appropriate for coding.

Unlike the previous case of lossy source coding where we used the law of iterated expectation at this stage, we need to use a concentration result since we are dealing with two distortion functions. Since $\tilde{p}(X_1^n, X_2^n, \hat{X}_1^n, \hat{X}_2^n)$ is an i.i.d.\ distribution we can use the weak law of large number (WLLN) to get that
\[d_j(X^n_j,\hat{X}_j^n)\rightarrow \e_{\tilde{p}}d_j(X_j^n,\hat{X}^n_j)<D_j, \quad \mbox{in}\ \tilde{p}_{X_j^n\hat{X}_j^n}.\]
Thus there exists ${\delta}_n\rightarrow 0$ such that
 \[\tilde{p}\{(x^n_1,x_2^n, \hx_1^n,\hx_2^n): d_j(x^n_j,\hx_j^n)<D_j\ ,j=1,2\}\geq 1-{\delta}_n.\]
Using equation \eqref{eq:BT}, the probability of the same set with respect to $\hat{p}$ should be converging to one, that is
 \[\hat{p}\{(x^n_1,x^n_2, \hx_1^n,\hx_2^n): d_j(x^n,\hx_j^n)<D_j\ ,j=1,2\}\geq 1-\epsilon_n-\delta_n.\]
Thus, there exists some $f_{[1:2]}$ such that
 \[\hat{p}(\{(x_1^n,x_2^n, \hx_1^n,\hx_2^n): d_j(x_j^n,\hx_j^n)<D_j\ ,j=1,2\}|F_{[1:2]}=f_{[1:2]})\geq 1-\epsilon_n-\delta_n.\]
This would imply that
 \[\e_{\hat{p}_{X^n_j\hat{X}_j^n|F_{[1:2]}=f_{[1:2]}}}[d_j(X_j^n,\hat{X}_j^n)]<D_j+ (\epsilon_n+\delta_n)d_{j,max},\]
where we have used the fact that the distortion functions are bounded.

  Finally, specifying $p(m_j|x^n_j,f_j)$ as the encoder $j, j=1,2$ (which is equivalent to generating a random sequence $u_j^n$ according to $p(u_j^n|x_j^n,f_j)$ and then transmitting the bin index $m_j$ assigned to $u_j^n$) and ($p^{SW}(\hu_{[1:2]}^n|m_{[1:2]},f_{[1:2]}),\hat{x}_{[1:2]}^n(\hu_{[1:2]}^n)$) as the decoder results in a pair of encoder-decoder obeying the desired distortion.

{To get the region of \eqref{eq:BT00}, it suffices to choose $\tR_j=H(U_j|X_j)-\epsilon$ in inequalities \eqref{eq:BT14} and \eqref{eq:BT15}, where $\epsilon$  is arbitrarily small.}
\end{proof}
\color{black}

\subsection{Lossy coding over broadcast channels}\label{sub:jscc}
\emph{Problem definition:}
Consider the problem of lossy transmission of an i.i.d.\ source $S^n$ distributed according to $p(s)$, over the broadcast channel $p(y_1,y_2|x)$. Here, the sender wishes to communicate the source to the two receivers within desired distortions $(D_1,D_2)$. Formally, there are
\begin{itemize}
\item an encoder that assigns a random sequence $x^n$ to each $s^n$ according to $p^{enc}(x^n|s^n)$,
\item two decoders, where decoder $j=1,2$ assigns an estimate $\hat{s}^n_j\in\hat{\ms}_j$ to each $y_j^n$ according to $p^{dec_j}(\hs_j^n|y^n_j)$,
\item two distortion measures $d_j(s,\hat{s}_j)$.
\end{itemize}
\par A distortion pair $(D_1,D_2)$ is said to be achievable, if there exists a sequence of encoder-decoder such that $\e d_j(S^n,\hS_j^n)\le D_j+\epsilon_n$, $j=1,2$ and $\epsilon_n\rightarrow 0$.

We now state a new result on the above problem:
\begin{theorem}\label{thm:jscc}
A distortion pair $(D_1,D_2)$ is achievable for the lossy transmission of the source $S$ over the broadcast channel $p(y_1y_2|x)$, if there exist a pmf $p(u_{[0:2]})$, an encoding function $x(u_{[0:2]},s)$ and two decoding functions $\hs_1(u_0,u_1,y_1)$ and $\hs_2(u_0,u_2,y_2)$ such that $\e d_j(S,\hS_j)\le D_j$, $j=1,2$ and the following inequalities hold:
\begin{align}
I(U_0U_j;S)&<I(U_0U_j;Y_j)\quad ,j=1,2,\n
I(U_{[0:2]};S)+I(U_1;U_2|U_0)&<\min\left\{I(U_0;Y_1),I(U_0;Y_2)\right\}
+I(U_1;Y_1|U_0)+I(U_2;Y_2|U_0),\n
I(U_0U_1;S)+I(U_0U_2;S)&<I(U_0U_1;Y_1)+I(U_0U_2;Y_2)
-I(U_1;U_2|U_0S).\label{eq:bc}
\end{align}
\end{theorem}

\begin{remark}
The above result is related to the result of Han and Costa, \cite{miner,costa} for the lossless transmission of correlated sources over broadcast channels when $S$ is of the form $(S_1, S_2)$. In this case we can include $S_1$ in $U_1$, and $S_2$ in $U_2$. If we take the distortion function to be the Hamming distance function, the above bound reduces to a weaker version of the result of Han and Costa since instead of a vanishing probability of error we have a vanishing distortion. However the proof can be modified to recover the result of \cite{miner}.
\end{remark}

\begin{proof}
Take some arbitrary $p(s,u_{[0:2]},y_1,y_2)$ and functions $x(u_{[0:2]},s)$, $\hs_1(u_0,u_1,y_1)$ and $\hs_2(u_0,u_2,y_2)$ such that $\e d_j(S,\hS_j)< D_j$, $j=1,2$. 

\emph{Part (1) of the proof:}
We define two protocols each of which induces a joint distribution on random variables that are defined during the protocol.

\emph{Protocol A. }
Let $(S^n,U_{[0:2]}^n,Y_1^n,Y_2^n)$ be i.i.d.\ and distributed according to $p(s,u_{[0:2]},y_1,y_2)$.

\underline{Random Binning}: Consider the following random binning:
\begin{itemize}
\item To each sequence $u_0^n$ assign uniformly and independently a random bin index $f_0\in[1:2^{nR_0}]$,
\item For $j=1,2$, to each pair $(u_0^n,u_j^n)$ assign uniformly and independently a random bin index $f_j\in[1:2^{nR_j}]$,
\item We use a Slepian-Wolf decoder to recover $\hat{u}_{0,1}^n,\hat{u}_1^n$ from $(y_1^n,f_0,f_1)$, and another Slepian-Wolf decoder to recover $\hat{u}_{0,2}^n,\hat{u}_2^n$ from $(y_2^n,f_0,f_2)$.  Note that we denote the two estimates of $u_0^n$ by the two receivers with $\hat{u}_{0,1}^n$ and $\hat{u}_{0,2}^n$. The rate constraints for the success of these decoders will be imposed later, although these decoders can be conceived even when there is no guarantee of success.
\item Random variable $\hs_1^n$ is created as a function of $(\hat{u}_0^n,\hat{u}_1^n,y_1^n)$ and $\hs_2^n$ is created as function of $(\hat{u}_0^n,\hat{u}_2^n,y_2^n)$ using the two decoding functions given at the beginning.
\end{itemize}

The random pmf induced by the random binning, denoted by $P$, can be expressed as follows:
\begin{align}
P(s^n,u_{[0:2]}^n, y_1^n,y_2^n,f_{[0:2]},\hat{u}_{0,1}^n,\hat{u}_1^n,\hat{u}_{0,2}^n,\hat{u}_2^n)&=p(s^n,u_{[0:2]}^n, y_1^n,y_2^n)P(f_{[0:2]}|u_{[0:2]}^n) P^{SW}(\hat{u}_{0,1}^n,\hat{u}_1^n|y_1^n,f_0,f_1)\n&\qquad\times
P^{SW}(\hat{u}_{0,2}^n,\hat{u}_2^n|y_2^n,f_0,f_2)
\n
&=P(f_{[0:2]},s^n,u_{[0:2]}^n)p(y_{[1:2]}^n|u_{[0:2]}^n,s^n)P^{SW}(\hat{u}_{0,1}^n,\hat{u}_1^n|y_1^n,f_0,f_1)\n&\qquad\times P^{SW}(\hat{u}_{0,2}^n,\hat{u}_2^n|y_2^n,f_0,f_2)\n
&=P(f_{[0:2]},s^n)P(u_{[0:2]}^n|f_{[0:2]},s^n)p(y_{[1:2]}^n|u_{[0:2]}^n,s^n)P^{SW}(\hat{u}_{0,1}^n,\hat{u}_1^n|y_1^n,f_0,f_1)\n&\qquad\times P^{SW}(\hat{u}_{0,2}^n,\hat{u}_2^n|y_2^n,f_0,f_2).\label{eq:LBpmf0}
\end{align}
We have ignored $\hs_1^n$ and $\hs_2^n$ from the pmf at this stage since they are functions of other random variables. They will be introduced later.

\emph{Protocol B.} In this protocol we assume that the transmitter and the two receivers have access to the
shared randomness $F_{[0:2]}$ where $F_{[0:2]}$ is uniformly distributed over $[1:2^{nR_0}]\times [1:2^{nR_1}]\times [1:2^{nR_2}]$. Observe that this implies that $F_0$, $F_1$ and $F_2$ are mutually independent.
Then, the protocol proceeds as follows:
\begin{itemize}
\item The transmitter generates $U_{[0:2]}^n$ according to the conditional pmf $P(u_{[0:2]}^n|s^n,f_{[0:2]})$ of protocol A.
\item Next, $X^n$ is computed from  $(U_{[0:2]}^n, S^n)$ using $n$ copies of the function $x(u_{[0:2]},s)$ (the arbitrary function we chose at the beginning). R.v. $X^n$ is transmitted over the broadcast channel.
\item At the final stage, the receiver $j=1,2$, knowing $(y_j^n,f_{0},f_j)$ uses the Slepian-Wolf decoder\\ $P^{SW}(\hat{u}_{0,j}^n,\hat{u}_j^n|y_j^n,f_0,f_j)$ of protocol A to obtain estimates of $u_{0}^n$ and $u_j^n$. We note that while the receiver $j=1,2$ knows $f_0, f_1$ and $f_2$, it uses only $f_{0},f_j$ in its Slepian-Wolf decoder.
\item Random variable $\hs_1^n$ is created as a function of $(\hat{u}_{0,1}^n,\hat{u}_1^n,y_1^n)$ and $\hs_2^n$ is created as function of $(\hat{u}_{0,2}^n,\hat{u}_2^n,y_2^n)$ using the two decoding functions given at the beginning.
\end{itemize}
The random pmf induced by the protocol, denoted by $\hat{P}$, factors as
\begin{align}
&\hat{P}(s^n,u_{[0:2]}^n, y_1^n,y_2^n,f_{[0:2]},\hat{u}_{0,1}^n,\hat{u}_1^n,\hat{u}_{0,2}^n,\hat{u}_2^n)=\n
&p^U(f_{[0:2]})p(s^n)P(u_{[0:2]}^n|s^n,f_{[0:2]})p(y_{[1:2]}^n|u_{[0:2]}^n,s^n)P^{SW}(\hat{u}_{0,1}^n,\hat{u}_1^n|y_1^n,f_0,f_1)P^{SW}(\hat{u}_{0,2}^n,\hat{u}_2^n|y_2^n,f_0,f_2).\label{eq:LBpmf2}
\end{align}
Again we have ignored $\hs_1^n$ and $\hs_2^n$ from the pmf at this stage since they are functions of other random variables.

\emph{Part (2a) of the proof: Sufficient conditions that make the induced pmfs approximately the same}: To find the constraints that imply that the pmf $\hat{P}$ is close to the pmf $P$ in total variation distance,
we start with $P$ and make it close to $\hat{P}$ in a few steps. The first step is to observe that in protocol A, $f_0$ is a bin index of $u_0^n$, $f_1$ is a bin index of $(u_0^n,u_1^n)$ and $f_2$ is a bin index of $(u_0^n, u_2^n)$. Substituting $X_1=U_0, X_2=U_0U_1,X_3=U_0U_2, Z=S$ in Theorem \ref{thm:re} implies that $S^n$ is nearly independent of $F_{[0:2]}$ and that its pmf is close to $p(s^n)$, if
\begin{align}
R_0&<H(U_0|S),\n
R_0+R_j&<H(U_0U_j|S)\ ,j=1,2,\n
R_0+R_1+R_2&<H(U_{[0:2]}|S).\label{eq:14}
\end{align}
In other words, the above constraints imply that $P(f_{[0:2]},s^n)\apx{}p^U(f_{[0:2]})p(s^n)=\hat{P}(f_{[0:2]},s^n)$.
Equations \eqref{eq:LBpmf0} and \eqref{eq:LBpmf2} imply
\begin{align}
\hat{P}(s^n,u_{[0:2]}^n, y_1^n,y_2^n,f_{[0:2]},\hat{u}_{0,1}^n,\hat{u}_1^n,\hat{u}_{0,2}^n,\hat{u}_2^n)
&\apx{}
P(s^n,u_{[0:2]}^n, y_1^n,y_2^n,f_{[0:2]},\hat{u}_{0,1}^n,\hat{u}_1^n,\hat{u}_{0,2}^n,\hat{u}_2^n).\label{eqn:LBee1}
\end{align}

\emph{Part (2b) of the proof: Sufficient conditions that make the Slepian-Wolf decoder succeed}: The next step is {to see that when the Slepian-Wolf decoder $j, j=1,2$ of protocol A can reliably decode the  pair $(U_0^n,U_j^n)$.} Lemma \ref{le:sw} for $X_1=U_0, X_2=U_0U_j, Z=Y_j$ yields that the decoding of $U_0^nU_j^n$ is reliable if,
\begin{align}
R_0+R_j&>H(U_0U_j|Y_j),\n
R_j&>H(U_j|U_0Y_j)~~\mbox{for }j=1,2.\label{eq:13}
\end{align}
It yields
\begin{align}
P(s^n,u_{[0:2]}^n, y_1^n,y_2^n,f_{[0:2]},\hat{u}_{0,1}^n,\hat{u}_1^n,\hat{u}_{0,2}^n,\hat{u}_2^n)&\apx{}P(s^n,u_{[0:2]}^n, y_1^n,y_2^n,f_{[0:2]})\n&\qquad\times\ind\{\hat{u}_{0,1}^n=\hat{u}_{0,2}^n=u_0^n, \hat{u}_1^n=u_1^n, \hat{u}_2^n=u_2^n\}.\label{eq:LBpmf1.5}
\end{align}

Using equations \eqref{eqn:LBee1}, \eqref{eq:LBpmf1.5} and the triangle inequality we have
\begin{align}
\hat{P}(s^n,u_{[0:2]}^n, y_1^n,y_2^n,f_{[0:2]},\hat{u}_{0,1}^n,\hat{u}_1^n,\hat{u}_{0,2}^n,\hat{u}_2^n)&\apx{}P(s^n,u_{[0:2]}^n, y_1^n,y_2^n,f_{[0:2]})\n&\qquad\times\ind\{\hat{u}_{0,1}^n=\hat{u}_{0,2}^n=u_0^n, \hat{u}_1^n=u_1^n, \hat{u}_2^n=u_2^n\}.\label{eq:LBpmf1}
\end{align}

\emph{Part (3) of the proof: Eliminating the shared randomness {$F_{[0:2]}$} without disturbing the desired distortions: }
The proof of this part follows exactly the same step used in the Part (3) of the proof of Berger-Tung. Following exactly the same steps used in the Part (3) of the proof of Berger-Tung, we find $f_{[0:2]}$ such that
 \[\e_{\hat{p}_{S^n\hS_j^n|F_{[0:2]}=f_{[0:2]}}}[d(S^n,\hS_j^n)]<D_j+ (\epsilon_n+\delta_n)d_{j,max}.\]

 Specifying $(p(u_{[0:2]}^n|f_{[0:2]},s^n),x^n(u_{[0:2]}^n,s^n))$ as the encoder and $(p^{SW}(\hat{u}_{0,j},\hat{u}_j|y_j^n,f_0,f_j),\hs^n_j(\hat{u}_{0,j},\hat{u}_j,y_j^n))$
 as the decoder $j$ results in encoder and decoders obeying the desired distortions.

Finally applying FME on \eqref{eq:13} and \eqref{eq:14} gives \eqref{eq:bc}.
\end{proof}

\begin{remark}[Connection to Hybrid coding \cite{hybrid}]\label{rmkhybridcoding}
In this problem, the OSRB framework has a close relation to the hybrid coding approach. Hybrid coding \cite{hybrid} is a recent approach for establishing achievability results for the joint source-channel coding scenarios. In this approach the same code (codebook)  is used for both source coding and channel coding. Observing the source, the encoder adopts a codeword from the codebook and then generates the channel input as the symbol-by-symbol function of the codeword and the source. Similarly observing the channel output, the decoder adopts a codeword from the codebook and then generates the source estimate  as the symbol-by-symbol function of the codeword and the channel output. 

Although in OSRB proof of Theorem \ref{thm:jscc} we did not deal with codebook explicitly but one can find similarity between it and hybrid coding as follows:
\begin{itemize}
\item Conditioned on an instance $F_{[0:2]}=f_{[0:2]}$ of the shared randomness, the inputs and the outputs of encoder and decoders are limited to those sequence assigned to $f_{[0:2]}$. Thus one can interpret   these sequences as codewords of a codebook. The encoder generates the codeword $U_{[0:2]}^n$ according to $p(u_{[0:2]}^n|f_{[0:2]},s^n)$ and the decoder $j$ attempts to find $(U_0^n,U_j^n)$. So the same codebook has been used for both of the encoder and decoders. In particular,  $p(u_{[0:2]}^n|f_{[0:2]},s^n)$ can be regarded as source  encoder (compressor) and the SW decoder (conditioned on an instance of shared randomness) can be interpreted as the channel decoder. In fact, this is a general phenomenon  in OSRB framework, which is not restricted to joint source-channel coding problems.
\item As in hybrid coding, we use symbol-by-symbol function to map the codeword and the source to channel input. Also, we use symbol-by-symbol function to map the codeword and the channel output to a source estimate.
\end{itemize}
{The OSRB framework provides an alternative and  straightforward achievability proof for hybrid coding scheme. Further,} in \cite{farzin13} we proposed a hybrid coding based achievability proof using OSRB framework for the problem of channel simulation (synthesis) using another channel.
{To best of our knowledge, there is no known solution for this problem using the traditional approach based on codebook generation.}
\end{remark}
\color{black}

\subsection{Relay channel with/without secrecy}\label{sub:nnc}
Until now, we only considered one-hop networks. In this section we investigate our framework for multi-hop setting through studying wiretap relay channel. As other applications of the OSRB framework in multi-hopping setting, please see \cite{me2} and \cite{me3}. In particular, our proof for the problem of  \emph{interactive channel simulation (synthesis)} \cite{me2} is a reminiscent of two well-known strategies for relay channel, namely decode-forward and compress-forward.

 In this subsection, we prove noisy network coding inner bound \cite{NNC} for relay channel  and its extension to wiretap relay channel. Extension to multiple relays is also possible, but for simplicity we only consider one relay.

\emph{Problem definition:}
Consider the problem of secure transmission over a relay channel with a wiretapper, $p(y_{r},y,z|x,x_{r})$, where $X$ and $X_r$ are the channel inputs at transmitter and relay, respectively and $Y_r,Y,Z$ are the channel outputs at the relay, receiver and eavesdropper, respectively.  Here, we wish to securely transmit a message $m\in[1:2^{nR}]$ to the receiver $Y$ with the help of the relay, while concealing it from the eavesdropper. We again use the strong notion of secrecy as a measure for analyzing the secrecy. Formally there are,
\begin{itemize}
\item A message $M$ which is uniformly distributed,
\item A stochastic encoder at transmitter which maps the message to  a channel input $x^n$ according to $p^{enc}(x^n|m)$,
\item  A set of stochastic relay-encoding functions $p^{enc,relay}_t(x_{r,t}|y_r^{t-1},x_r^{t-1}),\ t=1,\cdots,n$  mapping the sequence $(y_r^{t-1},x_r^{t-1})$ to a channel input $x_{r,t}$ at time $t$,\footnote{
In a relay channel without an eavesdropper, the sequence $x_{r,t}$ can be taken to be a deterministic function of $y_r^{t-1}$ (relay randomization could only confuse the receivers); thus $x_{r,t}$ is implicitly related to its past sequence $x_r^{t-1}$. However in the presence of an eavesdropper, the relay might randomize to confuse the adversary and we cannot remove the dependency of $x_{r,t}$ on $x_r^{t-1}$.}
\item A decoder that assigns an estimate $\hat{m}$ of $m$ to each $y^n$.
\end{itemize}

 A  secrecy rate $R$ is said to be achievable if $\Pr\{\hat{M}\neq M\}\rightarrow 0$ and $M$ is nearly independent of the wiretapper output, $Z^n$, that is,
\[
\tv{p(m,z^n)-p^U_{\mm}(m)p(z^n)}\rightarrow 0,
\]
where, here $p(z^n)$ is the induced pmf on $Z^n$ and is not an i.i.d.\ pmf.
The secrecy capacity $C_s$ is the supremum of the set of all achievable secrecy rate.

To show the applicability of our framework in complicated networks, we prove an extension of noisy network coding inner bound for relay channel to include an eavesdropper. In fact, we again show  that adding secrecy is simple using the OSRB framework. We have adopted noisy network coding for investigation since it has a simpler analysis compared to other relaying protocols such as decode-forward or compress-forward. However, these protocols can also be studied using the OSRB framework.
\color{black}

\begin{theorem}\label{thm:WR}
The secrecy capacity of relay channel with an eavesdropper is lower bounded as
\begin{equation}
\begin{split}
C_s\ge \sup&\left\{\max\big\{ R_{\mathsf{NNC}}-R_{\mathsf{BC-Z}},\right.\\
&\left.\min\{R_{\mathsf{NNC}}-I(U;Z),R_{\mathsf{MAC-Y}}-R_{\mathsf{MAC-Z}}\}\big\}\right\},
\end{split}\label{eq:nncbound}
\end{equation}
where
\begin{align*}
R_{\mathsf{BC-Y}}&=I(U;Y\hY_r|U_r),\\
R_{{\mathsf{BC-Z}}}&=I(U;Z\hY_r|U_r),\\
R_{{\mathsf{MAC-Y}}}&=I(UU_r;Y)-I(Y_r;\hY_r|UU_rY),\\
R_{{\mathsf{MAC-Z}}}&=I(UU_r;Z)-I(Y_r;\hY_r|UU_rZ),\\
R_{\mathsf{NNC}}&=\min\{R_{{\mathsf{BC-Y}}},R_{{\mathsf{MAC-Y}}}\},
\end{align*}
 and the supremum is taken over all joint p.m.f of $(u,x,u_r,x_r,y_r,y,z,\hy_r)$ factor as $$p(u,x)p(u_r,x_r)p(y_r,y,z|x,x_r)p(\hy_r|u_r,y_r).$$
\end{theorem}
\begin{remark}
If we disable the compression part of NNC by setting $\hY_r=\phi$, the NNC strategy reduces to noise forwarding strategy and we obtain the achievable secrecy rate of \cite[Theorem 3]{h-elgamal} under strong secrecy criterion.
\end{remark}
\begin{remark}
In \cite[Corollary 3.1]{perron}, a lower bound on the secrecy capacity of deterministic networks is derived. For the special case of deterministic relay channel, we can establish this corollary using Theorem \ref{thm:WR} by setting $\hY_r=Y_r$.
\end{remark}

\begin{proof}
Without loss of generality, let $U=X$ and $U_{r}=X_{r}$.\\
Take some arbitrary $p(x,x_r,y_r,y,z,\hy_r)=p(x)p(x_r) p(y_r,y,z|x,x_r)p(\hy_r|x_r,y_r)$. 

\emph{Part (1) of the proof:}
We define two protocols each of which induces a joint distribution on random variables that are defined during the protocol. Fix an arbitrarily large integer number $B$.

\emph{Protocol A. }
Let $(X^{nB},X_r^{nB},Y_r^{nB},Y^{nB},Z^{nB},\hY_r^{nB})$ be i.i.d.\ and distributed according to $p(x,x_r,y_r,y,z,\hy_r)$. We divide these sequences to $B$ blocks and denote the sub-sequences in the block $b\in[1:B]$ by index $(b)$. Observe that the sequence of r.v.s $\left\{(X^{n}_{(b)},X_{r,(b)}^{n},Y_{r,(b)}^{n},Y^{n}_{(b)},Z^{n}_{(b)},\hY_{r,(b)}^{n})\right\}_{b=1}^B$ is mutually independent and has the same distribution over the blocks.

\underline{Random Binning}: Consider the following random binning:
\begin{itemize}
\item To each $x^{nB}=x^n_{([1:B])}=(x^n_{(1)},x^n_{(2)},\cdots,x^n_{(B)})$ assign uniformly and independently two random bin indices $m\in[1:2^{nBR}]$ and $f\in[1:2^{nB\tR}]$ (this resembles the repetition of a message in the \emph{noisy network coding scenario, because we considered only one message for all blocks}),
\item For $b=1, 2, \cdots, B-1$, to each tuple $(\hy^n_{r,(1)},x^n_{r,(2)},\hy_{r,(2)}^n,x^n_{r,(3)},\cdots,\hy^n_{r,(b)},x^n_{r,(b+1)})$ assign uniformly and independently a random bin index $f_{r,(b)}\in[1:2^{n\tR_r}]$. This bin will be used to convey some information about $(\hy^n_{r,([1:b])},x^n_{r,([2:b+1])})$ to receiver in the block $b+1$.
\item We use a Slepian-Wolf decoder to obtain an estimate $\hat{x}^{nB}$ of $x^{nB}$ from $(y^{n}_{([1:B])},f,f_{r,([1:B-1])})$.
\item Random variable $\hat{M}$ is created as a bin index assigned to $\hat{X}^{nB}$.
 \end{itemize}
The random pmf induced by the random binning, denoted by $P$, can be expressed as follows:
\begin{align}
&P(x^{nB},x_{r}^{nB},y^{nB}_r,y^{nB},z^{nB},\hy_r^{nB},m,f,f_{r,([1:B-1])},\hat{x}^{nB})
=p(x^{nB},x_{r}^{nB},y^{nB}_r,y^{nB},z^{nB},\hy^{nB}_r)P(m,f|x^{nB})\n
&\qquad\qquad\qquad\qquad\qquad~~~~~~\qquad\qquad\quad\times\left[\prod_{b=1}^{B-1}P(f_{r,(b)}|\hy^n_{r,([1:b])},x^n_{r,([2:b+1])})\right]P^{SW}(\hat{x}^{nB}|y^{n}_{[1:B]},f,f_{r,[1:B-1]})
\n
&~~=P(x^{nB},m,f)p(x_{r,(1)}^n)\left[\prod_{b=1}^{B}p(y^n_{r,(b)},y^n_{(b)},z^n_{(b)}|x^n_{(b)},x^n_{r,(b)})
p(\hy^n_{r,(b)}|y_{r,(b)}^n,x_{r,(b)}^n)p(x^n_{r,(b+1)})P(f_{r,(b)}|\hy^n_{r,([1:b])},x^n_{r,([2:b+1])})\right]
\n
&\qquad ~~~~~~~~\qquad~~~~~~~~~~~~~~~~~~~\times P^{SW}(\hat{x}^{nB}|y^{n}_{[1:B]},f,f_{r,[1:B-1]})\n
&=P(x^{nB},m,f)p(x_{r,(1)}^n)\left[\prod_{b=1}^{B}p(y_{r,(b)}^n,y^n_{(b)},z^n_{(b)}|x^n_{(b)},x^n_{r}(b))
P(\hy^n_{r,(b)},x^n_{r,(b+1)},f_{r,(b)}|\hy^n_{r,([1:b-1])},x^n_{r,([2:b])},y^n_{r,(b)})\right]\n
&\qquad ~~~~~~~~\qquad~~~~~~~~~~~~~~~~~~~\times P^{SW}(\hat{x}^{nB}|y^{n}_{[1:B]},f,f_{r,[1:B-1]})\n
&=P(m,f)P(x^{nB}|m,f)p(x^n_{r,(1)})
\left[\prod_{b=1}^{B}p(y_{r,(b)}^n,y^n_{(b)},z^n_{(b)}|x^n_{(b)},x^n_{r,(b)})
 P(f_{r,(b)}|\hy^n_{r,([1:b-1])},x^n_{r,([2:b])},y^n_{r,(b)})\right.\n&\left.~~\qquad\qquad\qquad\qquad\qquad\qquad\qquad\qquad\qquad\qquad\qquad\qquad
P(\hy^n_{r,(b)},x^n_{r,(b+1)}|f_{r,(b)},\hy^n_{r,[1:b-1])},x^n_{r,([2:b])},y^n_{r,(b)})\right]\n&\qquad ~~~~~~~~\qquad~~~~~~~~~~~~~~~~~~~\times P^{SW}(\hat{x}^{nB}|y^{n}_{[1:B]},f,f_{r,[1:B-1]})
\label{eq:SRpmf0}
\end{align}

\emph{Protocol B.} In this protocol we assume that the transmitter, the relay, the receiver and the wiretapper have access to the
shared randomness $(F,F_{r,([1:B-1])})$ where $(F,F_{r,([1:B-1])})$ is uniformly distributed over $[1:2^{n\tR}]\times [1:2^{n\tR_r}]^{B-1}$. Observe that this implies that $F$, $F_{r,(1)},F_{r,(2)},\cdots,F_{r,(B-1)}$ are mutually independent.
Then, the protocol proceeds as follows:
\begin{itemize}
\item \emph{Encoding at transmitter:}
\begin{enumerate}
\item The transmitter chooses a message $m$ uniformly distributed over $[1:2^{nBR}]$ and independently of $(F,F_{r,([1:B-1])})$.
\item The transmitter generates $x^{nB}$ according to the conditional pmf $P(x^{nB}|m,f)$ of protocol A. R.v. $X^n_{(b)}$ is transmitted over the channel in the block $b$.
\end{enumerate}
\item  \emph{Encoding at relay:}
\begin{enumerate}
\item In the \emph{first block} the relay generates an i.i.d.\ sequence $x_{r,(1)}^n$ according to the pmf $p(x_r)$ and sends it over the channel.
\item At the end of block $b\in[1:B-1]$, knowing $(f_{r,(b)},\hy^n_{r,([1:b-1])},x^n_{r,([2:b])},y^n_{r,(b)})$ the relay generates $(\hy_{r,(b)}^n,x_{r,(b+1)}^n)$ according to conditional pmf $P(\hy^n_{r,(b)},x^n_{r,(b+1)}|f_{r,(b)},\hy^n_{r,([1:b-1])},x^n_{r,([2:b])},y^n_{r,(b)})$ of protocol A. Then the relay transmits $x_{r,(b+1)}^n$ in the block $b+1$.
\end{enumerate}
\item \emph{Decoding at receiver:}
\begin{enumerate}
\item At the final stage, the receiver acquiring $(y^{nB},f,f_{r,([1:B-1])})$ uses the Slepian-Wolf decoder \\ $P^{SW}(\hat{x}^{nB}|y^{n}_{([1:B])},f,f_{r,([1:B-1])})$ of protocol A to obtain an estimate of $x^{nB}$.
\item We use the output of the SW decoder  for decoding of the messages $M$. In protocol A, we constructed $M$ as a bin index of $X^{nB}$.  
{Here upon obtaining the estimate $\hat{X}^{nB}$ of $X^{nB}$, decoder declares the bin index $\hat{M}$ assigned to $\hat{X}^{nB}$ as the estimate of the transmitted message.}
\end{enumerate}
\end{itemize}
The random pmf induced by the protocol, denoted by $\hat{P}$, factors as
\begin{align}
\hat{P}(x^{nB},x_{r}^{nB},y^{nB}_r,y^{nB},z^{nB},\hy_r^{nB},m,f,f_{r,([1:B-1])}),\hat{x}^{nB})
&
=p^U(m)p^U(f)P(x^{nB}|m,f)p(x_{r,(1)}^n)\left[\prod_{b=1}^{B}\right.\n&\quad
p^U(f_{r,(b)})p(y_{r,(b)}^n,y^n_{(b)},z^n_{(b)}|x^n_{(b)},x^n_{r,(b)})
\n&\quad
\left.P(\hy^n_{r,(b)},x^n_{r,(b+1)}|f_{r,(b)},\hy^n_{r,([1:b-1])},x^n_{r,([2:b])},y^n_{r,(b)})\right]\n&\qquad\times P^{SW}(\hat{x}^{nB}|y^{n}_{[1:B]},f,f_{r,[1:B-1]}).
\label{eq:SRpmf2}
\end{align}\normalsize
We have ignored $\hat{M}$ from the pmf at this stage since they are (random) functions of other random variables.

\emph{Part (2a) of the proof: Sufficient conditions that make the induced pmfs approximately the same}: To find the constraints that imply that the pmf $\hat{P}$ is close to the pmf $P$ in total variation distance,
we start with $P$ and make it close to $\hat{P}$ in a few steps. Comparing the relations for the pmfs $P$ and $\hat{P}$ in \eqref{eq:SRpmf0} and \eqref{eq:SRpmf2}, respectively, suggests that the conditions $P(m,f)\apx{}p^U(m)p^U(f)$ and $P(f_{r,(b)}|\hy^n_{r,([1:b-1])},x^n_{r,([1:b])},y^n_{r,(b)})\apx{}p^U(f_{r,(b)})$ (more precisely, $P(f_{r,(b)},\hy^n_{r,([1:b-1])},x^n_{r,([1:b])},y^n_{r,(b)})\apx{}p^U(f_{r,(b)})P(\hy^n_{r,([1:b-1])},x^n_{r,([1:b])},y^n_{r,(b)})$) are sufficient to approximate $P$ by $\hat{P}$. We relegate the prove of the sufficiency of these approximations to  Appendix \ref{apx:induction}. 

Since $M$ and $F$ are both random bins of $X^{nBR}$, Theorem \ref{thm:re} yields that if
\be
R+\tilde{R}<H(X),\label{eq:SR0}
\ee
then $P(m,f)\apx{}p^U(m)p^U(f)$. Also $F_{r,(b)}$ is a random bin number assigned to $(\hY^n_{r,([1:b])},X^n_{r,([2:b+1])})$. Theorem \ref{thm:re} implies that the following constraint is sufficient for the (nearly) independence of  $F_{r,(b)}$ and $(\hY^n_{r,([1:b-1])},X^n_{r,([1:b])},Y^n_{r,(b)})$,
\begin{align}
\tR_r&<
H(\hY_{r,([1:b])},X_{r,([2:b+1])}|\hY_{r,([1:b-1])},X_{r,([1:b])},Y_{r,(b)})\n
        &=H(\hY_{r,(b)}|X_{r,(b)}Y_{r,(b)})+H(X_{r,(b+1)})\n
        &=H(\hY_r|X_rY_r)+H(X_r),\label{eq:SR1}
\end{align}
where we used the independence among blocks and the fact that the pmf of r.v.'s is the same over all the blocks, {that is, $p(x_{(b)},x_{r,(b)},\hy_{r,(b)},y_{r,(b)},y_{(b)},z_{(b)})=p(x,x_{r},\hy_{r},y_{r},y,z)$. Thus we can write 
$H(\hY_{r,(b)}|X_{r,(b)}Y_{r,(b)})+H(X_{r,(b+1)})$ as $H(\hY_r|X_rY_r)+H(X_r)$. This observation will be used in the rest of the proof.}

%
\emph{Part (2b) of the proof: Sufficient conditions that make the Slepian-Wolf decoder succeed}: The next step is {to see that when the Slepian-Wolf decoder  of protocol A can reliably decode the transmitted sequence  $X^{nB}=X^n_{(1:B)}$.} Setting   $X_1=X_{(1:B)}$, $Y=Y_{(1:B)}$, $X_b=(\hY_{r,(1:b-1)},X_{r,(2:b)})$ for $b=2,\cdots,B$ in Lemma \ref{le:2000} gives the following constraints for the success of the decoder:\footnote{{Here we only write the constraints associated to the subsets of $[2:B]$ of the form $[2:i], 2\le i\le B$ and omit the others, because the unwritten constraints are redundant. It is because the random variables $X_b$ are \emph{nested} r.v.'s. Each subset of $[2:B]$ can be written as $\ms=\{ m_1,m_2,\cdots,m_k\}$ where $\{m_j\}_{j=1}^k$ is an increasing sequence. In this case $X_{\ms}=(\hY_{r,(1:m_k-1)},X_{r,(2:m_k)})=X_{[2:m_k]}$ and the corresponding constraint is implied by the constraint associated to $[2:m_k]$.}}
\begin{align}
B\tR&>H(X_{(1:B)}|\hY_{r,(1:B-1)}X_{r,(2:B)}Y_{(1:B)})\n
&=H(X_{(1)}|\hY_{r,(1)}Y_{(1)})+\sum_{b=2}^{B-1}H(X_{(b)}|\hY_{r,(b)}X_{r,(b)}Y_{(b)})+H(X_B|X_{r,B}Y_B)\n
&=H(X|\hY_rY)+(B-2)H(X|\hY_{r}X_{r}Y)+H(X|X_rY)\n
&=BH(X|\hY_{r}X_{r}Y)+C_1,\label{eq:swr1}\\
\mathsf{for} ~~b=1:B-2,\quad B\tR+b\tR&> H\left(X_{(1:B)}\hY_{r,(B-b:B-1)}X_{r,(B-b+1:B)}|Y_{(1:B)},\hY_{r,(1:B-b-1)}X_{r,(2:B-b)}\right)\n
&=H(X|\hY_rY)+(B-b-2)H(X|\hY_rX_rY)+H(X\hY_r|X_rY)\n&\qquad+(b-1)H(X\hY_rX_r|Y)+H(XX_r|Y)\n&
=(B-b)H(X|\hY_rX_rY)+bH(X\hY_rX_r|Y)+C_1,\label{eq:swr2}\\
B\tR+(B-1)\tR_r&>H(X_{(1:B)}\hY_{r,(1:B-1)}X_{r,(2:B)}|Y_{(1:B)})\n
&=H(X|\hY_rX_rY)+(B-1)H(X\hY_rX_r|Y)+C_2,\label{eq:swr3}
\end{align}
where in \eqref{eq:swr1}-\eqref{eq:swr3} where we again use the independence among blocks and the fact that the pmf of r.v.'s is the same over all the blocks. Moreover $C_1$ and $C_2$ are finite constant not depending on $B$ (formed by taking leftover terms all together as a constant). Now \eqref{eq:swr1}-\eqref{eq:swr3} yield that (for sufficiently large $B$) the following constraints are sufficient for the success of the SW decoders:
\begin{align}
\tR&>H(X|X_r\hat{Y}_rY),\label{eq:SW1}
\\
\tR+\tR_r&>H(XX_r\hat{Y}_r|Y).\label{eq:SW2}
\end{align}
Using the approximation of $P$ by $\hat{P}$ and similar argument to the one used in the previous models (for example, equations \eqref{eq:Spmf1.5}-\eqref{eq:sbc-pmf1} for wiretap broadcast channel), we get
\be
\hat{P}(m,f,f_{r,(1:b)},\hat{m},z^n_{(1:B)})\apx{}P(m,f,f_{r,(1:b)},z^n_{(1:B)})\ind\{\hat{m}=m\}.\label{eq:nnc-pmf1}
\ee
Before we consider the secrecy part of problem, we assume that there is no eavesdropper, i.e. $Z=constant$. So we deal with the relay channel. It can be easily seen that the constraints \eqref{eq:SR0}, \eqref{eq:SR1}, \eqref{eq:SW1} and \eqref{eq:SW2} imply the noisy network coding inner bound for the relay channel. In the sequel, we show how one can easily find an extension of noisy network coding inner bound for wiretap relay channel.

\emph{Part (2c) of the proof: Sufficient conditions that make the protocols secure}: We must take care of independence of $M$, and $(Z^n_{([1:B])},F, F_{r,([1:B-1])})$ consisting of the wiretapper's output and the shared randomness. We use Corollary \ref{cor:OSRB} with two different choices for $\mv$ to get two different sufficient conditions for  (nearly) mutual independence  among $M$, $F$ and $(Z^n_{([1:B])}, F_{r,([1:B-1])})$.
 In other words, we find constraints that imply
\be
P(m,f,f_{r,(1:b)},z^n_{(1:B)})\apx{}p^U(m)p^U(f)P(f_{r,(1:B)},,z^n_{(1:B)}).\label{eq:nnc-pmf2}
\ee
Using equations \eqref{eq:nnc-pmf1} and \eqref{eq:nnc-pmf2} and the third part of Lemma \ref{le:total} we have
\be
\hat{P}(m,f,f_{r,(1:b)},\hat{m},z^n_{(1:B)})\apx{}p^U(m)p^U(f)P(f_{r,(1:B)},,z^n_{(1:B)}))\ind\{\hat{m}=m\}.\label{eq:nnc-pmf3}
\ee\begin{itemize}
\item Setting $T=B$, $\mv=\emptyset$, $Z=Z_{(1:B)}$, $X_1=X_{(1:B)}$, $X_b=(\hY_{r,(1:b-1)},X_{r,(2:b)})$ for $b=2,\cdots,B$ in Corollary \ref{cor:OSRB} shows  that the following constraints imply the desired independence,
\begin{align}
B(R+\tR)&<H(X_{(1:B)}|Z_{(1:B)})=BH(X|Z),\label{eq:SS0}\\
\mathsf{for} ~~b=1:B-1,~~
B(R+\tR)+b\tR_r&<H(X_{(1:B)}\hY_{r,(1:b)}X_{r,(2:b+1)}|Z_{(1:B)})\n
&=H(X\hY_r|Z)+(b-1)H(X\hY_rX_r|Z)+H(XX_r|Z)\n&\qquad\qquad+(B-b+1)H(X|Z)\n
&=(B-b)H(X|Z)+bH(X\hY_rX_r|Z)+C_3,\label{eq:osz}
\end{align}
where $C_3$ is a finite constant number (not depending on $B$). Observe that if the following constraint and \eqref{eq:SS0} hold, then for sufficiently large  $B$ the constraint \eqref{eq:osz} is satisfied,
\begin{align}
R+\tR+\tR_r&<H(XX_r\hat{Y}_r|Z).\label{eq:SS1}
\end{align}
\item Setting $T=B$, $\mv=[2:B]$, $Z=Z_{(1:B)}$, $X_1=X_{(1:B)}$, $X_b=(\hY_{r,(1:b-1)},X_{r,(2:b)})$ for $b=2,\cdots,B$ in Corollary \ref{cor:OSRB} yields  the following constraint for having the desired independence,
\be
B(R+\tR)<H(X_{(1:B)}|\hY_{r,(1:B-1)}X_{r,(2:B)}Z_{(1:B)})=(B-2)H(X|\hY_rX_rZ)+C_4,\label{eq:osz1}
\ee
where $C_4$ is a finite constant number (not depending on $B$). Observe that if the following constraint holds, then for sufficiently large  $B$ the constraint \eqref{eq:osz1} is satisfied,
\be
R+\tR<H(X|\hY_rX_rZ),\label{eq:SS2}
\ee
\end{itemize}
\color{black}

\emph{Part (3) of the proof: Eliminating the shared randomness {$(F,F_{r,([1:B-1])})$ without disturbing the secrecy and reliability requirements}: }
 This can be done by applying the same argument  as in the part (3) of the proof of wiretap broadcast channel to \eqref{eq:nnc-pmf3} and thus omitted.

Finally, identifying $p(x^{nB}|m,f)$ as the encoder, $P(\hy^n_{r,(b)},x^n_{r,(b+1)}|f_{r,(b)},\hy^n_{r,([1:b-1])},x^n_{r,([2:b])},y^n_{r,(b)})$ as the relay encoder for block $b=2,\cdots,B$, and the Slepian-Wolf decoder as decoder results in reliable and secure encoders-decoder.

Applying FME  on \eqref{eq:SR0}, \eqref{eq:SR1}, \eqref{eq:SW1}, \eqref{eq:SW2} and \eqref{eq:SS2} results in the first term in the maximization of \eqref{eq:nncbound}. Applying FME  on \eqref{eq:SR0}, \eqref{eq:SR1}, \eqref{eq:SW1}, \eqref{eq:SW2}, \eqref{eq:SS0} and \eqref{eq:SS1} results in the second term in the maximization of \eqref{eq:nncbound}.
\end{proof}

\section{Covering and Packing: Revisited}\label{s:3}
Most of the achievability proofs in NIT are based on two primitive lemmas, namely packing lemma and covering lemma \cite{elgamal}. Thus it would be interesting to see how our probabilistic proofs relate to these lemmas. We show that Theorem 1 implies a certain form of multivariate covering (but not exactly the one mentioned in \cite{elgamal}). The discussion on packing lemma  is similar and hence omitted.
\par\emph{Multivariate covering}: We prove a version of multivariate covering that is similar to Marton coding \cite{elgamal}. Consider r.v.'s ${X_{[1:T]}Z}$. Roughly speaking, we want to prove that under certain conditions on $R_i$'s, there exists a partition of set of typical sequences of $\mx_i^n$ into $2^{nR_i}$ bins of size $2^{nR'_i}=2^{n(H(X_i)-R_i)}$ for $i=1:T$, such that if we choose any of the partitions of $\mx_1^n$, and any of the partitions of $\mx_2^n$, etc, we can find sequences $x_1^n$, $x_2^n$,..., $x_T^n$ in these partitions such that they are jointly typical with each other and with $Z^n$ with high probability, for almost all choice of partitions. The conditions imposed on the rate of the bins, $R'_i$ are given in inequality \eqref{eq:f}. This is a generalization of the mutual information terms showing up in Marton coding and match the ones reported in \cite{elgamal}.

To show this let $\typ[X_{[1:T]}Z]$ be the set of strongly typical sequences w.r.t. $p_{X_{[1:T]}Z}$. Theorem \ref{thm:re} says that if
\begin{align}\label{eq:f}\forall \ms\subseteq[1:T]: ~~\sum_{t\in\ms} R'_t>\sum_{t\in\ms}H(X_t)-H(X_{\ms}|Z),\end{align}
then $P(b_{[1:T]},z^n)\apx{}p^U(b_{[1:T]})p(z^n)$.
One can show that with high probability the number of the typical sequences assigned to \emph{each bin} $b_i\in[1:2^{nR_i}]$ is about $2^{nR'_i}$, for $i=1:T$, provided that $R_i<H(X_i)$ (for example, through the same lines as in the proof of balanced coloring lemma in \cite{ahlswede}). This fact alongside with Theorem \ref{thm:re} implies that there exists a fixed binning with the corresponding pmf $\bar{p}$ such that $\bar{p}(z^n,b_{[1:T]})\apx{}p^U(b_{[1:T]})p(z^n)$ and the number of the typical sequences assigned to each bin $b_i\in[1:2^{nR_i}]$ is about $2^{nR'_i}$, provided that \eqref{eq:f} is satisfied. Let $q(b_{[1:T]},x^n_{[1:T]},z^n)=p^U(b_{[1:T]})p(z^n)\bar{p}(x^n_{[1:T]}|b,z^n)$. Since $\bar{p}(x^n_{[1:T]},z^n)=p(x^n_{[1:T]},z^n)$, we have $\bar{p}(\typ[X_{[1:T]}Z]^c)<\epsilon_n\rightarrow 0$. Markov inequality and $q\apx{}\bar{p}$ imply that $q_{B_{[1:T]}}(\{b_{[1:T]}:q(\typ[X_{[1:T]}Z]^c|b_{[1:T]})>\sqrt{\epsilon_n}\})\rightarrow 0$. Therefore for almost all the choices of $b_{[1:T]}$, the probability of the typical set conditioned on $b_{[1:T]}$ is large, implying a non-zero intersection of the typical set and the product partition set.

\appendix

\section{Proof of Theorem \ref{thm:re}}\label{apx:osrb}

We prove a one-shot version of Theorem \ref{thm:re} via bounding the  fidelity between two pmfs over a same alphabet.
\begin{definition}
For two pmfs $p_X$ and $q_X$, the \emph{fidelity} (or Bhattacharyya coefficient) is defined as:
\be
F(p_X;q_X)=\sum_{x\in\mx}\sqrt{p_X(x)q_X(x)}.
\ee
\end{definition}
Fidelity measures the similarity between two pmfs and has wide applications in quantum information theory. We always have $0\le F(p_X;q_X)\le 1$.
The following well-known lemma gives an upper bound on the total variation distance in terms of fidelity (a similar statement holds for fidelity and trace distance of two arbitrary quantum states).
\begin{lemma}
For two pmf $p_X$ and $q_X$, we have
\[
\tv{p_X-q_X}\le\sqrt{1-F^2(p_X;q_X)}.
\]
\end{lemma}
Using Jensen's inequality for the concave function $f(x)=\sqrt{1-x^2}$ and the above lemma, we get the following upper bound on the expected total variation between two random pmfs $P_X$ and $Q_X$ via the expected fidelity.
\begin{lemma}\label{le:7}
For two random pmf $P_X$ and $Q_X$, we have
\[
\e\tv{P_X-Q_X}\le\sqrt{1-\left(\e\left[F(P_X;Q_X)\right]\right)^2}.
\]

In particular if for two sequences $P^{(n)}_{X^{(n)}}$ and $Q^{(n)}_{X^{(n)}}$ of random pmfs,  $\e\left[F(P^{(n)}_{X^{(n)}};Q^{(n)}_{X^{(n)}})\right]\rightarrow 1$, then $\e\tv{P^{(n)}_{X^{(n)}}-Q^{(n)}_{X^{(n)}}}\rightarrow 0$.
\end{lemma}
\begin{definition} A distributed random binning of correlated   sources $X_{[1:T]},{Z}$ consists of  a set of random mappings $\mb_i: \mx\rightarrow [1:\mathsf{M}_i]$, $i\in[1:T]$, in which $\mb_i$ maps each sequence of $\mx_i$ uniformly and independently to the set $[1:\mathsf{M}_i]$. We denote the random variable $\mb_t(X_t)$ by $B_t$. A random distributed  binning induces the following \emph{random pmf} on the set $\mx_{[1:T]}\times\mz\times\prod_{t=1}^T [1:\mathsf{M}_t]$,
\[
P(x_{[1:T]},z,b_{[1:T]})=p_{X_{[1:T]},Z}(x_{[1:T]},z)\prod_{t=1}^T\ind\{\mb_t(x_t)=b_t\}.
\]
\end{definition}

The following theorem provides a lower bound on the expected fidelity between the induced pmf $P(b_{[1:T]},z)$ on the r.v.'s $(B_{[1:T]},Z)$ and the desired pmf $q(b_{[1:T]},z)=p^U(b_{[1:T]})p(z)$.
\begin{theorem}\label{thm:oneshot}
The expected fidelity between the induced pmf $P(b_{[1:T]},z)$  and the desired pmf $q(b_{[1:T]},z)=p^U(b_{[1:T]})p(z)$ is bounded from below by
\be
\e F(P(b_{[1:T]},z);q(b_{[1:T]},z))\ge\e_{X_{[1:T]}Z}\sqrt{\dfrac{1}{1+\sum_{\emptyset\neq\ms\subseteq[1:T]}\mathsf{M}_{\ms}2^{-h(X_{\ms}|Z)}}},
\ee
where $\mathsf{M}_{\ms}=\prod_{v\in\ms}\mathsf{M}_v$ and the conditional information $h(x|y)$ is defined by $h(x|y):=\log \dfrac{1}{p_{X|Y}(x|y)}$.
\end{theorem}

\begin{proof}
For the sake of brevity, we use the following simplified notations. We let $\mv=[1:T]$. Also, we let $\ind\{\mb(x_{\mv})=b_{\mv}\}=\prod_{t\in\mv}\ind\{\mb(x_{t})=b_{t}\}$. Also we use $1_{\ms}$ to denote an all-one vector of length $|\ms|$.  Now consider
\begin{align}
\e F(P(b_{\mv},z);q(b_{\mv},z))&=\e\sum_{b_{\mv},z}\sqrt{\sum_{x_{\mv}}p(x_{\mv},z)\ind\{\mb(x_{\mv})=b_{\mv}\}.\frac{1}{\mathsf{M}_{\mv}}p(z)}\\
                                                 &=\e\sum_{z}\sqrt{\mathsf{M}_{\mv}\sum_{x_{\mv}}p(x_{\mv},z)\ind\{\mb(x_{\mv})=1_{\mv}\}.p(z)}\label{eq:jen1}\\
                                                 &=\e\sum_{x_{\mv},z}p(x_{\mv},z)\ind\{\mb(x_{\mv})=1_{\mv}\}\sqrt{\dfrac{\mathsf{M}_{\mv}}{\sum_{\bar{x}_{\mv}}p(\bar{x}_{\mv}|z)\ind\{\mb(\bar{x}_{\mv})=1_{\mv}\}}}\label{eq:jen2}\\
                                                 &=\sum_{x_{\mv},z}p(x_{\mv},z)\e_{\mb(x_{\mv})}\e_{\mb|\mb(x_{\mv})}\ind\{\mb(x_{\mv})=1_{\mv}\}\sqrt{\dfrac{\mathsf{M}_{\mv}}{\sum_{\bar{x}_{\mv}}p(\bar{x}_{\mv}|z)\ind\{\mb(\bar{x}_{\mv})=1_{\mv}\}}}\\
                                                &\ge\sum_{x_{\mv},z}p(x_{\mv},z)\e_{\mb(x_{\mv})}\ind\{\mb(x_{\mv})=1_{\mv}\}\sqrt{\dfrac{\mathsf{M}_{\mv}}{\e_{\mb|\mb(x_{\mv})}\sum_{\bar{x}_{\mv}}p(\bar{x}_{\mv}|z)\ind\{\mb(\bar{x}_{\mv})=1_{\mv}\}}}\label{eq:jen2}\\
                                                &\ge\sum_{x_{\mv},z}p(x_{\mv},z)\e_{\mb(x_{\mv})}\ind\{\mb(x_{\mv})=1_{\mv}\}\sqrt{\dfrac{\mathsf{M}_{\mv}}{\sum_{\ms\subseteq\mv}\mathsf{M}_{\ms}^{-1}p({x}_{\ms^c}|z)\ind\{\mb({x}_{\ms^c})=1_{\ms^c}\}}}\label{eq:jen3}\\
                                                &=\sum_{x_{\mv},z}p(x_{\mv},z)\sqrt{\dfrac{1}{\sum_{\ms\subseteq\mv}\mathsf{M}_{\ms^c}p({x}_{\ms^c}|z)}}\\
                                                &=\e_{X_{\mv}Z}\sqrt{\dfrac{1}{1+\sum_{\emptyset\neq\ms\subseteq\mv}\mathsf{M}_{\ms}2^{-h(X_{\ms}|Z)}}}
\end{align}
where \eqref{eq:jen1} is due to the symmetry and \eqref{eq:jen2} follows from the Jensen inequality for the convex function $f(x)=\dfrac{1}{\sqrt{x}}$ on $\mathbb{R}_+$. To obtain \eqref{eq:jen3} from \eqref{eq:jen2}, we partition the tuples in the set $\mx_{\mv}$ according to its difference with the tuple $x_{\mv}$. Define $\mn_{\ms}:=\{\bar{x}_{\mv}:\bar{x}_{\ms^c}=x_{\ms^c}, \forall v\in\ms:\ \bar{x}_v\neq x_v\}$, {i.e. given a subset $\ms\subset\mv$ and a sequence $x_\mv$ we define $\mn_{\ms}$ as the set of all sequences $\bar{x}_{\mv}$ whose coordinate $\bar{x}_v$ is equal to $x_v$ if and only if $v\notin\ms$}. Then $\mx_{\mv}=\cup_{\ms\subseteq\mv}\mn_{\ms}$ and for each $\bar{x}_{\mv}\in\mn_{\ms}$, we have
\bes
\e_{\mb|\mb(x_{\mv})}\ind(\mb(\bar{x}_{\mv})=1_{\mv})=\e_{\mb(\bar{x}_{\mv})|\mb(x_{\mv})}\ind(\mb(\bar{x}_{\ms})=1_{\ms},\mb(x_{\ms^c})=1_{\ms^c})=\mathsf{M}_{\ms}^{-1}\ind(\mb({x}_{\ms^c})=1_{\ms^c}),
\ees
where we have used the fact that $[\mb(\bar{x}_v):v\in\ms]$ and $\mb(x_{\mv})$ are mutually independent. Substituting this in \eqref{eq:jen2} gives,
\begin{align}
\e_{\mb|\mb(x_{\mv})}\sum_{\bar{x}_{\mv}}p(\bar{x}_{\mv}|z)\ind\{\mb(\bar{x}_{\mv})=1_{\mv}\}&=\sum_{\ms\subseteq\mv}\sum_{\bar{x}_{\mv}\in\mn_{\ms}}\e_{\mb(\bar{x}_{\mv})|\mb(x_{\mv})}p(\bar{x}_{\mv}|z)\ind\{\mb(\bar{x}_{\mv})=1_{\mv}\}\n
&=\sum_{\ms\subseteq\mv}\sum_{\bar{x}_{\mv}\in\mn_{\ms}}p(\bar{x}_{\mv}|z)\mathsf{M}_{\ms}^{-1}\ind(B({x}_{\ms^c})=1_{\ms^c})\n
&\le\sum_{\ms\subseteq\mv}\sum_{\bar{x}_{\ms}}p(\bar{x}_{\ms^c},x_{\ms}|z)\mathsf{M}_{\ms}^{-1}\ind(B({x}_{\ms^c})=1_{\ms^c})\label{eq:jen4}\\
&=\sum_{\ms\subseteq\mv}\mathsf{M}_{\ms}^{-1}p(x_{\ms}|z)\ind(B({x}_{\ms^c})=1_{\ms^c}),
\end{align}
where \eqref{eq:jen4} follows from the definition of $\mn_{\ms}$ by relaxing the constraint $(\bar{x}_v\neq x_v,v\in\ms)$ from its definition.
\end{proof}

We are now ready to prove Theorem \ref{thm:re} as a corollary to Theorem \ref{thm:oneshot}.
\begin{proof}[Proof of Theorem \ref{thm:re}]
By Lemma \ref{le:7}, it suffices to prove that $\e \left[F(P(z^n,b_{[1:T]});p(z^n)p^U(b_{[1:T]}))\right]\rightarrow 1$, as $n\rightarrow\infty$. Using Theorem \ref{thm:oneshot} for $p_{X_{[1:T]}Z}=p_{X^n_{[1:T]}Z^n}$ and $\mathsf{M}_t=2^{nR_t}$, we get
\be
\e \left[F(P(z^n,b_{[1:T]});p(z^n)p^U(b_{[1:T]}))\right]\ge\e_{X^n_{[1:T]}Z^n}\sqrt{\dfrac{1}{1+\sum_{\emptyset\neq\ms\subseteq\mv}2^{nR_{\ms}-h(X^n_{\ms}|Z^n)}}},\label{eq:apx-asym}
\ee
where $R_{\ms}=\sum_{t\in\ms}R_t$. We define the following typical set,
\[\styp:=\left\{(x_{[1:T]}^n,z^n): (x_{\ms}^n,z^n)\in\styp(\ms), \forall \ms\subseteq[1:T]\right\},\]
where $\styp(\ms)$ is defined as follows:
\begin{align}
\styp(\ms):=\left\{(x_{\ms}^n,z^n): \frac{1}{n}h(x_{\ms}^n|z^n)\ge H(X_{\ms}|Z)-\epsilon, \right\},
\end{align}
and $\epsilon$ is an arbitrary positive number. By the weak law of large number, we have $\forall\ms\subseteq[1:T]$, $\lim_{n\rightarrow\infty}p(\styp(\ms))=1$. Hence we get $\lim_{n\rightarrow\infty}p(\styp)=1$. Using this definition, we find the following lower bound on the RHS of \eqref{eq:apx-asym},
\begin{align}
\e_{X^n_{[1:T]}Z^n}\sqrt{\dfrac{1}{1+\sum_{\emptyset\neq\ms\subseteq\mv}2^{R_{\ms}-h(X^n_{\ms}|Z^n)}}}&\ge \e_{X^n_{[1:T]}Z^n}\sqrt{\dfrac{1}{1+\sum_{\emptyset\neq\ms\subseteq\mv}2^{nR_{\ms}-h(X^n_{\ms}|Z^n)}}}\ind\{(x_{[1:T]}^n,z^n)\in\styp\}\n
&\ge \e_{X^n_{[1:T]}Z^n}\sqrt{\dfrac{1}{1+\sum_{\emptyset\neq\ms\subseteq\mv}2^{n(R_{\ms}-H(X_{\ms}|Z)+\epsilon)}}}\ind\{(x_{[1:T]}^n,z^n)\in\styp\}\n
&=p(\styp)\sqrt{\dfrac{1}{1+\sum_{\emptyset\neq\ms\subseteq\mv}2^{n(R_{\ms}-H(X_{\ms}|Z)+\epsilon)}}}\n
&\rightarrow\sqrt{\dfrac{1}{1+\sum_{\emptyset\neq\ms\subseteq\mv}2^{n(R_{\ms}-H(X_{\ms}|Z)+\epsilon)}}}.\label{eq:apx-asym1}
\end{align}
Finally, if for each $\ms\subseteq[1:T]$ we have $R_{\ms}< H(X_{\ms}|Z)-\epsilon$ then \eqref{eq:apx-asym1} tends to one as $n$ goes to infinity. This concludes the proof.
\end{proof}
\begin{remark}\label{rem:re}
The above proof can be easily extended to the case of general correlated sources $p_{X^n_{[1:T]}Z^n}$ \cite{book:han}. The general result for this general correlated sources is the same as the one for i.i.d.\ sources with one exception, the average entropy should be  replace by \underline{spectral inf-entropy}. The proof is similar to above, we only replace average entropy in the definition of $\styp$ by spectral inf-entropy. In this case we again have, $\lim_{n\rightarrow\infty}p(\styp)=1$.
\end{remark}
\section{Proof of Corollary \ref{cor:OSRB}}\label{apx:osrbcor}
Without loss of generality, we can assume $\mv=\emptyset$. We prove this corollary by induction on $T$. For $T=1$ the statement of the theorem is the same as the statement of Theorem \ref{thm:re}. Assume that this corollary holds for any $k<T$. If all the constraints of Theorem \ref{thm:re} are satisfied, then the proof follows from Theorem \ref{thm:re}. Thus, suppose that the constraint $\sum_{t\in\ms}R_t<H(X_{\ms}|Z)$ does not hold for some $\ms\subseteq [1:T]$. Note that by \eqref{eq:1000} $1\notin\ms$. On the other hand for any $\mv\subseteq[2:T]-\ms$, we have $R_1+\sum_{t\in\ms}R_t+\sum_{t\in\mv}R_t<H(X_1X_{\ms}X_{\mv}|Z)$. This and $\sum_{t\in\ms}R_t>H(X_{\ms}|Z)$ yields that for any  $\mv\subseteq[2:T]-\ms$, $R_1+\sum_{t\in\mv}R_t<H(X_1X_{\mv}|X_{\ms}Z)$. By induction assumption, this implies that $B_1$, $B_{[2:T]-\ms}$ and $(X_{\ms}^n,Z^n)$ are nearly independent. More precisely, we have
\[
P(x_{\ms}^n,z^n,b_1,b_{[2:T]-\ms})\apx{}p^U(b_1)P(x_{\ms}^n,z^n,b_{[2:T]-\ms}).
\]
Since $B_{\ms}$ is a function of $X_{\ms}^n$, we can introduce it to the above approximation. We have
\[
P(x_{\ms}^n,z^n,b_1,b_{[2:T]-\ms},b_{\ms})\apx{}p^U(b_1)P(x_{\ms}^n,z^n,b_{[2:T]-\ms},b_{\ms}).
\]
Using the second item in part 1 of \ref{le:total} gives
\[
P(z^n,\underbrace{b_1,b_{[2:T]-\ms},b_{\ms}}_{b_{[1:T]}})\apx{}p^U(b_1)P(z^n,\underbrace{b_{[2:T]-\ms},b_{\ms}}_{b_{[2:T]}}),
\]
which is the desired approximation.
\section{Proof of Lemma \ref{le:0-total}}\label{apx:le-total}
The proof of the first part can be found in \cite{cuff}. Next consider the second part. To prove this, we bound above the expectation $\e_{p_X}\tv{p_{Y|X}-q_{Y|X}}$ as follows:
\begin{align}
\e_{p_X}\tv{p_{Y|X}-q_{Y|X}}&=\sum_{x}p_X(x)\left(\frac{1}{2}\sum_y \left|p_{Y|X}(y|x)-q_{Y|X}(y|x)\right|\right)\n
                                             &\le\frac{1}{2}\sum_{x,y}\left|p_X(x)p_{Y|X}(y|x)-q_{X}(x)q_{Y|X}(y|x)\right|+\frac{1}{2}\sum_{x,y}\left|q_{X}(x)q_{Y|X}(y|x)-p_X(x)q_{Y|X}(y|x)\right|\n
                                             &=\tv{p_{X}p_{Y|X}-q_{X}q_{Y|X}}+\tv{q_{X}q_{Y|X}-p_Xq_{Y|X}}\n
                                             &\stackrel{(a)}{=}\tv{p_{X}p_{Y|X}-q_{X}q_{Y|X}}+\tv{q_{X}-p_X}\n
                                             &\stackrel{(b)}{\le}2\tv{p_{X}p_{Y|X}-q_{X}q_{Y|X}}\n
                                             &\le 2\epsilon,
\end{align}
where in the steps (a) and (b) we use the first part of this lemma. Thus there exists a specified $x\in\mx$ such that $2)$ hold. To show $2')$ we use Markov's inequality
$$p_X\left(\{x\in\mx: \tv{p_{Y|X=x}-q_{Y|X=x}}>\sqrt{\epsilon}\}\right)\leq \frac{\e_{p_X}\tv{p_{Y|X}-q_{Y|X}}}{\sqrt\epsilon}\leq 2\sqrt\epsilon.$$

Finally, consider the third part of the lemma. By the triangular inequality and the first part of the lemma, we have
\begin{align}
\e\tv{P_{X}P_{Y|X}-Q_{X}Q_{Y|X}}&\le \e\tv{P_{X}P_{Y|X}-P_{X}Q_{Y|X}}+\e\tv{P_{X}Q_{Y|X}-Q_{X}Q_{Y|X}}\n
                                                     &=   \e\tv{P_{X}P_{Y|X}-P_{X}Q_{Y|X}}+\e\tv{P_{X}-Q_{X}}\n
                                                     &\le \epsilon+\delta.
\end{align}
 \section{Proof of Lemma \ref{le:distortion}}\label{apx:le-distortion}
 We have $\e_{q_{XY}}d(X,Y)=\sum_{x,y}q_{XY}(x,y)d(x,y)\le \sum_{x,y}p_{XY}(x,y)d(x,y)+\sum_{x,y}\left|q_{XY}(x,y)-p_{XY}(x,y)\right|d(x,y)\le D+d_{max}\sum_{x,y}\left|q_{XY}(x,y)-p_{XY}(x,y)\right| \le D+\epsilon d_{\max}$.
\section{Completing Proof of Theorem \ref{thm:WR}}\label{apx:induction}
In this appendix we prove that the following two approximations are sufficient to approximate the pmf in \eqref{eq:SRpmf0} by \eqref{eq:SRpmf2},
\begin{align}
P(m,f)&\apx{}p^U(m)p^U(f)\label{eq:nncapx0}\\
P(f_{r,(k)},\hy^n_{r,([1:k-1])},x^n_{r,([1:k])},y^n_{r,(k)})&\apx{}p^U(f_{r,(k)})P(\hy^n_{r,([1:k-1])},x^n_{r,([1:k])},y^n_{r,(k)}),~~ b\in[1:B-1].\label{eq:nncapx1}
\end{align}
We prove this by induction on the number of blocks. That is, we show that the following approximation holds for each $b=0,1,\cdots,B$ by induction on $b$.
\begin{align}
\hat{P}(x^{nB},x_{r,([1:b+1])}^n,y^{n}_{r,([1:b])},y^{n}_{([1:b])},&z^{n}_{([1:b])},\hy_{r,([1:b])}^{n},m,f,f_{r,([1:b])})\n&\apx{}P(x^{nB},x_{r,([1:b+1])}^n,y^{n}_{r,([1:b])},y^{n}_{([1:b])},z^{n}_{([1:b])},\hy_{r,([1:b])}^{n},m,f,f_{r,([1:b])}).\label{eq:induction}
\end{align}
It is obvious that  the case $b=B$ is the desired approximation. First, consider the base induction $b=0$. In this case, the approximation is reduced to  $P(x^{nB},x_{r,(1)}^n,m,f)\apx{}\hat{P}(x^{nB},x_{r,(1)}^n,m,f)$, which is satisfied by the assumption $P(m,f)\apx{}p^U(m)p^U(f)$ and the first part of Lemma \ref{le:total}. Now suppose that the induction assumption holds for  $b=k-1$. We prove the induction assumption for $b=k$. Consider
 \begin{align}
 \hat{P}(x^{nB},x_{r,([1:k])}^n,&y^{n}_{r,([1:k])},y^{n}_{([1:k])},z^{n}_{([1:k])},\hy_{r,([1:k-1])}^{n},m,f,f_{r,([1:k])})\n&=\hat{P}(x^{nB},x_{r,([1:k])}^n,y^{n}_{r,([1:k-1])},y^{n}_{([1:k-1])},z^{n}_{([1:k-1])},\hy_{r,([1:k-1])}^{n},m,f,f_{r,([1:k-1])}))\n
 &~~~~\qquad\times p(y_{r,(k)}^n,y^n_{(k)},z^n_{(k)}|x^n_{(k)},x^n_{r,(k)})p^U(f_{r,(k)})
\label{eq:ind0}\\
&\apx{}P(x^{nB},x_{r,([1:k])}^n,y^{n}_{r,([1:k-1])},y^{n}_{([1:k-1])},z^{n}_{([1:k-1])},\hy_{r,([1:k-1])}^{n},m,f,f_{r,([1:k-1])}))\n
 &~~~~\qquad\times p(y_{r,(k)}^n,y^n_{(k)},z^n_{(k)}|x^n_{(k)},x^n_{r,(k)})p^U(f_{r,(k)})\label{eq:ind1}\\
 &=P(x^{nB},x_{r,([1:k])}^n,y^{n}_{r,([1:k])},y^{n}_{([1:k])},z^{n}_{([1:k])},\hy_{r,([1:k-1])}^{n},m,f,f_{r,([1:k-1])}))p^U(f_{r,(k)})\label{eq:ind2}\\
&=p^U(f_{r,(k)})P(\hy^n_{r,([1:k-1])},x^n_{r,([1:k])},y^n_{r,(k)})\n&~~~{P}(x^{nB},y^{n}_{r,([1:k-1])},y^{n}_{([1:k])},z^{n}_{([1:k])},m,f,f_{r,([1:k-1])}|\hy^n_{r,([1:k-1])},x^n_{r,([1:k])},y^n_{r,(k)})\n
&\apx{}P(f_{r,(k)},\hy^n_{r,([1:k-1])},x^n_{r,([1:k])},y^n_{r,(k)})\n&~~~{P}(x^{nB},y^{n}_{r,([1:k-1])},y^{n}_{([1:k])},z^{n}_{([1:k])},m,f,f_{r,([1:k-1])}|\hy^n_{r,([1:k-1])},x^n_{r,([1:k])},y^n_{r,(k)})\label{eq:ind3}\\
&= {P}(x^{nB},x_{r,([1:k])}^n,y^{n}_{r,([1:k])},y^{n}_{([1:k])},z^{n}_{([1:k])},\hy_{r,([1:k-1])}^{n},m,f,f_{r,([1:k])}),\label{eq:ind4}
 \end{align}
where equation \eqref{eq:ind0} is due to pmf factorization \eqref{eq:SRpmf2}, equation \eqref{eq:ind1} follows from induction assumption and the first part of Lemma \ref{le:total}, equation \eqref{eq:ind2} is due to pmf factorization \eqref{eq:SRpmf0}, equation \eqref{eq:ind3} follows from the approximation \eqref{eq:nncapx0},
and equation \eqref{eq:ind4} is due to the Markov chain $$F_{r,(b)}-\left(\hY^n_{r,([1:k-1])},X^n_{r,([1:k])},Y^n_{r,(k)}\right)-\left(X^{nB},Y^{n}_{r,([1:k-1])},Y^{n}_{([1:k])},Z^{n}_{([1:k])},M,F,F_{r,([1:k-1])}\right),$$ which is satisfied by \eqref{eq:SRpmf0}. Finally, the desired approximation \eqref{eq:induction} for $b=k$ is implied by the pmf factorizations \eqref{eq:SRpmf0} and \eqref{eq:SRpmf2}, the approximation \eqref{eq:ind4} and the first part of Lemma \ref{le:total}. This completes the induction proof.

\end{document}